\newif\iffull\fulltrue

\documentclass[acmsmall]{acmart}

\setcopyright{rightsretained}
\acmPrice{}
\acmDOI{10.1145/3498719}
\acmYear{2022}
\copyrightyear{2022}
\acmSubmissionID{popl22main-p550-p}
\acmJournal{PACMPL}
\acmVolume{6}
\acmNumber{POPL}
\acmArticle{57}
\acmMonth{1}

\usepackage{booktabs}   
\usepackage{subcaption} 
\usepackage{xspace}
\usepackage{amsmath}
\usepackage{thmtools} 
\usepackage{stmaryrd}
\usepackage{xcolor}
\usepackage{cleveref}
\usepackage{wrapfig}
\usepackage{bussproofs, mathpartir} 
\usepackage{mathtools}
\usepackage{relsize} 

\input{hoare}

\newcommand{\defeq}{\vcentcolon=}
\newcommand{\defdefeq}{\vcentcolon\vcentcolon=}

\newcommand\joe[1]{\textcolor{red}{joe: #1}}
\newcommand\jh[1]{\textcolor{blue}{justin: #1}}
\newcommand\jb[1]{\textcolor{brown}{jialu: #1}}
\newcommand\mg[1]{\textcolor{cyan}{MG: #1}}

\newcommand{\hush}[1]{}
\newcommand{\quieteveryone}{%
	\let\joe\hush%
	\let\jh\hush%
	\let\jb\hush%
	\let\mg\hush%
}

\quieteveryone

\newcommand{\idxset}{M}
\newcommand{\idx}{m}
\newcommand\LOGIC{$\idxset$-BI\xspace}
\newcommand\AssertLOGIC{\text{$\idxset$-BI${}_{\text{restricted}}$}\xspace}
\newcommand{\programlogic}{\textsc{LINA}\xspace}
\newcommand{\Lang}{\textsc{pWhile}\xspace}
\newcommand{\Config}{\textsf{Config}\xspace}

\newcommand\sepid{\ensuremath{I}}
\newcommand\sepand{\mathrel{\ast}}
\newcommand\sepimp{\mathrel{-\mkern-6mu*}}
\newcommand{\bigsep}{\mathop{\Huge \mathlarger{\mathlarger{\ast}}}}

\newcommand\indand{\mathrel{\ast}}
\newcommand\indimp{\mathrel{-\mkern-6mu*}}
\newcommand{\bigind}{\mathop{\Huge \mathlarger{\mathlarger{\ast}}}}

\newcommand\negand{\mathrel{\circledast}}
\newcommand\negimp{\mathrel{-\mkern-6mu\circledast}}
\newcommand{\bigneg}{\mathop{\Huge \mathlarger{\mathlarger{\circledast}}}}

\newcommand{\PROP}{\textsf{Prop}}
\newcommand{\TRUE}{\textsf{True}}
\newcommand{\FALSE}{\textsf{False}}
\newcommand{\valuation}{\mathcal{V}}
\newcommand{\complex}{\textsf{Com}}
\newcommand{\prf}{\textsf{Prf}}

\newcommand{\PSLset}{X}
\newcommand{\PSLunit}{E_{\text{indep}}}
\newcommand{\PSLmodel}{\mathcal{\PSLset}_{\text{indep}}}
\newcommand{\PNAmodel}{\mathcal{\PSLset}_{\CPNA}}
\newcommand{\PNAunit}{E_{\CPNA}}
\newcommand{\DetModel}{\mathcal{X}_{\Mem}}
\newcommand{\DetSet}{X'}
\newcommand{\DetUnit}{E_{\text{d}}}
\newcommand{\ProbModel}{\mathcal{X}_{\mathcal{D}(\Mem)}}
\newcommand{\CombModel}{\mathcal{X}_{\textsf{comb}}}

\newcommand\Var{\ensuremath{\mathbf{Var}}}
\newcommand\Val{\ensuremath{\mathbf{Val}}}
\newcommand{\Exp}{\ensuremath{\mathbf{Exp}}} 

\newcommand\Mem{\mathbf{Mem}}
\newcommand{\DMem}[1]{\mathcal{D}(\Mem[#1])}

\newcommand{\dom}{\mathbf{dom}}
\newcommand{\p}{\mathbf{p}} 
\newcommand{\comb}{\circ} 
\newcommand{\Dist}{\mathcal{D}}

\newcommand{\CPNA}{\text{PNA}\xspace}

\newcommand{\Sem}[1]{\llbracket #1 \rrbracket}
\newcommand{\denot}[1]{\ensuremath{\llbracket #1 \rrbracket}}
\newcommand{\dbind}{\mathbf{bind}}
\newcommand{\dunit}{\mathbf{unit}}

\newcommand{\ktt}{\ensuremath{\mathit{tt}}}
\newcommand{\kff}{\ensuremath{\mathit{ff}}}

\newcommand{\FV}{\ensuremath{\text{FV}}}
\newcommand{\RV}{\ensuremath{\text{RV}}}
\newcommand{\WV}{\ensuremath{\text{WV}}}
\newcommand{\MV}{\ensuremath{\text{MV}}}

\newcommand{\Expr}{\mathcal{E}}
\newcommand{\Event}{\mathcal{EV}}

\newcommand{\DetVar}{\mathcal{DV}}

\newcommand{\RanVar}{\mathcal{RV}}

\newcommand{\Command}{\mathcal{C}}

\newcommand{\onehot}{oh}
\newcommand{\perm}{perm}
\newcommand{\Unif}[1]{\mathbf{U}_{#1}}

\newcommand{\Skip}{\mathbf{skip}}
\newcommand{\Seq}[2]{{#1} \mathrel{;} {#2}}
\newcommand{\Assn}[2]{\ensuremath{{#1} \leftarrow {#2}}}
\newcommand{\Rand}[2]{{#1} \stackrel{\raisebox{-.25ex}[.25ex]%
{\tiny $\mathdollar$}}{\raisebox{-.2ex}[.2ex]{$\leftarrow$}} {#2}}

\newcommand{\DWhile}[2]{\mathbf{while}\ #1\ \mathbf{do}\ #2}
\newcommand{\RCond}[3]{\mathbf{if}\ #1\ \mathbf{then}\ #2\ \mathbf{else}\ #3}
\newcommand{\RCondt}[2]{\mathbf{if}_R\ #1\ \mathbf{then}\ #2}
\newcommand{\Cond}[3]{\mathbf{if}\ #1\ \mathbf{then}\ #2\ \mathbf{else}\ #3}
\newcommand{\Condt}[2]{\mathbf{if}\ #1\ \mathbf{then}\ #2}
\newcommand{\RWhile}[2]{\mathbf{while}\ #1\ \mathbf{do}\ #2}

\newcommand{\apEq}[2]{{#1} \sim {#2}}
\newcommand{\apUnif}[2]{\Unif{#2}\langle #1 \rangle}
\newcommand{\apBern}[2]{\mathbf{Bern}_{#2}\langle #1\rangle}
\newcommand{\apDetm}[1]{\mathbf{Detm}\langle #1 \rangle}
\newcommand{\apIn}[1]{\apDist{#1}}
\newcommand{\apDist}[1]{\langle #1 \rangle}
\newcommand{\apEvent}[1]{{#1}={1}}
\newcommand{\apEventgeneral}[2]{{#1}={#2}}
\newcommand{\apPerm}[2]{\mathbf{Perm}_{#2} \langle #1 \rangle}
\newcommand{\apOnehot}[2]{\mathbf{OH}_{#2} \langle #1 \rangle}

\newcommand{\CCond}{\text{CC}}
\newcommand{\CheckMem}{\textsc{CheckMem}\xspace}
\newcommand{\bv}[3]{\mathbf{bv}(#1, #2, #3)}

\newcommand{\EE}{\mathbb{E}}
\newcommand{\NN}{\mathbb{N}}

\newcommand{\bigmid}{\middle\vert}

\begin{document}

\title{A Separation Logic for Negative Dependence}         



\author{Jialu Bao}
\affiliation{
  \institution{Cornell University}            
  \country{USA}                    
}

\author{Marco Gaboardi}
\affiliation{
  \institution{Boston University}            
  \country{USA}                    
}

\author{Justin Hsu}
\affiliation{
  \institution{Cornell University}            
  \country{USA}                    
}

\author{Joseph Tassarotti}
\affiliation{
  \institution{Boston College}            
  \country{USA}                    
}

\begin{abstract}
  Formal reasoning about hashing-based probabilistic
  data structures often requires reasoning about random
  variables where when one variable gets larger (such as
  the number of elements hashed into one bucket), the others
  tend to be smaller (like the number of elements hashed into the
  other buckets).
		This is an example of \emph{negative dependence}, a
  generalization of probabilistic independence
  that has recently found interesting
  applications in algorithm design and machine learning.
  Despite the usefulness of negative dependence for the analyses of
  probabilistic data structures, existing verification methods cannot
  establish this property for randomized programs.

		To fill this gap, we design \programlogic, a probabilistic separation logic
		for reasoning about negative dependence. Following recent works on
		probabilistic separation logic using \emph{separating conjunction} to reason
		about the probabilistic independence of random variables, we use separating
		conjunction to reason about negative dependence.  Our assertion logic
		features two separating conjunctions, one for independence and one for
		negative dependence. We generalize the logic of bunched
		implications (BI) to support multiple separating conjunctions, and provide a sound
		and complete proof system. Notably, the semantics for separating conjunction
		relies on a \emph{non-deterministic}, rather than partial, operation for
		combining resources. By drawing on closure properties for negative
		dependence, our program logic supports a Frame-like rule for negative
		dependence and \emph{monotone} operations. We demonstrate how \programlogic
		can verify probabilistic properties of hash-based data structures and
		balls-into-bins processes.
\end{abstract}

\ccsdesc[500]{Theory of computation~Separation logic}

\keywords{Probabilistic programs, separation logic, negative dependence}  

\maketitle

\section{Introduction}
\label{sec:intro}

Hashing plays a fundamental role in many probabilistic data structures, from
basic hash tables to more sophisticated schemes such as Bloom filters. In
these applications, a \emph{hash function} $h$ maps a \emph{universe} of
possible values, typically large, to a set of \emph{buckets}, typically small.
Hash-based data structures satisfy a variety
of probabilistic guarantees. For instance, we may
be interested in the \emph{false positive rate}: the
probability that a data structure mistakenly identifies an element as
being stored in a collection, when it was not inserted. We
may also be interested in load measures, such as the probability
that a bucket in the data structure overflows. A typical way to
analyze these quantities is to treat random hash functions as
\emph{balls-into-bins} processes. For example, hashing $N$ unique elements into
$B$ bins can be modeled as throwing $N$ balls into $B$ bins, where each bin is
drawn uniformly at random.

While this modeling is convenient, one complication is that the counts
of the elements in the different buckets are not probabilistically
independent: one bin containing many elements makes it more
likely that other bins contain few elements. The lack of
independence makes it difficult to reason about
multiple bins, for instance bounding the number of empty bins. Moreover, many
common tools for analyzing probabilistic processes, like concentration
bounds, usually require independence. This subtlety has also been a
source of problems in pen-and-paper analyses of probabilistic
data structures. For instance, the standard analysis of the Bloom filter
bounds the number of occupied bins in order to bound the false
positive rate. The original version of this
analysis presented by~\citet{DBLP:journals/cacm/Bloom70}, and also repeated in many papers,
assumes that the bin counts are independent. However,
\citet{DBLP:journals/ipl/BoseGKMMMST08} pointed out that this assumption is
incorrect, and in fact the claimed upper-bound on the false-positive rate is
actually a \emph{lower} bound. Proving a correct bound on the false-positive
rate required a substantially more complicated argument; recently,
\citet{DBLP:conf/cav/GopinathanS20} mechanized a correct, but complex proof in
Coq.

We aim to develop a simpler method to formally reason about hash-based
data structures and balls-into-bins processes, drawing on a key concept in
probability theory: \emph{negative dependence}.

\paragraph*{Towards a simpler analysis: negative dependence.}
To study balls-into-bins processes and other phenomena, researchers in
probability theory have developed a theory of \emph{negative
dependence}~\citep{pemantle:negdep}. Intuitively, variables are \emph{negatively
dependent} if when one is larger, then the others tend to be smaller. The counts
of the bins in the balls-into-bins process is a motivating example of negative
dependence.

While there are multiple incomparable definitions of negative dependence,
\citet{joag1983negative} proposed a notion called \emph{negative association}
(NA) that has many good probabilistic properties. For instance, the bins' counts
in the balls-into-bins process satisfies NA, and NA's closure properties enable
simple, calculation-free proofs of NA. More intriguingly, as
\citet{dubhashi-ranjan} identified, sums of NA variables satisfy some
concentration bounds that usually assume probabilistic independence, including
the widely-used Chernoff bounds.

\paragraph*{Our goal: formal reasoning about negative dependence.}
From a verification perspective, the closure properties suggest a compositional
method for proving NA in probabilistic programs. In this work, we develop a
separation logic for negative dependence, building on a separation logic for
probabilistic programs called PSL~\citep{PSL}. Like all separation logics, PSL
is a program logic where assertions are drawn from the logic of bunched
implications (BI)~\citep{OhP99}, a substructural logic. In PSL, the separating
conjunction $\sepand$ states that two sets of variables are probabilistically
independent, a common and useful property when analyzing probabilistic programs.

We aim to extend the assertions of PSL so that they can describe both
independence and negative dependence. There are three main difficulties:
\begin{itemize}
  \item To support reasoning about negative dependence, the assertion logic
    needs to be extended with a second separating conjunction that is weaker
    than the separating conjunction of PSL. It is easy to extend the syntax of
    formulas, but the extended logic should also enjoy good metatheoretical
    properties like BI does, including a sound and complete proof system.
  \item The standard resource semantics of BI~\citep{PymMono}, based on partial
    commutative monoids (PCMs), is not expressive enough to model negative
    association because two variables with given marginal distributions can be
    negatively associated in more than one way.
  \item Defining the semantics of separating conjunction to capture NA is
    surprisingly challenging. Straightforward definitions fail to satisfy
    expected properties, like associativity of separating conjunction.
\end{itemize}
Beyond the assertions, it is also unclear how to integrate negative association
with the proof rules of PSL. In particular, to view negative association as a
kind of separation, our program logic should have an analogue of the Frame rule
for NA.

\paragraph*{Contributions and outline.}
In this paper, we offer the following contributions.
\begin{itemize}
  \item A novel logic \LOGIC that extends BI with multiple separating
    conjunctions, related by a pre-order. Following \citet{docherty:thesis},
    models of \LOGIC allow two states to be combined into a single state in more
    than one way (\Cref{sec:mbi}). We develop a proof system for \LOGIC, and use
    Docherty's duality-theoretic approach to prove soundness and completeness.
  \item A probabilistic model of \LOGIC that can capture both the independence and
    negative association (\Cref{sec:model}). There are two interesting aspects
    of our model:
    \begin{itemize}
      \item We crucially use the ``non-deterministic'' combination of resources
        allowed by Docherty's semantics of BI. While this semantics was
        originally used to simplify the metatheory of BI, our model shows that
        the added flexibility can enable new applications of the logic.
      \item Our model relies on a novel notion called \CPNA that is more
        expressive than \citet{joag1983negative}'s NA. The generalization is
        needed to satisfy the conditions for an \LOGIC model. Moreover, the
        closure properties and useful consequences of NA continue to hold for
        our generalization.
    \end{itemize}
  \item A program logic, \programlogic (\textbf{L}ogic of \textbf{I}ndependence
    and \textbf{N}egative \textbf{A}ssociation), extending PSL with
    \LOGIC-assertions and a new negative-association Frame rule
    (\Cref{sec:psl}). Being a conservative extension of PSL, the proof rules of
    PSL remain valid in \programlogic. We demonstrate our program logic by
    proving negative association and related properties on several case studies
    (\Cref{sec:ex}). For example, using NA, it is possible to give a
    significantly simpler verification of the false positive rate of the Bloom
    filter. Another example---an analysis of a repeated balls-into-bins process
    motivated by distributed computing---involves a loop with a probabilistic
    guard, and requires reasoning about conditional distributions.
\end{itemize}
We discuss related work in \Cref{sec:rw}, and conclude in \Cref{sec:conc}.

\section{Overview and Key Idea}
\label{sec:overview}

In this section, we introduce negative association as a tool for analyzing
hashing-based algorithms. We use Bloom filters, a hash-based data structure, as
a motivating example. After sketching a standard proof applying negative
association to Bloom filters, we will show how the same analysis can be
formalized in \programlogic.

\subsection{Background on negative association}
Negative association is a property of a set of random variables, which
intuitively says that when some variables are larger, we expect the others to be
smaller. It is formalized as follows:

\newcommand{\Ex}[1]{\mathbb{E}[#1]}

\begin{definition}[Negative Association (NA)] \label{def:na}
  Let $X_1, \dots X_n$ be random variables. The set $\{X_i\}_i$ is
  \emph{negatively associated (NA)} if for every pair of subsets $I, J \subseteq
  \{1, \dots, n\}$ such that $I \cap J = \emptyset$, and every pair of both
  monotone or both antitone functions\footnote{%
    In the following, we will consistently use monotone to mean monotonically
  non-decreasing and antitone to mean monotonically non-increasing.}
  $f : \mathbb{R}^{|I|} \rightarrow \mathbb{R}$ and $g :\mathbb{R}^{|J|} \rightarrow \mathbb{R}$,
  where $f, g$ is either lower bounded or upper bounded, we have:
  \[ \Ex{f(X_i, i\in I) \cdot g(X_j, j\in J)} \leq \Ex{f(X_i, i \in I)} \cdot \Ex{g(X_j, j\in J)} \]
\end{definition}

We can view NA as generalizing \emph{independence}: a set of independent random
variables is NA because equality holds. NA also strengthens \emph{negative
covariance}, a simpler notion of negative dependence that requires $\Ex{
\prod_{i \in [n]} X_i} \leq  \prod_{i \in [n]} \Ex{X_i} $.

The survey paper by \citet{dubhashi-ranjan} explains several properties of NA
random variables useful for algorithm analysis. First, some standard theorems
about sums of independent random variables apply more generally to sums of NA
random variables. In particular, the widely-used \emph{Chernoff bound}, which
intuitively says that the sum of independent random variables is close to the
expected value of the sum with high probability, holds also for NA variables. In
addition, NA is preserved by some common operations on random variables. Thus,
we can easily prove that a set of random variables satisfies NA if they are
generated by applying NA-preserving operations to a few basic, building-block
random variables:

\begin{theorem}
  \label{thm:buildingblock}
  The random variables $\{X_i\}_i$ in the following cases are negatively
  associated:
  \begin{enumerate}
    \item Let $\{X_1, \dots, X_n\}$ be Bernoulli random variables such that
      $\sum X_i = 1$.
    \item Let $X_i$ be the $i$-th entry in the vector $X$, where $X$ is a
      uniformly random permutation of a finite, nonempty set $A$.
    \item Let $\{X_1, \dots, X_n\}$ be independent random variables.
  \end{enumerate}
\end{theorem}

In particular, the first case of this theorem implies that if we draw a
length-$n$ \emph{one-hot vector}, i.e., a vector that has one entry being one
and all remaining entries being zero, uniformly at random, then the entries of
the vector satisfy negative association.

The following theorem states two key closure properties of NA random variables.
\begin{theorem}
  \label{thm:closure}
  The set $S$ of random variables in the following cases are negatively
  associated:
  \begin{enumerate}
  \item Let $T$ be negatively associated, and let $S$ be a non-empty subset of $T$.
  \item Let $T$ and $U$ be two sets of negatively associated random variables such that every $X \in T$ and $Y \in U$ is independent of each other. Let $S = T \cup U$.
  \item Let $\{X_1, \dots X_n\}$ be negatively-associated, and
			$I_1, \dots, I_m$ be a partition of the set $\{1, \dots, n\}$.
			For each $1 \leq j \leq m$, let $f_j : \mathbb{R}^{|I_j|} \rightarrow \mathbb{R}$ be monotone.
			Let $S = \{ f_1(X_k, k \in I_1), \dots, f_m(X_k, k \in I_m) \}$.
  \end{enumerate}
\end{theorem}
The first case shows that negative association is preserved if we discard random
variables, while the second case allows us to join two independent sets of
negatively associated random variables to form a larger negatively associated
set. Finally, the third case guarantees that negative association is preserved
under applying monotone maps on disjoint subsets of variables.

\subsection{Example: Bloom filters}

We demonstrate how NA and its closure properties can be used to analyze Bloom
filters. A Bloom filter is a space-efficient probabilistic data structure for
storing a set of items from a universe $U$. An $N$-bit Bloom filter consists of
a length-$N$ array $bloom$ holding zero-one entries. We assume there is a
family $\mathcal{A}$ of hash functions mapping $U$ to $\{0, \dots, N-1\}$ such
that for any $x \in U$ and any bucket $k$, $\mathbb{P}_{f \in \mathcal{A}}(f(x)
= k) = 1/N $. Let $l_1, \dots, l_H$ be a collection of hash functions drawn from
$\mathcal{A}$. We assume the hash functions are independent, meaning the
collection of variables $\{l_i(x) \mid x \in U, i \in \{1, \dots, H\}\}$ are
mutually independent. To add an item $x \in U$ to the filter, we compute
$l_1(x), \dots, l_H(x)$ to get $H$ positions in the bit array and then set the
bits at each of these positions to $1$. To check if an item $y$ is in the
filter, we check whether the bits at positions $l_1(y), \dots, l_H(y)$ in
$bloom$ are all $1$. If they are, the item is said to be in the filter, but if
any is $0$, then the item is not in the filter. This membership test may suffer
from \emph{false positives}, i.e., it may show that an item $y$ is in the filter
even when $y$ was never added to the filter. This can happen because with hash
collisions, other items added to the Bloom filter could set all the bits at
locations $l_1(y), \dots, l_H(y)$ to 1. A basic quantity of interest is the
\emph{false positive rate}: the probability that a Bloom filter reports a false
positive.

\newcommand{\fpevent}{\textsf{FP}}

Our goal is to bound an $N$-bit Bloom filter's false positive rate after $M$
distinct items are added, assuming that it uses $H$ independent hash functions.
Here, we briefly sketch a standard proof where negative association plays a key
role. Let $x$ be some data item not in the set, and let $\fpevent$ be the event
the Bloom filter returns a false positive on $x$. We split the analysis of the
probability of $\fpevent$ into two steps.

For the first step, we \emph{condition} on $\rho$, the fraction of bits in
$bloom$ that are set to 1 after all items have been added. With $\rho$ fixed to
some value, for each hash function $l_{\alpha}$ the probability that
$bloom[l_{\alpha}(x)] = 1$ is $\rho$.  For $x$ to be a false positive, we must have
$bloom[l_{\alpha}(x)] = 1$ for all $\alpha$.  Since $l_1, \dots, l_H$ are
independent hash functions, this occurs with probability $\rho^H$.

The next step is to show that with high probability, $\rho$ lies within a narrow
range around its expected value. Before showing how to prove this, let us first
see why such a bound is useful. Let $\mu$ be the expected value of $\rho$.
Suppose we have that for some small $\epsilon$ and $\delta$ that $\Pr[\rho \leq \mu
+ \epsilon] \geq 1 - \delta$.  Then, by the law of total probability we have:
\begin{align*}
\Pr[\fpevent]
&\leq \Pr[\fpevent \mid \rho \leq \mu + \epsilon] + \Pr[\rho > \mu + \epsilon] \\
&\leq (\mu + \epsilon)^H + \delta
\end{align*}
where the second line follows from the calculation of the probability
of $\fpevent$ when conditioned on $\rho$.

Now, we turn to the question of how to obtain a bound of the form
$\Pr[\rho \leq \mu + \epsilon] < 1 - \delta$.
As mentioned in \Cref{sec:intro}, a common (incorrect) analysis of $\rho$ assumes that
the entries of $bloom$ are independent, and then applies a Chernoff bound. However,
the entries in $bloom$, $\{bloom[\beta]\}_\beta$, are \emph{not} independent---what we
can actually prove is that are \emph{negatively associated}, which fortunately
still allows us to apply the Chernoff bound, as stated later in \Cref{thm:chernoff}.

\begin{wrapfigure}[13]{r}{0.3\textwidth}
  \begin{minipage}[c]{0.3\textwidth}
    \vspace{-3ex}
    \[
      \begin{array}{l}
        \Assn{bloom}{zero(N)}; \\
        \Assn{m}{0}; \\
        \DWhile{m < M}{} \\
        \quad \Assn{h}{0} \\
        \quad \DWhile{h < H}{} \\
        \quad\quad \Rand{bin}{oh([N])}; \\
        \quad\quad \Assn{upd}{bloom \mathop{||}  bin}; \\
        \quad\quad \Assn{bloom}{upd}; \\
        \quad\quad \Assn{h}{h+1}; \\
        \quad \Assn{m}{m + 1}
      \end{array}
    \]
  \end{minipage}
  \caption{Example: Bloom filter}
  \label{fig:ex:bloom-overview}
\end{wrapfigure}

To see that the $\{bloom[\beta]\}_\beta$ are NA, consider the program in
\Cref{fig:ex:bloom-overview}, which models the process of adding $M$ distinct
items to the Bloom filter. Because the $M$ items are distinct, we model the
hash functions as independently, randomly sampling hash values for each item as
they are added, a standard model used in the analysis of hashing data
structures~\citep{MitzenmacherUpfal}. That is, we encode the hashing step as
sampling a one-hot vector with the command $oh$ and storing it in the variable
$bin$, where the hot bit of the vector $bin$ represents the selected position.
To set the corresponding position in the filter to 1, we update $bloom$ to be $bloom
\mathop{||} bin$, the bitwise-or of the current array and the sampled one-hot
array. To show that $\{ bloom[\beta] \}_\beta$ are NA, we can reason using the
closure properties. Initially, $bloom$ is set to $zero(N)$. Any set of constant
random variables is independent, and hence negatively associated by
\Cref{thm:closure}. Next, when an item is added, the $bin$ array is NA by
\Cref{thm:buildingblock}. Because $bin$ is sampled independently of $bloom$, the set
$\{bloom[\beta]\}_\beta \cup \{bin[\beta]\}_\beta$ is NA. The bitwise-or operation
$\mathop{||}$ is monotone, so again by \Cref{thm:closure}, the array $upd$ is
negatively associated, thereby showing that $bloom$ is NA at the end of each
loop iteration.

\subsection{Representing negative association with separating conjunction}

Now that we have seen some properties of negative association and how they can
be used to analyze the Bloom filter, we give a high-level explanation of how
these ideas are formalized in \programlogic, a novel program logic that is a
core contribution of our work. As mentioned in \Cref{sec:intro}, an earlier
separation logic PSL has a separating conjunction $\sepand$ which is interpreted
as probabilistic independence.  That is, a program state satisfies $P \sepand Q$
if its randomized program variables can be split so that one subset satisfies
$P$, another satisfies $Q$, \emph{and the distributions of these two sets are
independent}. \programlogic augments PSL with a weaker separating
conjunction $\negand$ modeling negative association. The precise
definition of negative association needs to be modified to form a
proper model of bunched implications, but for now, one can informally think of $P
\negand Q$ as meaning the random variables can be split into two sets negatively
associated with each other that satisfy $P$ and $Q$ respectively.

The proof rules of \programlogic can be used to prove NA by applying closure
properties to building-block NA-distributions, much like in our proof sketch
above for the Bloom filter. For example, we can derive a rule that captures NA
of entries in a one-hot distribution:
\begin{mathpar}
			\inferrule*[]{}
			{ \vdash \hoare{\top}{\Rand{x}{oh(n)}}{\bigneg_{\beta = 0}^n\apDist{x[\beta]}}} , %
\end{mathpar}
where the assertion $\apDist{y}$ means that the program variable $y$ is
distributed according to some unspecified probability distribution, and
$\bigneg$ is an iterated version of the $\negand$ separating conjunction. Thus,
the post-condition here says that all of the entries of the $x$ vector are
negatively associated.

Meanwhile, since NA is closed under monotone maps, we obtain a form of
separation logic's frame rule for $\negand$:
\begin{mathpar}
  \inferrule*[]
  {
    \vdash \hoare{\phi}{c}{\apEq{y}{f(X)}} \\
    \textrm{$f$ monotone} \\
    \textrm{(side conditions omitted)}
  }
  { \vdash \hoare{ \phi \negand \eta}{c}{\apIn{y}  \negand \eta} },
\end{mathpar}
where $X$ is a set of variables contained in any program states satisfying
$\phi$, $f$ is a monotone function mapping $X$ to a variable $y$, and $\eta$ is
an assertion on some other random variables that are negatively associated with
those satisfying $\phi$.
(We describe a complete version of this rule with all side conditions later, in
\Cref{sec:programlogic}.) Using this rule, we can show monotone vector operations
like $\mathop{||}$ in the Bloom filter example preserve negative associativity.
For example, we can derive:
\begin{mathpar}
  \inferrule*[]{}
  { \vdash \hoare{ \bigneg_{\beta = 0}^n \apDist{x[\beta]} \negand \bigneg_{\gamma = 0}^n \apDist{y[\gamma]}}{\Assn{z}{x \mathop{||} y}}{\bigneg_{\beta = 0}^n \apDist{z[\beta]}} },
\end{mathpar}
which says that if the union of entries in $x$ and entries in $y$ satisfy NA,
then entries in $z = x \mathop{||} y$ also satisfy NA.

We now sketch how to formalize the proof that $\{bloom[\beta]\}_\beta$ in the
Bloom filter are NA; we defer the rest of the proof of this example to
\Cref{sec:ex}.  The basic idea is to establish $\bigneg_{\beta = 0}^N
{\apDist{bloom[\beta]}}$ as a loop invariant. When an item is added in the loop,
we combine the frame rule with the one-hot sampling vector rule to get that the
$bin$ vector is negatively associated, thus showing: $\bigneg_{\beta = 0}^N
{\apDist{bloom[\beta]}} \negand \bigneg_{\gamma=0}^N {\apDist{bin[\gamma]}}$.
Applying the rule for $\mathop{||}$ above, we obtain that $upd$ is negatively
associated, $\bigneg_{\beta = 0}^N \apDist{upd[\beta]}$. At that point $upd$ is
assigned to $bloom$, restoring the loop invariant.

\section{The logic \LOGIC}
\label{sec:mbi}

Having seen the role of negative association in analyzing randomized algorithms
and how its properties correspond to rules in \programlogic, we now show how
negative association can be interpreted by separating conjunction. As a first
step, we extend the logic of bunched implications (BI), the assertion logic
underlying separation logic, to support multiple forms of separating conjunction
simultaneously, related by a pre-order. Our motivation to design this logic is
to reason about independence and negative association in one logic and capture
that independence implies negative association, but the logic is more general
and accommodates other potentially interesting models.

\subsection{The syntax and proof rules}
Let $\mathcal{AP}$ be a set of atomic propositions, and $(\idxset, \leq)$ be a finite pre-order.
The formula in the logic of $\idxset$-bunched implications ($\idxset$-BI) has the following grammar:
\begin{align*}
	P, Q &::= p \in \mathcal{AP}
	\mid \top
	\mid \sepid_{\idx \in \idxset}
	\mid \bot
	\mid P \land Q
	\mid P \lor Q
	\mid P \rightarrow Q
	\mid P \sepand_{\idx \in \idxset} Q
	\mid P \sepimp_{\idx \in \idxset} Q.
\end{align*}
\LOGIC associates each element of $\idxset$ with a separating conjunction
$\sepand_\idx$, a corresponding multiplicative identity $\sepid_{\idx}$ and a
separating implication $\sepimp_{\idx}$.  The proof system for
\idxset-BI is based on the proof system for BI, with indexed copy of rules for
each separation, and in addition has $\sepand$-\textsc{Weakening} rules. We
present the full Hilbert-style proof system in
\iffull
\Cref{sec:app:logic};
\else
the extended version;
\fi
most of the rules are the
same as in the proof system for BI. Here, we only comment on the new rules.

The $\sepand$-\textsc{Weakening} rule says that if $m_1 \leq m_2$, then the
assertion $P \sepand_{\idx_{1}} Q$ implies $P \sepand_{\idx_{2}} Q$.
\begin{mathpar}
   \inferrule* [right= $\sepand$-\textsc{Weakening}]
 	{\idx_1 \leq \idx_2}
 	{P \sepand_{\idx_1} Q \vdash P \sepand_{\idx_2} Q}
\end{mathpar}
We can derive analogous weakening rules for separating implications and
multiplicative identities, in the reverse direction.

\begin{restatable}{lemma}{InclusionOthers}
	The following rules are derivable in $\idxset$-BI:
\begin{mathpar}
        \inferrule* [right= $\sepimp$-\textsc{Weakening}]
 	{\idx_1 \leq \idx_2}
 	{P \sepimp_{\idx_2} Q \vdash P \sepimp_{\idx_1} Q}

        \inferrule* [right= \textsc{UnitWeakening}]
 	{\idx_1 \leq \idx_2}
 	{\sepid_{\idx_2} \vdash \sepid_{\idx_1}}
\end{mathpar}
\end{restatable}

\subsection{Semantics}
As is standard with bunched logics~\citep{POhY04}, we give a Kripke style
semantics to \LOGIC. We will define a structure called \LOGIC~{\em frame}, and
then define \LOGIC~{\em models} and the satisfaction rules on \LOGIC models.

An \LOGIC frame is a collection of BI frames satisfying some \emph{frame
conditions}. While BI frames are often presented as partial, pre-ordered
commutative monoids over states, we need the more general presentation due
to~\citet{docherty:thesis}, where the binary operation returns a \emph{set} of
states, instead of at most one state. Such binary operations can be
deterministic (returning a set of at most one element) or non-deterministic. The
admission of non-deterministic models was originally motivated by the
metatheory; somewhat surprisingly, it is also a crucial ingredient in defining
the negative association model we will see in~\Cref{sec:model}.
\begin{definition}[BI Frame] \label{def:bi-frame}
  A \emph{(Down-Closed) BI frame} is a structure $\mathcal{X} = (X, \sqsubseteq,
  \oplus, E)$ such that $\sqsubseteq$ is a pre-order on the set of states $X$,
		$\oplus\colon X^2 \rightarrow \mathcal{P}(X)$ is a binary
  operation,  and $E \subseteq X$, satisfying following frame conditions (with outermost universal
  quantification omitted for readability):
	\[\begin{array}{ll}
		\text{(Down-Closed)} &  z \in x \oplus y \land x' \sqsubseteq x \land y' \sqsubseteq y \rightarrow
		\exists z' (z' \sqsubseteq z \land z' \in x' \oplus y'); \\
		\text{(Commutativity)} & z \in x \oplus y \rightarrow z \in y \oplus x; \\
		\text{(Associativity)} & w \in t \oplus z \land t \in x \oplus y \rightarrow
		\exists s(s \in y \oplus z \land w \in x \oplus s); \\
		\text{(Unit Existence)} & \exists e \in E(x \in e \oplus x); \\
		\text{(Unit Coherence)} & e \in E \land x \in y \oplus e \rightarrow y \sqsubseteq x; \\
		\text{(Unit Closure)} & e \in E \land e \sqsubseteq e' \rightarrow e' \in E.
	\end{array} \]
\end{definition}
Since all frames we consider in this paper will be Down-Closed, we often
abbreviate Down-Closed BI frame as just BI frame.
%
%
The (Commutativity), (Associativity) and (Unit Existence) conditions capture the
properties of commutative monoids, and the (Down-Closed) condition ensures that
the binary operation is coherent with the pre-order. Properties (Associativity)
and (Commutativity) are generalizations of the usual algebraic properties to
accommodate the non-determinism. The three Unit frame conditions
ensure that the set $E$ behaves like a set of units, and satisfies various
closure properties under the binary operation and the pre-order.

We can now define \LOGIC frame to be a collection of BI frames sharing the same
set of states and pre-order, with ordered binary operations:
\begin{definition}[$\idxset$-BI Frame] \label{def:mbi-frame}
  An $\idxset$-BI frame is a structure $\mathcal{X} = (X, \sqsubseteq, \oplus_{\idx \in \idxset}, E_{\idx})$ such that
  for each $\idx$, $(X, \sqsubseteq, \oplus_\idx, E_\idx)$ is a BI frame, and
		there is a preorder $\leq$ on $M$ satisfying:
	\[\begin{array}{ll}
		\text{(Operation Inclusion)} & \idx_1 \leq \idx_2 \rightarrow x \oplus_{\idx_1} y \subseteq x \oplus_{\idx_2} y . \\
	\end{array} \]
\end{definition}
The (Operation Inclusion) condition together with the frame conditions of BI also imply an
inclusion on unit sets:

\begin{lemma}
  Let $\mathcal{X}$ be an $\idxset$-BI frame. If $\idx_1 \leq \idx_2$ then $E_{\idx_2} \subseteq E_{\idx_1}$.
\end{lemma}

\begin{proof}
	Let $e_2 \in E_{\idx_2}$. By (Unit Existence), there exists $e_1 \in
	E_{\idx_1}$ such that $e_2 \in e_1 \oplus_{\idx_1} e_2$. By (Operation
	Inclusion), $e_2 \in e_1 \oplus_{\idx_2} e_2$, so (Unit Coherence) implies
	that $e_1 \sqsubseteq e_2$, and then (Unit Closure) implies $e_2 \in
	E_{\idx_1}$.  So $E_{\idx_2} \subseteq E_{\idx_1}$.
\end{proof}

To obtain a BI model over a given BI frame, we must provide a \emph{valuation},
which defines which atomic propositions hold at each states in the frame. For
the soundness of the proof system, it is important that the valuation is
\emph{persistent}: any formula true at a state remains true at any larger state.
Formally, we define \LOGIC models as follows.
\begin{definition}[Valuation and model]
A persistent
valuation is a map $\valuation : \mathcal{AP} \rightarrow \mathcal{P}(X)$ such
that, for all $P \in \mathcal{AP}$, if $x \in \valuation(P)$ and $x \sqsubseteq y$ then $y \in
\valuation(P)$.
An \LOGIC model  $(\mathcal{X}, \valuation)$ is an \LOGIC frame
$\mathcal{X} = (X, \sqsubseteq, \oplus_{\idx},
E_{\idx})$ associated with a persistent valuation $\valuation$ on it.
\end{definition}

Next, we define which \LOGIC formula are true at a state in a \LOGIC model.

\begin{definition}
On model $(\mathcal{X}, \valuation)$,
we define the satisfaction relation $\models_\valuation$ between states in $\mathcal{X}$ and \LOGIC
formula: for $x \in \mathcal{X}$
\begin{center}
\begin{tabular}{@{}l@{ }l@{} l@{}}
  $x \models_\valuation p$ &\quad iff \quad & $x \in \valuation(p)$ \\
  $x \models_\valuation \top$ &\quad iff \quad & $\TRUE$ \\
  $x \models_\valuation \sepid_{\idx}$ &\quad iff \quad & $x \in E_{\idx} $ \\
  $x \models_\valuation \bot$ &\quad iff \quad & \FALSE $ $ \\
  $x \models_\valuation P \land Q$ &\quad iff \quad & $x \models_\valuation P $ and $ x \models_\valuation Q$ \\
  $x \models_\valuation P \lor Q$ &\quad iff \quad & $x \models_\valuation P$ or $ x \models_\valuation Q$ \\
  $x \models_\valuation P \rightarrow Q$ &\quad iff \quad & for all $y$ such that $x \sqsubseteq y $, if $y \models_\valuation P$ then $y \models_\valuation Q $ \\
  $x \models_\valuation P \sepand_{\idx} Q$ &\quad iff \quad & there exist $x'$, $y$, and $z$ with $x' \sqsubseteq x$ and $x' \in y \oplus_{\idx} z$ such that $y \models_\valuation P$ and $z \models_\valuation Q$\\
  $x \models_\valuation P \sepimp_{\idx} Q$ &\quad iff \quad & for all $y$ and $z$ such that $z \in x \oplus_{\idx} y$, if $y \models _\valuation P$ then $z \models_\valuation Q$ \\
\end{tabular}
\end{center}
\end{definition}

We say that a formula $P$ is \emph{valid in the model} $(\mathcal{X}, \valuation)$, written as $\mathcal{X} \models_{\valuation} P$, iff $x \models_{\valuation} P$
for all $x \in \mathcal{X}$.
We also say that $P$ is \emph{valid} if and only if $P$ is valid in all models and write that as $\models P$. Finally, we write $P \models Q$ if and only if for any model $(\mathcal{X}, \valuation)$,
$\mathcal{X} \models_{\valuation} P$ implies $\mathcal{X} \models_{\valuation} Q$.

We prove the following theorem in
\iffull
\Cref{sec:app:completeness}.
\else
the extended version.
\fi
\begin{theorem}
  Let $P$ and $Q$ be any two \LOGIC formulas. Then $P \models Q$ iff $P \vdash Q$.
\end{theorem}
The reverse direction of (soundness) is straightforward by induction on the
proof derivation, but the forward direction (completeness) is less obvious; we
use the duality-theoretic framework proposed by \citet{docherty:thesis} to
establish this theorem.

\subsection{Potential models}
Our design of \LOGIC is mainly motivated by our intended model of negative
association and probabilistic independence, which we will see in the next
section, but the logic \LOGIC is quite flexible and we can see other natural
models. Here we outline three models, inspired by the heap model of separation
logic~\citep{Reynolds01}.

\paragraph*{Hierarchical heaps.}
In the heap model of separation logic, heaps are partial maps $h : \NN
\rightharpoonup \Val$ from integer addresses to values, and heap combination
$\circ_{heap}$ is a partial binary operation that takes the union if the two
heaps have disjoint domains, and is not defined otherwise. In many systems,
memory addresses are partitioned into larger units, for instance \emph{pages}.
We can define another partial binary operation $\circ_{page}$ that takes the
union if the two heaps have disjoint domains \emph{and} are defined on disjoint
pages. Then, $h_1 \circ_{page} h_2 \subseteq h_1 \circ_{heap} h_2$, so we can
build a model of \LOGIC on the two-point pre-order. The two separating
conjunctions $\sepand_{heap}$ and $\sepand_{page}$ then describe heap and page
separation respectively, which could be useful for reasoning about which memory
accesses may require a page-table lookup.

\paragraph*{Strong separation logic.}
Heaps in separation logic can store values, but also addresses of other
locations. In the standard heap model, two separate heaps must have disjoint
domains but may store common addresses, i.e., they may hold dangling pointers to
the same locations. Searching for a separation logic with better decidability
properties, \citet{DBLP:conf/esop/PagelZ21} proposed a notion of \emph{strong
separation logic}, where two strongly-separated heaps can only hold common
addresses that are already stored in stack variables. The resulting form of
separation can be modeled by a separating conjunction $\sepand_{st}$, and the
standard (weak) form of separation can be modeled by a separating conjunction
$\sepand_{wk}$. Since strong separation implies weak separation, we can again
build a model of \LOGIC supporting both conjunctions on the two-point pre-order.

\begin{wrapfigure}[8]{r}{0.4\textwidth}
\tikzset{every picture/.style={line width=0.75pt}} 
\vspace{-2ex}
\centering
\begin{tikzpicture}[x=0.75pt,y=0.75pt,yscale=-1,xscale=1]

\draw   (234.31,60.41) -- (255.99,75.9) -- (257.66,73.56) -- (265.45,85.72) -- (251.43,82.29) -- (253.1,79.95) -- (231.42,64.46) -- cycle ;
\draw   (153.31,94.41) -- (174.99,109.9) -- (176.66,107.56) -- (184.45,119.72) -- (170.43,116.29) -- (172.1,113.95) -- (150.42,98.46) -- cycle ;
\draw   (182,61.41) -- (160.32,76.9) -- (158.65,74.56) -- (150.86,86.72) -- (164.89,83.29) -- (163.22,80.95) -- (184.9,65.46) -- cycle ;
\draw   (263.01,96.41) -- (241.32,111.9) -- (239.65,109.56) -- (231.86,121.72) -- (245.89,118.29) -- (244.22,115.95) -- (265.9,100.46) -- cycle ;

\draw (195,47) node [anchor=north west][inner sep=0.75pt]   [align=left] {$\displaystyle \sepand _{(h, p)}$};
\draw (272,85) node [anchor=north west][inner sep=0.75pt]   [align=left] {$\displaystyle \sepand _{h}$};
\draw (130,81.72) node [anchor=north west][inner sep=0.75pt]   [align=left] {$\displaystyle \sepand _{p}$};
\draw (200,120) node [anchor=north west][inner sep=0.75pt]   [align=left] {$\displaystyle \sepand _{-}$};

\end{tikzpicture}

\caption{The pre-order for $\sepand_{\idx}$ }
\label{fig:pre-order}
\end{wrapfigure}

\paragraph*{Tagged memory.}
In some security-focused architectures, pointers contain an address as well as a
tag, indicating capabilities that may be performed with that piece of memory.
To reason about these machines, we can consider a resource frame where states
are pairs of $(h, p)$, where $h$ is a heap and $p$ is a permission (say shared
access, or exclusive access). We can then consider four kinds of separation
taking all combinations of heaps aliasing/non-aliasing, and permissions
compatible/incompatible. The result is a \LOGIC frame, with the lattice of
separating conjunctions depicted in \Cref{fig:pre-order}. Assertions in models
on this frame can reason about all four kinds of separation, where
$\sepand_{-}$ degenerates to the standard conjunction $\land$.

\section{A model of negative association and independence}
\label{sec:model}

In this section, we will present a \LOGIC model for reasoning about both
probabilistic independence and negative association. \citet{PSL} proposed a BI
model that captures probabilistic independence, developed a program logic (PSL)
to reason about independence in probabilistic programs.  We will construct a BI
model for negative association and combine it with the PSL model to obtain a
$\mathbf{2}$-BI model for both probabilistic independence and negative
association, where $\mathbf{2} = \{ 1, 2\}$ is the two-point set with pre-order
$1 \leq 2$.

One may wonder if it is a simple exercise to replace the independence semantics
of the separating conjunction in PSL by a semantics that capture negative
association, but there are technical challenges. As we will show
in~\Cref{sec:failed-attempts}, two intuitive BI model definitions fail to
satisfy all frame conditions. To overcome the difficulties,
in~\Cref{sec:our-model} we define a new notion of negative association that can
express negative dependence of various strengths, and then define a model based
on our new notion.
\iffull
\else
All omitted proofs can be found in the extended version.
\fi


\subsection{Preliminaries and the PSL model}
We need to introduce some notation to define the models.

First, we represent program states as memories.
Let the set of all program variables be $\Var$, and the set of all possible
values be $\Val$.  For any finite set of variables $S \subseteq \Var$, a \emph{memory}
on $S$ is a map $S \to \Val$, and $\Mem[S]$ denotes the set of memories on $S$.
For disjoint sets of variables $S,T \subseteq \Var$, and $m_1 \in \Mem[S]$,
$m_2 \in \Mem[T]$, we define $m_1 \bowtie m_2$ to be the union of $m_1, m_2$.

Now we will introduce probabilistic memories.  For a real-valued function $f$,
we say that $x$ is in the support of $f$ if $f(x) \neq 0$.  A countable
distribution on a set $X$ is a countable support function $\mu: X \to [0,1]$
such that $\sum_{x \in X} \mu(x) = 1$.  Let $\Dist(X)$ denote the set of
countable distributions on $X$.  A \emph{probabilistic memory} on variables $S$
is a distribution over $\Mem[S]$, so the set of probabilistic memories on $S$ is
$\Dist(\Mem[S])$.

Next, we will need some constructions on distributions.
A family of special distributions in $\DMem{S}$ is \emph{Dirac distributions}: for any $x \in \Mem[S]$, the Dirac distribution $\delta(x)$ puts all the weight on $x$, that is, for any $y \in \Mem[S]$,
$\delta(x)(y) = 1$ if $x = y$, and $\delta(x)(y) = 0$ otherwise.
%
For any sets of variables $S \subseteq S'$, we define the \emph{projection} map $\pi_{S', S}$ to
map a distribution $\mu$ on $\Mem[S']$ to a distribution on $\Mem[S]$: for any $x \in
\Mem[S]$,
\begin{align*}
	\pi_{S', S} \mu(x) \defeq \sum_{x' \in \Mem[S'] \text{ and } \p_{S}(x') = x } \mu(x'),
\end{align*}
where $\p_{S}(x')$ is $x'$ restricted on $S$.
Often $S'$ is clear, so we just write $\pi_{S}$ for $\pi_{S', S}$.
Then, we can formally define independence of variables in a distribution:
\begin{definition}[Independence]
	For any $\mu \in \DMem{S}$, and disjoint $T_1, T_2 \subseteq S$, we say
	$T_1, T_2$ are independent in $\mu$ if for any $x \in \Mem[T_1 \cup T_2]$,
	\begin{align*}
		\pi_{T_1 \cup T_2} \mu (x) = \pi_{T_1} \mu (\p_{T_1}(x)) \cdot  \pi_{T_2} \mu (\p_{T_2}(x)).
	\end{align*}
\end{definition}
We then define the \emph{independent product}
$\otimes$ as: for any $\mu_1 \in \DMem{S}$, $\mu_2 \in \DMem{T}$,
\begin{align*}
	\mu_1 \otimes \mu_2
	=
\begin{cases}
	\emptyset & \text{if $S, T$ not disjoint} \\
	\{\mu  \mid \text{for any $x \in \Mem[S \cup T]$, } \mu(x) = \mu_1(\p_S(x)) \cdot \mu_2(\p_T(x))\}  & \text{if $S, T$ disjoint}
\end{cases}
\end{align*}
For any distribution $\mu \in \DMem{S}$, we call $S$ the \emph{domain} of $\mu$,
denoted $\dom(\mu)$.  By construction, if $\mu \in \mu_1 \otimes \mu_2$, then
$\dom(\mu_1)$ and $\dom(\mu_2)$ are independent in $\mu$, and $\mu$ is the
unique element in $\mu_1 \otimes \mu_2$. Simple calculations also show that if
$\mu \in \mu_1 \otimes \mu_2$, then $\pi_S \mu = \mu_1$, $\pi_T \mu = \mu_2$.

We can then present the Independence frame from~\citet{PSL} as the following BI
frame. For simplicity, we restrict its states to probabilistic memories for
now.\footnote{%
  Technically we take a slightly different notion of BI frames that is more
  suitable for our purposes. \citet{PSL} presents BI frames with partial,
  pre-ordered commutative monoids, which require a unique unit for all states.
  But we can encode their frame as our BI frame by taking its partial operation
as an operation that returns sets of size at most one and defining the unit set
$E$ to include their unique unit and be closed under $\sqsubseteq$.}

\begin{definition} \label{def:indep-BI}
	Let $\PSLset = \cup_{S \subseteq \Var} \DMem{S}$.
	Say $\mu \sqsubseteq \mu'$ iff $\dom(\mu) \subseteq \dom(\mu')$ and $\pi_{\dom(\mu)} \mu' = \mu$.
	Let $\PSLunit = \PSLset$.
	We call $\PSLmodel = (\PSLset, \sqsubseteq, \otimes, \PSLunit)$ the \emph{Independence structure}.
\end{definition}

This Independence structure $\PSLmodel$ is a BI frame.

\subsection{Initial attempts at a NA model}
\label{sec:failed-attempts}
Our goal is to design a BI model $\PNAmodel$ that can capture negative
association and can be combined with $\PSLmodel$. To be compatible with
$\PSLmodel$, we let $\PNAmodel$ have the same set of states and the same
pre-order as $\PSLmodel$.  The important remaining piece of the puzzle is the
binary operation $\oplus$, which must satisfy the frame conditions.

One first attempt is to let  $\mu_1 \oplus \mu_2$ return the set of
distributions that agree with $\mu_1, \mu_2$, and satisfy \emph{strong NA}---we
say $\mu$ satisfies strong NA if $\dom(\mu)$ satisfies NA.

\begin{definition}(Attempt 1: Strong NA model)
	\label{def:strongNA}
	Let $\PSLset = \cup_{S \subseteq \Var} \DMem{S}$.
	For $\mu, \mu' \in \PSLset$, say $\mu \sqsubseteq \mu'$ iff $\dom(\mu) \subseteq \dom(\mu')$ and $\pi_{\dom(\mu)} \mu' = \mu$.
Let $E_s = \PSLset$.
Define $\oplus_s : \PSLset \times \PSLset \to \mathcal{P}(\PSLset)$:
\[
  \mu_1 \oplus_s \mu_2 =
  \{ \mu \in \DMem{S \cup T} \mid \mu \text{ satisfies strong NA}, \pi_S \mu = \mu_1, \pi_T \mu = \mu_2, S \cap T = \emptyset \} .
\]
We call $\mathcal{\PSLset}_s = (\PSLset, \sqsubseteq, \oplus_s, E_s)$ the \emph{strong NA structure}.
\end{definition}

Unfortunately, the strong NA structure fails to have the (Unit Existence)
property: if $\mu$ does not satisfy strong NA, then there exists no $\mu'$ that
marginalizes to $\mu$ and satisfies strong NA, and thus no $e$ such that $\mu
\in e \oplus_s \mu$. The failure of this property implies that whether or not
two states can be combined depends on properties of the single states in
isolation (e.g., whether a distribution satisfies strong NA), and not just on
how the two states relate to each other; this is hard to justify if we are to
read $\oplus$ as describing which pairs of states can be safely combined.

Looking for a different way of capturing NA, we can take inspiration from the
$\PSLmodel$. There, $\mu_1 \otimes \mu_2$ returns a distribution that agrees with
$\mu_1, \mu_2$ and on which $\dom(\mu_1)$ are independent from $\dom(\mu_2)$.
Thus, we can try letting $\mu_1 \oplus \mu_2$ return distributions that agree
with $\mu_1, \mu_2$ where any variable $x$ in $\dom(\mu_1)$ must be negatively
associated with any variable $y$ in $\dom(\mu_2)$, but variables within
$\dom(\mu_1)$ and variables within $\dom(\mu_2)$ need not be negatively
associated. We call this notion \emph{weak NA}.
\begin{definition}[Weak NA] \label{def:weakna}
	Let $S \subseteq \Var$ be a set of variables,
  and let $A, B$ be two disjoint subsets of $S$.  A
	distribution $\mu \in \DMem{S}$ satisfies \emph{$(A,B)$-NA} if
	for every pair of both monotone or both antitone functions $f : \Mem[A] \to
	\mathbb{R}$, $g : \Mem[B] \to \mathbb{R}$, where we take the point-wise orders on $\Mem[A]$ and $\Mem[B]$,
	such that $f, g$ is either lower bounded or upper bounded, we have
	\[
		\mathbb{E}_{m \sim \mu} [ f(\p_A m) \cdot g(\p_B m) ]
		\leq \mathbb{E}_{m \sim \mu} [ f(\p_A m) ] \cdot \mathbb{E}_{m \sim \mu} [ g(\p_B m) ] .
	\]
\end{definition}

By definition, being $(A, B)$-NA for all disjoint $A, B
\subseteq S$ is equivalent to strong NA on $S$. Now, we can try defining another
model based on weak NA.
\begin{definition}(Attempt 2: Weak NA model)
	\label{def:weakNA}
	Let $\PSLset = \cup_{S \subseteq \Var} \DMem{S}$.
	For $\mu, \mu' \in \PSLset$, $\mu \sqsubseteq \mu'$ iff $\dom(\mu) \subseteq \dom(\mu')$ and $\pi_{\dom(\mu)} \mu' = \mu$. Let $E_w = \PSLset$.
Define $\oplus_w : \PSLset \times \PSLset \to \mathcal{P}(\PSLset)$:
\[
\mu_1 \oplus_w \mu_2 =
  \{ \mu \in \DMem{S \cup T} \mid \mu \text{ satisfies } (S, T)\text{-NA}, \pi_S \mu = \mu_1, \pi_T \mu = \mu_2, S \cap T = \emptyset \} .
\]
We call $\mathcal{\PSLset}_w = (\PSLset, \sqsubseteq, \oplus_w, E_w)$ the \emph{weak NA structure}.
\end{definition}

This weak NA structure satisfies most BI frame conditions, except that
(Associativity) is unclear.  In short, the definition of $\oplus_w$ and
(Associativity) requires that: if $w$ satisfies $(R \cup S, T)$-NA and $(R,
S)$-NA, then $w$ also satisfies $(S, T)$-NA and $(R, S \cup T)$-NA. Now
$w$ satisfies $(S, T)$-NA by projection closure, but it is unclear whether
$w$ must satisfy $(R, S \cup T)$-NA; we leave this question as an interesting
open problem.
Failing to satisfy (Associativity) would lead to a logic where separating
conjunction is not associative, and significantly more difficult to use. Since
it is unknown whether $\mathcal{\PSLset}_w$ is a BI frame, we will define
another structure to capture negative association.

\subsection{Our NA model}
\label{sec:our-model}
Facing the problems with the strong NA structure and the weak NA structures, we
will define a BI model for negative association based on a new notion of
negative association called \emph{partition negative association} (\CPNA). This
notion \emph{interpolates} weak NA and strong NA, in the following sense: $\{A,
B\}$-\CPNA is equivalent to $(A,B)$-NA for disjoint $A, B \subseteq S$, and $\{
\{s\} \mid s \in S\}$-\CPNA is equivalent to strong NA for distributions in
$\DMem{S}$.
\begin{definition}[Partition Negative Association]
	\label{def:PNA}
  We say a partition $\mathcal{S}'$ \emph{coarsens} a partition $\mathcal{S}$ if
  $\cup \mathcal{S} = \cup \mathcal{S}'$ and for any $s' \in \mathcal{S}'$, $s' = \cup \mathcal{R}$ for some $\mathcal{R} \subseteq \mathcal{S}$.

	A distribution $\mu$ is $\mathcal{S}$-\CPNA if and only if
	for any $\mathcal{T}$ that coarsens $\mathcal{S}$,
	for any family of \emph{non-negative} monotone functions
	(or family of \emph{non-negative} antitone functions),
	$\{f_{A} : \Mem[A] \to \mathbb{R}^+\}_{A \in \mathcal{T}}$,\footnote{%
    We restrict the family of functions to be non-negative: prior work
    like~\citet{joag1983negative} has assumed non-negativity when working with
    notions of NA on partitions; furthermore, without that requirement, for
    partitions with odd number of components, \CPNA would be equivalent to
  independence, a strange property.}
	where we take the point-wise order on $\Mem[A]$ for each $A \in \mathcal{T}$,
	we have
	\[
		\mathbb{E}_{m \sim \mu} \left[ \prod_{A \in \mathcal{T}} f_{A}(\p_A m) \right]
		\leq \prod_{A \in \mathcal{T}} \mathbb{E}_{m \sim \mu} [ f_{A}(\p_A m) ].
	\]
\end{definition}

We can use \CPNA to prove NA:
\begin{theorem}
	\label{theorem:PNAeqNA}
	Given a set of variables $S$, $S$ satisfies NA in $\mu$
	iff $\mu$ satisfies $\mathcal{S}$-\CPNA for any
$\mathcal{S}$ partitioning $S$ iff $\mu$  satisfies $\{ \{s\} \mid s \in S\}$-\CPNA.\footnote{
	Technically, we slightly modify \citet{dubhashi-ranjan}'s NA when defining it in~\Cref{def:na} by
	in addition assuming that $f, g$ are bounded from one side.
	We add that condition
		to have a cleaner version of this theorem and~\Cref{theorem:starcapturesNA}.
	All our other results and properties we state about NA in~\Cref{sec:intro} hold with or without this condition. }
\end{theorem}

We require \CPNA to be closed under coarsening, which helps us to prove the
next structure we define is a BI frame.

\begin{definition}
Let $\PNAmodel = (\PSLset, \sqsubseteq, \oplus, \PNAunit)$,
where $\PSLset = \PNAunit = \cup_{S \subseteq \Var} \DMem{S}$.
For $\mu, \mu' \in \PSLset$, say $\mu \sqsubseteq \mu'$ iff $\dom(\mu) \subseteq \dom(\mu')$ and $\pi_{\dom(\mu)} \mu' = \mu$.
Define the operation $\oplus : \PSLset \times \PSLset \to \mathcal{P}(\PSLset)$:
\[
	\begin{aligned}
    \mu_1 \oplus \mu_2 = \{\mu \in \DMem{S \cup T} \mid{} &\pi_{S} \mu = \mu_1,  \pi_{T} \mu = \mu_2, \\
																																		&\text{$\mu$ is $(\mathcal{S} \cup \mathcal{T})$-\CPNA for any partition $\mathcal{S}, \mathcal{T}$ such that } \\
																																		&\text{$\mu_1$ is $\mathcal{S}$-\CPNA, and $\mu_2$ is $\mathcal{T}$-\CPNA, and $(\cup \mathcal{S}) \cap (\cup \mathcal{T}) = \emptyset$} .\}
		\end{aligned}
	\]
\end{definition}
This definition of $\oplus$ interpolates $\oplus_w$ and $\oplus_s$, in the
following sense.
\begin{restatable}{theorem}{PNAoplusinterpolate}
\label{theorem:PNAoplusinterpolate}
For any two states $\mu_1, \mu_2 \in \PSLset$,
	$ \mu_1 \oplus_s \mu_2 \subseteq \mu_1 \oplus \mu_2 \subseteq \mu_1 \oplus_w \mu_2$.
\end{restatable}

The first inclusion is because $\mu$
satisfying strong NA implies $\mu$ is $\mathcal{R}$-\CPNA for any partition
$\mathcal{R}$ on $\dom(\mu)$. The second inclusion is because $\mu_1 \in
\DMem{S}$ satisfies $\{S\}$-\CPNA and $\mu_2 \in \DMem{T}$ satisfies
$\{T\}$-\CPNA trivially, which implies any $\mu \in \mu_1 \oplus \mu_2$ would
satisfy $(S, T)$-NA.

Note that $\oplus$ is non-deterministic, and not just partial.
\begin{theorem} \label{thm:PNAoplus:nondet}
  There are distributions $\mu_1, \mu_2$ such that $|\mu_1 \oplus \mu_2| \geq
  2$.
\end{theorem}
\begin{proof}
  Let $\mu_1 \in \DMem{\{ x \}}$ and $\mu_2 \in \DMem{\{ y \}}$ be uniform
  distribution over memories over 0/1 variables $x$, $y$. Then the
  independent product $\mu_{\otimes} \in \mu_1 \otimes \mu_2$ is in $\mu_1 \oplus
  \mu_2$, because the projections to $x$ and to $y$ are $\mu_1$ and $\mu_2$
  respectively, and $\mu_{\otimes}$ satisfies \CPNA since independence
		implies \CPNA (we will see this shortly in \Cref{lemma:indep_closure}). But the
  one-hot uniform distribution $\mu_{oh}$ over variables $x$ and $y$, i.e.,
  $\mu_{oh}([x \mapsto 1, y \mapsto 0]) = \mu_{oh}([x \mapsto 0, y \mapsto 1]) =
  1/2$, is also in $\mu_1 \oplus \mu_2$,
  since again the projections match $\mu_1$ and $\mu_2$ and the one-hot
  distribution satisfies NA, and hence \CPNA. Since $\mu_{oh} \neq
  \mu_{\oplus}$, we are done.
\end{proof}

Thus, we can build a BI frame on probabilistic memories, crucially using a
non-deterministic combination operation on states~\citep{docherty:thesis}.
\begin{restatable}{theorem}{PNAisBI}
	\label{theorem:PNAisBI}
	The structure $\PNAmodel = (\PSLset, \sqsubseteq, \oplus, \PNAunit)$ is a Down-Closed BI frame.
	\end{restatable}
	\iffull
  See the full proof in~\Cref{proof:PNAisBI}.
	\fi
	For the frame conditions
  where the previous attempts failed, (Unit Existence) holds by letting the unit
  $e$ to always be the trivial distribution on the empty set, and
  (Associativity) can be proved using the facts that \CPNA is closed under coarsening
  and coarsening commute with projections.  We call
  $\PNAmodel$ the \CPNA model.

  Now that we know the \CPNA frame is a BI frame and captures NA, we want to
  combine it with the PSL frame to construct a \LOGIC frame.  To combine them,
  we need to show that for any $\mu_1, \mu_2 \in \PSLset$,
	\begin{align*}
		 \mu_1 \otimes \mu_2 \subseteq \mu_1 \oplus \mu_2.
	\end{align*}
The inclusion is implied by the following theorem:

\begin{restatable}[Independence implies \CPNA]{theorem}{IndepClosure}
		\label{lemma:indep_closure}
   Let $S, T \subseteq \Var$ be two disjoint sets of variables.
			Suppose $\mu_1 \in \DMem{S}$, $\mu_2 \in \DMem{T}$.
  If $\mu_1$ satisfies $\mathcal{S}$-\CPNA and
$\mu_2$ satisfies $\mathcal{T}$-\CPNA, then any $\mu \in \mu_S \otimes \mu_T$
satisfies $\mathcal{S} \cup \mathcal{T}$-\CPNA.
\end{restatable}

This theorem generalizes the independence closure for NA from~\Cref{thm:closure}.
Its proof, however, is more involved because \CPNA
is more expressive and is closed under coarsening.
\iffull
(See the proof in~\Cref{proof:IndepClosure}.)
\fi
%
%

Thus, we can combine $\PSLmodel$ and $\PNAmodel$ into a $\mathbf{2}$-BI model.
\begin{theorem}
	\label{theorem:XisMBI}
	Let $\mathbf{2} = \{1, 2\}$ with pre-order $1 \leq 2$.
	Let $\oplus_1 = \otimes$, $E_1 = \PSLunit$, $\oplus_2 = \oplus$, $E_2 = \PNAunit$.
	The structure $\ProbModel = (\PSLset, \sqsubseteq, \oplus_1, E_1, \oplus_2, E_2)$ is a $\mathbf{2}$-BI model.
\end{theorem}

Thus $\ProbModel$ is a $\mathbf{2}$-BI frame on probabilistic memories.

\subsection{Combining with deterministic memory}

While we can model the program states of probabilistic programs as
probabilistic memories, some variables might only get deterministic
assignments. It is useful to know whether a variable is deterministic; for
instance, a deterministic variable is automatically independent of other
variables. To keep track of deterministic variables, we want a $\mathbf{2}$-BI
frame whose states distinguish deterministic memories and probabilistic
memories. We will construct it using a general approach for composing \LOGIC
models. In particular, we will compose $\ProbModel$ with a $\mathbf{2}$-BI
frame on deterministic memories.

We can define the product of two \LOGIC frames if they share the same pre-order for indexing, $\idxset$.
	\begin{definition}
		Let $\idxset$ be a pre-order.
		Given two \LOGIC frames,
 $\mathcal{X}_1 = (X_1, \sqsubseteq_1, \oplus_{(1, \idx \in \idxset)}, E_{(1, m  \in \idxset)})$ and
	$\mathcal{X}_2 = (X_2, \sqsubseteq_2, \oplus_{(2, \idx \in \idxset)}, E_{(2, \idx \in \idxset)})$.
	The product frame, $\mathcal{X} = \mathcal{X}_1 \times \mathcal{X}_2 = (X, \sqsubseteq, \oplus_{\idx \in \idxset}, E_{\idx \in \idxset})$ is defined as
\begin{itemize}
	\item $X = X_1 \times X_2$;
	\item $(x_1, x_2) \sqsubseteq (x_1', x_2')$ if and only if $x_1 \sqsubseteq_1 x_1'$ and $x_2 \sqsubseteq_2 x_2'$;
	\item For $\idx \in \idxset$, $(x_1, x_2) \oplus_{\idx} (x_1', x_2') = \{
    (y_1, y_2) \mid y_1 \in x_1 \oplus_{1, \idx} x_1' \land y_2 \in x_2 \oplus_{2, \idx} x_2' \}$;
	\item $E_{\idx} = E_{1, \idx} \times E_{2, \idx}$.
\end{itemize}
\end{definition}
\begin{restatable}{theorem}{ProductIsBI}
	If $\mathcal{X}_1$ and $\mathcal{X}_2$ are two $\LOGIC$ frames,
	then $\mathcal{X} = \mathcal{X}_1 \times \mathcal{X}_2$
	is also an \LOGIC frame.
\end{restatable}
The proof is straightforward.

We now define a $\mathbf{2}$-BI frame modeling the independence and NA separation
on the deterministic memories. Because deterministic variables are
automatically independent of other variables, it is meaningless to check whether a set of
deterministic variables can be separated into two disjoint subsets
independent of each other. Thus, we do not require the separation of
domain when modeling the independence and NA of deterministic variables:
\begin{definition}
	Let $\DetSet = \cup_{S \in \Var} \Mem[S]$, and $\sqsubseteq$ be $=$, and the unit set $\DetUnit = \DetSet$. Define $\oplus_d$ by:

	\[ m_1 \oplus_d m_2 =
		\begin{cases}
			\{ m_1 \} &\quad \text{if $m_1 = m_2$} \\
			\emptyset &\quad \text{if $m_1 \neq m_2$}
		\end{cases}
	\]
\end{definition}

	\begin{theorem}
		The structure $\DetModel = (\DetSet, \sqsubseteq_d, \oplus_1, E_1, \oplus_2, E_2)$, where $\oplus_1 = \oplus_2 = \oplus_d$ and $E_1 = E_2 = \DetUnit$,
 is a $\mathbf{2}$-BI frame.
	\end{theorem}

  Both $\DetModel$ and $\ProbModel$ are $\mathbf{2}$-BI frames, so we can take
  their product.
\begin{corollary}
	$\CombModel = \DetModel \times \ProbModel$ is a $\mathbf{2}$-BI frame.
\end{corollary}

As desired, the states of $\CombModel$ describe both deterministic memories and
probabilistic memories.  Furthermore, restricting to the BI model in
$\CombModel$ with operators indexed by 1 recovers the probabilistic BI model
in~\citet{PSL}.

\section{Program logic}
\label{sec:psl}
Given the model for NA developed in the previous section, we now have a suitable
logic of assertions. In this section, we complete the picture by designing a
program logic, named \programlogic, for reasoning about negative association and
independence on probabilistic programs.
We defer proofs and details to
\iffull
\Cref{sec:app:seplog}.
\else
the extended version.
\fi

\subsection{Probabilistic programs}
We consider probabilistic programs in a basic probabilistic imperative language
\Lang. Let $\DetVar, \RanVar$ be disjoint countable subsets of $\Var$ that
respectively contain all deterministic variables and all probabilistic
variables. We consider program states to be a pair of a deterministic memory
$\sigma$, and a distribution $\mu$ over the probabilistic memory, i.e.,
$(\sigma, \mu) \in \Mem[\DetVar] \times \DMem{\RanVar}$.

Because we will want to decompose a program state as a product of two disjoint
memories, each satisfying a sub-formula, we also want to interpret program
expressions on memories whose probabilistic part is only on part of $\RanVar$. These
memories have type $\Mem[\DetVar] \times \DMem{T}$ for some $T \subseteq
\RanVar$, and we call them \emph{configurations}, denoted $\Config$.

We assume all expressions in \Lang are well-typed:
\begin{align*}
	\Expr &\ni e \defdefeq \DetVar \mid \RanVar \mid [\Expr, \dots, \Expr] \mid \Expr + \Expr \mid \Expr \land \Expr \mid \dots
\end{align*}
Given an expression $e$, we can interpret it as $\Mem[\DetVar] \times \Mem[T]
\to \Val$ for any $T \subseteq \RanVar$ that includes all the free variables in
$e$.  We can also lift it to an interpretation from configurations to
distributions of values, i.e.,
\iffull
$\Sem{e}: \Mem[\DetVar] \times \DMem{T} \to \Dist(\Val)$
(see~\Cref{def:semantics_expressions}).
\else
$\Sem{e}: \Mem[\DetVar] \times \DMem{T} \to \Dist(\Val)$.
\fi

We then define commands in \Lang and again assume that they are well-typed:
\begin{align*}
	\Command \ni c &\defdefeq \Skip \mid \Assn{\DetVar}{\Exp} \mid \Assn{\RanVar}{\Exp} \mid  \Rand{\RanVar}{\Unif{T}} \mid \Seq{\Command}{\Command} \\
																	&	\mid \RCond{\Expr}{\Command}{\Command}
																\mid \RWhile{\Expr}{\Command}.
\end{align*}
The randomization is introduced by the sampling command:
$\Rand{\RanVar}{\Unif{T}}$, where $\Unif{T}$ stands for the uniform
distribution on a multi-set  $T$.
We assume that the $\mathbf{while}$ loops terminate in finite steps
	on all inputs.
We also assume that an expression assigned to a
deterministic variable only mentions deterministic variables, and a command
branching on a randomized expression does not assign to deterministic variables
in its body/branches. This assumption ensures that deterministic
variables will not receive randomized values during the execution. It is not
difficult
to enforce this condition by a syntactic restriction, which we omit for a cleaner
presentation.

Following the standard semantics for probabilistic programs due
to~\citet{Kozen81}, we interpret \Lang programs as transformers from
program states to program states, i.e.,
\[\Sem{c}:  \Mem[\DetVar] \times \DMem{\RanVar} \to \Mem[\DetVar] \times \DMem{\RanVar}. \]
\iffull
The semantics of \Lang is standard
(see~\Cref{def:semantics_lang}).
\else
The semantics of \Lang is standard.
\fi

In our examples, \emph{permutation distributions}, uniform distributions over
\emph{$\textsf{permutation}(A)$}:
\begin{definition}
Given a finite multi-set of $A$, a \emph{permutation of $A$} is a bijective function
$\alpha : A \to A$.
We let \emph{}$\textsf{permutation}(A)$ be the multi-set of $A$'s permutations.
When $A$ has duplicates, we distinguish them using additional labels;
so there are always $|A|!$ elements in $\textsf{permutation}(A)$.
\end{definition}

Let \textsf{one-hot}([n]) denote the set of length-$n$ one hot vectors. We then
define the shorthands:
\begin{align*}
  \Rand{\RanVar}{\perm(A)} &\triangleq \Rand{\RanVar}{\Unif{\textsf{permutation}(A)}} \\
  \Rand{\RanVar}{\onehot(n)} &\triangleq \Rand{\RanVar}{\Unif{\textsf{one-hot}([n])}}\\
  \Condt{b}{c} &\triangleq \Cond{b}{c}{\Skip}.
\end{align*}

\subsection{Assertion Logic: atomic propositions and axioms}

Like other program logics, \programlogic has two layers: the program logic
layer describing the relation between pre-conditions, programs and
post-conditions, and the assertion logic layer describing program states.
In~\Cref{sec:model}, we have constructed a probabilistic model of
$\mathbf{2}$-BI, $\CombModel$, whose states encompass all of \Config, so our
starting point for the assertion logic is this model. In this section, we
introduce atomic propositions $\AP$ for describing states in $\CombModel$ and
some axioms that will hold on $\CombModel$.

We extend the core atomic formula from~\citet{PSL}.
To talk about probabilities on program states distributions,
we first define an \emph{event} to be a function that maps a deterministic
program configuration to 0 or 1,
and let $\Event$ be a set of expressions that can be interpreted as event on deterministic configurations
i.e., for any $ev \in \Event$, $\Sem{ev}:  \Mem[\DetVar] \times \Mem[T]
\to \{0, 1\}$ for some $T \subseteq \RanVar$. Since boolean expressions in the
programming language can also be interpreted as this type, we will let $\Event$
include all boolean expression. Let
\begin{align} \label{def:ap}
	\AP \ni p &\defdefeq \apUnif{\Expr}{T} \mid \apBern{\Expr}{p}	\mid \apDetm{\Expr}
  \mid \apEq{\Expr}{\Expr} \mid \Expr \leq \Expr \mid
		\apEventgeneral{\Event}{b} \mid \Pr[\Event] \bowtie \delta
\end{align}
where ${\bowtie} \mathrel{\in} \{ =, \leq, \geq \}$, $b \in \{0, 1\}$, and $\delta \in \mathbb{R}$ is a constant.
In particular, for boolean expression $e$ and for $b \in \{0,1\}$,
	since we can also view $e$ as an event,
	$\apEq{e}{b}$ and $\apEventgeneral{e}{b}$ are both valid atomic propositions.
	We distinguish their notations ($\apEq{}{}$ v.s. $\apEventgeneral{}{}$) because, in
	general, the left hand side of $\apEventgeneral{\Event}{b}$ may not be an
	expression and the left hand side of $\apEq{\Expr}{\Expr}$ may not be an
	event.

We define the satisfaction of atomic proposition on program configurations as follows.
Let $\FV(e)$ be the set of free variables in expression $e$.

\begin{definition}[Atomic Propositions]
	For $(\sigma, \mu) \in \CombModel$, define
	\begin{itemize}
		\item $(\sigma, \mu) \models \apUnif{e}{T}$ iff $\FV(e) \subseteq
		\dom(\sigma) \cup \dom(\mu)$ and $\Sem{e}(\sigma, \mu)$ is a distribution
		that assigns probability $\frac{1}{|T|}$ to each element of $T$;
		\item $(\sigma, \mu) \models \apBern{e}{p}$ iff $\FV(e) \subseteq \dom(\sigma)
		\cup \dom(\mu)$ and $\Sem{e}(\sigma, \mu)$ is a distribution that assign
		probability $p$ to 1 and probability $1-p$ to 0, i.e., the \emph{Bernoulli}
		distribution;
		\item $(\sigma, \mu) \models \apDetm{e}$ iff $\FV(e) \subseteq
		\dom(\sigma) \cup \dom(\mu)$ and $\Sem{e}(\sigma, \mu)$ is a Dirac
		distribution;
		\item $(\sigma, \mu) \models \apEq{e}{e'}$ iff $\FV(e) \cup \FV(e')
		\subseteq \dom(\sigma) \cup \dom(\mu)$ and $\Sem{e}(\sigma, m)
		=\Sem{e'}(\sigma, m)$ for any $m$ in the support of $\mu$;
		\item $(\sigma, \mu) \models e \leq e'$ iff  $\FV(e) \cup
		\FV(e') \subseteq \dom(\sigma) \cup \dom(\mu)$ and $\Sem{e} (\sigma, m) \leq
		\Sem{e'} (\sigma, m)$ for any $m$ in the support of $\mu$;
	\item $(\sigma, \mu) \models \apEventgeneral{ev}{b}$ if for any $m$ in the support
		of $\mu$, $\Sem{ev}(\sigma, m) = b$.
		\item $(\sigma, \mu) \models \Pr[ev] \bowtie \delta$ iff the probability of
		event $\Sem{ev}$ in $(\sigma, \mu)$, defined to be $\Pr_{(\sigma, \mu)}[ev] =
		\sum_{m \in \Mem[\dom(\mu)]} \mu(m) \cdot \Sem{ev}(\sigma, m)$, satisfies
		$\Pr_{(\sigma, \mu)}[e] \bowtie \delta$.
	\end{itemize}

  We use the abbreviations:
  \begin{itemize}
    \item $\apIn{e} \triangleq \apEq{e}{e}$. That is, $(\sigma, \mu) \models
      \apIn{e}$ holds if all of the variables in $e$ are defined in $\sigma$ and
      $\mu$.
    \item	$\apOnehot{e}{N} \triangleq \apUnif{e}{\textrm{one-hot([N])}}$.
    \item For multi-set $A$, $\apPerm{e}{A} \triangleq \apUnif{e}{\textrm{permutation(A)}}$.
  \end{itemize}
\end{definition}

For any operation $\odot \in \{ \land, \lor, \negand, \indand\}$, we pick the
corresponding big-operation $\bigodot \in \left\{ \bigwedge, \bigvee, \bigneg,
\bigind\right\}$ to be
\iffull
their iterated version (see~\Cref{def:bignotation}).
\else
their iterated version.
\fi

With the atomic propositions and abbreviations, we can formally state that $\CombModel$ captures NA.
\begin{restatable}{theorem}{starcapturesNA}
		\label{theorem:starcapturesNA}
		Let $S$ be any subset of $\RanVar$.
		A set of randomized program variables $Y = \{y_i \mid 0 \leq i < K\}$ satisfies NA in distribution $\mu \in \DMem{S}$ if and only if
		for any deterministic memory $\sigma \in \Mem[\DetVar]$, we have
		$(\sigma, \mu) \models \bigneg_{i = 0}^{K} \apIn{y_i} $.
	\end{restatable}

  In the $\CombModel$ model, all axioms from~\citet[Lemma 3, 4]{PSL} still hold,
  and we have new axioms for the negative association conjunction and the
  permutation distribution.

	\begin{restatable}{lemma}{PNAIntro}
		Let $x_{\gamma}$ be variables. The following axioms are valid in $\CombModel$.
  \begin{align}
			&\models \apOnehot{[x_0, \dots, x_{N - 1}]}{N} \to \bigneg_{\gamma = 0}^{N} \apIn{x_{\gamma}} \tag{OH-PNA}
      \label{ax:oh-pna}
			\\
			&\models \apPerm{[x_0, \dots, x_{N -1}]}{A} \to \bigneg_{\gamma = 0}^{N} \apIn{x_{\gamma}} \tag{Perm-PNA}
      \label{ax:perm-pna}
		\end{align}
\end{restatable}
The two axioms follow from~\Cref{thm:buildingblock}, which shows that random
variables in one-hot distributions and permutation distributions are NA,
and~\Cref{theorem:starcapturesNA}, which shows that $\negand$ captures the NA of
random variables. We can also encode the monotone map closure
in~\Cref{thm:closure} as an axiom in the logic.

\begin{restatable}[\textsc{Monotone map}]{lemma}{MonotoneMapAxiom}
	Let $x$, $x_{\gamma, \alpha}$ and $y_{\gamma}$ be variables.
	The following is valid in $\CombModel$.
  \begin{align}
			\models \bigneg_{\gamma = 0}^{N} \left(\bigwedge_{\alpha = 0}^{K_{\gamma} + 1} \apIn{x_{\gamma, \alpha}} \right) &\land \bigwedge_{\gamma = 0}^{N} y_{\gamma} = f_{\gamma} \left(x_{\gamma, 0}, \dots, x_{\gamma, K_{\gamma} } \right) \to \bigneg_{\gamma=0}^{N} \apIn{y_{\gamma}}
      \notag \\
			&\textrm{ when } f_1, \dots, f_N \textrm{ all monotone or all antitone}
			\tag{Mono-Map} \label{ax:mono-map}
		\end{align}
\end{restatable}

When we establish NA from permutation distributions, it is preserved under not
only monotone/antitone maps but also any element-wise homogeneous maps. The
reason is that fixing a multi-set and a permutation, permuting first and then
applying the same map on each element is the equivalent to applying the map on
each element and then permuting. So applying homogeneous maps on a permutation
distribution gives another permutation distribution. We can capture this
property in an axiom.
\begin{lemma}[Permutation Map]
  Let $x_{\gamma}$ be variables, and $f(A)$ be $\{f(a) \mid a \in A \}$. The
  following axiom is valid in $\CombModel$.
  \begin{equation}
    \label{ax:perm-map}
    \models \apPerm{[x_1, \dots, x_N]}{A} \land \apEq{y}{[f(x_1), \dots, f(x_N)]} \to \apPerm{y}{f(A)} \tag{Perm-Map}
  \end{equation}
\end{lemma}

\subsection{Restricting the assertion language}
When designing a separation logic for reasoning about negative association and
independence, we sometimes want to separate out a smaller configuration
$(\sigma', \mu')$ inside a given program state $(\sigma, \mu) \models \phi$,
such that $(\sigma', \mu')$  satisfies some sub-formula of $\phi$.  In the
program logic we will present in~\Cref{sec:programlogic}, the soundness of
\textsc{RCase}, \textsc{Const}, \textsc{Frame} and \textsc{NegFrame} rules all
rely on the ability to do that.  To ensure there exists such a smaller
configuration, we require the assertion logic to satisfy a key condition called
\emph{restriction}, which says that to check whether a configuration satisfies
$\phi$, it suffices to check whether the configuration's projection on
$\FV(\phi)$ satisfies $\phi$.  We identify a subset of \LOGIC formulas that
satisfy the restriction property when interpreted on states in $\CombModel$:

\begin{definition}
	We define \AssertLOGIC as
\begin{align*}
	\AssertLOGIC \ni P, Q &::= p \in \mathcal{AP}
	\mid \top
	\mid \bot
	\mid P \land Q
	\mid P \lor Q
	\mid P \rightarrow Q
	\mid P \indand Q
	\mid P \indimp Q
	\mid P \negand Q
\end{align*}
where $\mathcal{AP}$ is defined as in~\Cref{def:ap}.
\end{definition}

\AssertLOGIC omits multiplicative identities $I_{\idx}$ because on $\CombModel$
they are all equivalent to $\top$.
The only limitation is that \AssertLOGIC excludes the use of $\negimp$.

\begin{restatable}[Restriction]{theorem}{MBIRestriction}
	\label{theorem:MBIRestriction}
	Let $(\sigma, \mu)$ be any configuration,
		and let $\phi$ be an \AssertLOGIC formula interpreted on $\CombModel$,
		Then, for any $m \in \Mem[\DetVar \setminus \FV(\phi)]$,

		\[(\sigma, \mu) \models \phi \iff (\p_{\FV(\phi)}\sigma \bowtie m, \pi_{\FV(\phi)}\mu) \models \phi. \]
	\end{restatable}


  Indeed, we can exhibit a counterexample showing that $\negimp$ does not
  satisfy restriction.

	\begin{restatable}{theorem}{CounterEx}
		There exists $(\sigma, \mu) \in \Config$ and formula $\phi$ such that
		$(\sigma, \mu) \models \phi$ but $(\sigma, \pi_{\FV(\phi)}) \not\models \phi$.
	\end{restatable}

  In the following, we will consider $\AssertLOGIC$ formula on the $\CombModel$
  model as the assertion logic.

\subsection{The program logic}
\label{sec:programlogic}
We now introduce the program logic layer of \programlogic.
Judgements in \programlogic have the form $\hoare{\prop}{c}{\propB}$, where $c \in \Command$ is a probabilistic program,
and $\prop, \propB \in \AssertLOGIC$ are restricted assertions.

	\begin{definition}[Validity]
    A \programlogic judgment is \emph{valid}, written $\models
    \hoare{\prop}{c}{\propB}$, if for all $( \sigma, \mu ) \in \Mem[\DetVar]
    \times \DMem{\RanVar}$ such that $( \sigma, \mu )  \models \prop$, we have
    $\Sem{c}(\sigma, \mu ) \models \propB$.
	\end{definition}

	\begin{figure}
		\begin{mathpar}
			\inferrule*[Left=DAssn]
			{~}
			{ \vdash \hoare{\psi[e_d/x_d]}{\Assn{x_d}{e_d}}{\psi} }
			\and
			\inferrule*[Left=Skip]
			{~}
			{ \vdash \hoare{\phi}{\Skip}{\phi} }
			\and
			\inferrule*[Left=Seqn]
			{ \vdash \hoare{\phi}{c}{\psi} \\ \vdash \hoare{\psi}{c'}{\eta} }
			{ \vdash \hoare{\phi}{\Seq{c}{c'}}{\eta} }
			\\
			\inferrule*[Left= Cond]
      { \vdash \hoare{\phi \land \apEq{b}{\ktt}}{c}{\psi}
        \\ \vdash \hoare{\phi \land \apEq{b}{\kff} }{c'}{\psi}
		\\ \models \phi \rightarrow \apDetm{b} }
			{ \vdash \hoare{\phi}{\RCond{b}{c}{c'}}{\psi} }
			\and
			\inferrule*[Left= Loop]
      { \vdash \hoare{\phi \land \apEq{b}{\ktt}}{c}{\phi}
			\\ \models \phi \to \apDetm{b} }
      { \vdash \hoare{\phi}{\RWhile{b}{c}}{\phi \land \apEq{b}{\kff} } }
		\and
			\inferrule*[Left=RAssn]
			{ x_r \notin \FV(e_r) }
      { \vdash \hoare{\top}{\Assn{x_r}{e_r}}{\apEq{x_r}{e_r} } }
			\and
			\inferrule*[Left=RSamp]
			{~}
			{ \vdash \hoare{\top}{\Rand{x_r}{\Unif{S}}}{\apUnif{x_r}{S}} }
			\and
			\inferrule*[Left=RSamp*]
			{x_r \notin \FV(\phi) }
			{ \vdash \hoare{\phi}{\Rand{x_r}{\Unif{S}}}{\phi \indand \apUnif{x_r}{S}} }
      \\
			\inferrule*[Left=Weak]
			{ \vdash \hoare{\phi}{c}{\psi} \\ \models \phi' \to \phi \land \psi \to \psi' }
			{ \vdash \hoare{\phi'}{c}{\psi'} }
			\and
			\inferrule*[Left=True]
			{~}
			{ \vdash \hoare{\top}{c}{\top} }
			\\
			\inferrule*[Left=Conj]
			{ \vdash \hoare{\phi_1}{c}{\psi_1} \\
			\vdash \hoare{\phi_2}{c}{\psi_2} }
			{ \vdash \hoare{\phi_1 \land \phi_2}{c}{\psi_1 \land \psi_2} }
			\and
			\inferrule*[Left=Case]
			{ \vdash \hoare{\phi_1}{c}{\psi_1} \\
			\vdash \hoare{\phi_2}{c}{\psi_2} }
			{ \vdash \hoare{\phi_1 \lor \phi_2}{c}{\psi_1 \lor \psi_2} }
			\\
			\inferrule*[Left=Const]
			{ \vdash \hoare{\phi}{c}{\psi} \\ \FV(\eta) \cap \MV(c) = \emptyset }
			{ \vdash \hoare{\phi \land \eta}{c}{\psi \land \eta} }
			\\
			\inferrule*[Left=Frame]
			{ \vdash \hoare{\phi}{c}{\psi} \\
				\FV(\eta) \cap \MV(c) = \emptyset \\
				\FV(\psi) \subseteq T \cup \RV(c) \cup \WV(c) \\
			\models \phi \to \apIn{T \cup \RV(c)}}
			{ \vdash \hoare{\phi \indand \eta}{c}{\psi \indand \eta} }
		\end{mathpar}

		\caption{\programlogic rules: from PSL.}
		\label{def:programlogic:psl}
  \end{figure}

  \begin{figure}
		\begin{mathpar}
      \inferrule*[Left=RCase]
      { \eta \in \CCond
        \\ \models_\Mem \eta \to \bigvee_{\alpha \in S} \eta_{\alpha}
        \\ \psi \in \text{CM}
        \\ \forall \alpha \in S.~\vdash \hoare{\phi \indand \eta_{\alpha}}{c}{\psi}
      }
      { \vdash \hoare{\phi \indand \eta}{c}{\psi} }
						\\
      \inferrule*[Left=NegFrame]
      {
        \models \phi \to \apIn{\RV(c)} \\
        \FV(\eta) \cap \MV(c) = \emptyset \\
        X \subseteq \RV(c) \setminus \MV(c) \\
        y \notin \FV(\eta) \\\\
        \vdash \hoare{\phi}{c}{y \sim f(X)} \\
        \textrm{$f$ is a monotone function}
      }
      { \vdash \hoare{ \phi \negand \eta}{c}{\apIn{y}  \negand \eta} }
						\\
						\inferrule*[Left=ProbBound]
						{\vdash \hoare{\apEvent{ev_1}}{c}{\Pr[ev_2] \leq \delta}}
						{\vdash \hoare{\Pr[ev_1] \geq 1 - \epsilon}{c}{\Pr[ev_2] \leq \delta + \epsilon}}
		\end{mathpar}

		\caption{\programlogic rules: new and extended.}
		\label{def:programlogic:new}
	\end{figure}

  Next, we present the proof system of \programlogic. Since our assertions are a
  conservative extension of assertions from PSL, most of the rules carry over
  unchanged; we list existing rules in \Cref{def:programlogic:psl}. Here, we
  comment on the new and generalized rules, which we list in
  \Cref{def:programlogic:new}.

  \paragraph*{NA frame rule.}
  Our most important addition is the frame rule for the negative association
  conjunction $\negand$. Informally, the \textsc{NegFrame} rule says that if a
  set of variables $X$ is negatively associated with another set of variables
  $Y$ that satisfy $\eta$ in a program state, and the program $c$ performs a
  monotone operation $f$ on $X$ and stores the result in a variable $y$, then in
  the resulting program state, $y$ and the untouched variables $Y$ will also be
  negatively associated, and $Y$ will still satisfy $\eta$.  Like the
  \textsc{Frame} rule for independence $\indand$, the \textsc{NegFrame} rule
  uses syntactic restrictions to control which variables the program may read
  and write. The three sets of variables $\RV(c), \WV(c), \MV(c)$ represent the
  variables that $c$ may read from, must write to, and may modify, respectively;
  these sets can be defined by induction on the syntax of the program. Roughly,
  the side conditions guarantee the program $c$ does not read from or modify $Y$,
  the set of variables satisfying $\eta$; they in addition guarantee that $X$,
  the domain of the monotone map will not be modified by $c$, and $y$, the
  codomain of the monotone map does not belong to $Y$.


  \paragraph*{Generalized random case analysis.}
  As a more minor extension, we also generalize the randomized case analysis
  rule from PSL in \textsc{RCase}. At a high level, this rule allows reasoning
  by case analysis on a property $\eta$ of the program memory (e.g., whether a
  variable is true or false). Since the input is a distribution, which may have
  some probability of $\eta$ holding, and some probability of $\eta$ not
  holding, soundness of the rule is a delicate matter requiring several
  technical side conditions. The original rule in PSL only allowed case analysis
  on a Boolean expression; we generalize this rule to allow a case analysis on
  any finite number of cases (e.g., performing case analysis on the value of a
  bounded variable).

  To explain this rule, we first introduce the side conditions in order. We say
  that a formula $\eta$ is \emph{closed under conditioning} (CC) if for any
  $(\sigma, \mu) \models \eta$, for any $m$ in the support of $\mu$,
  $(\sigma, \delta(m)) \models \eta$. In the second condition, $\models_\Mem
  \phi$ denotes that for any $\sigma \in \Mem[\DetVar]$, $m \in \Mem[T]$ where
  $T \subseteq \RanVar$, $(\sigma, \delta(m)) \models \phi$, which says
		$\phi$ is valid on all effectively deterministic configurations. Finally, we say
  that a formula $\phi$ is \emph{closed under mixtures} (CM) if $(\sigma, \mu_1)
  \models \phi$,	$(\sigma, \mu_2) \models \phi$ and $\mu$ is a convex
  combination of $\mu_1, \mu_2$ together imply $(\sigma, \mu') \models \phi$.

  Then, the rule \textsc{RCase} says if an assertion $\eta$ is independent from
  the rest of the assertions in the pre-condition, $\eta$ is closed under
  conditioning, and the post-condition $\psi$ is closed under mixtures, then we
  can perform case analysis on $\eta$ to derive $\hoare{\phi \sepand
  \eta}{c}{\psi}$. Intuitively, every memory $m$ in the support of the input
  memory distribution satisfies $\eta_a$ for some case $a \in S$. The main premise shows that
  the output distribution of program $c$ from any such input $m$ satisfies $\psi$. Then,
  since any distribution $\mu$ on inputs is a convex combination of
  such memories $m$, and $\psi$ holds on each conditional output distribution,
		we have $\psi$ holds on the entire output distribution by convex closure.

  \paragraph*{Bounding bad events.}
  In addition, we present the rule \textsc{ProbBound} to facilitate bounding
  tail probabilities. It says that if the pre-condition $\apEvent{ev_1}$
  guarantees that event $ev_2$ happens for at most $\delta$ probability
  after command $c$, then in general, event $ev_2$ happens for at most
  probability $\delta + \epsilon$ after $c$, where $\epsilon$ upper bounds the
  probability that $ev_1$ is not true in the pre-condition.  The validity of
  this rule uses the law of total probability, which says for any two events
  $ev_1$ and $ev_2$,
  \begin{align*}
    \Pr(ev_1) &= \Pr(ev_1 \mid ev_2) \cdot \Pr(ev_2) + \Pr(ev_1 \mid \neg ev_2) \cdot \Pr(\neg ev_2) \\
              &\leq \Pr(ev_1 \mid ev_2) +  \Pr(\neg ev_2).
  \end{align*}
  As expected, the \programlogic proof system is sound.
		\begin{restatable}{theorem}{SoundnessProgramlogic}
    \textrm{(Soundness of \programlogic)}
    If $\vdash \hoare{\phi}{c}{\psi}$ is derivable,
    then it is valid: $\models \hoare{\phi}{c}{\psi}$.
  \end{restatable}

\section{Examples}

\label{sec:ex}

Now that we have introduced \programlogic, we present a series of formalized
case studies. Our examples are extracted from various algorithms using hashing
and balls-into-bins processes.

\subsection{Preliminaries: probabilities, expectations, concentration bounds}

Our examples will use a handful of standard facts about probability
distributions, encoded as axioms in the assertion logic. We will
generally mention these axioms before they are used, but here we
introduce one fact that we will use through all of our examples: the
\emph{Chernoff bound}.

In each of our examples we will establish negative
dependence of a sequence of random variables $\{ X_i \}_i$ and apply a
\emph{concentration bound}: a theorem showing that the sum $X_1 + \cdots + X_n$
is usually close to its expected value. This kind of analysis is useful for
establishing \emph{high-probability guarantees} of randomized algorithms, e.g.,
showing that the error of a random estimate is at most $0.01$ with probability
at least $99\%$.

\begin{theorem}[Chernoff bound for NA variables~\citep{dubhashi-ranjan}]
  \label{thm:chernoff}
  Let $X_1, \dots, X_n$ be a sequence of NA random variables, each bounded in
  $[0, 1]$, and let $Y = \sum_{i = 1}^n X_i$. Then for any failure probability
  $\beta \in (0, 1]$, we have:
  \[
    \Pr[ |Y - \EE[Y]| \geq T(\beta, n) ] \leq \beta
				\ \
				\text{ where }
				T(\beta, n) = \sqrt{ (n/2) \ln(2/\beta) }.
	\]
\end{theorem}

To hide complex numerical bounds, we use the notation $T(\beta, n)$
for the above function throughout. In our assertion logic, the
Chernoff bound can be encoded as the following axiom schema:

\begin{theorem}[Chernoff bound, axiom]
  Let $\{ x_\alpha \}$ be a family of variables indexed by $\alpha$, where
  each variable is bounded in $[0, 1]$ and is a monotone function of its program
  variables. Then for any $\beta \in (0, 1]$, the following axiom schema is
  sound in our model:
  \begin{equation} \label{ax:na-chernoff}
    \models \bigneg_{\alpha = 0}^N \apDist{x_\alpha}
    \to \Pr \left[ \ \bigmid \sum_{\alpha = 0}^N x_\alpha - \EE\left[\sum_{\alpha = 0}^N x_\alpha\right] \bigmid\ \geq T(\beta, n) \right] \leq \beta
    \tag{NA-Chernoff}
  \end{equation}
\end{theorem}

We will also use a new expression in our assertions: $\EE[f]$, where $f$ is a
non-negative and bounded numeric expression, denotes the expected value of $f$
in the current program configuration. We also observe the following conventions
throughout the examples: logical variables are denoted by Greek ($\alpha, \beta,
\gamma, \dots$) and capital Roman letters ($M, N, K, \dots$). Program variables
start with lower-case Roman letters ($x, y, z, \dots$).

\subsection{Bloom filter, high-level}

\begin{figure*}
  \begin{subfigure}[b]{0.48\textwidth}
    \[
      \begin{array}{l}
        \textsc{Bloom}: \\
        \quad\Assn{bloom}{zero(N)}; \\
        \quad\Assn{m}{0}; \\
        \quad\DWhile{m < M}{} \\
        \quad\quad \Assn{h}{0} \\
        \quad\quad \DWhile{h < H}{} \\
        \quad\quad\quad \Rand{bin}{oh([N])}; \\
        \quad\quad\quad \Assn{upd}{bloom \mathop{||}  bin}; \\
        \quad\quad\quad \Assn{bloom}{upd}; \\
								\\
								\\
								\\
        \quad\quad\quad \Assn{h}{h+1}; \\
        \quad\quad \Assn{m}{m + 1};
      \end{array}
    \]
    \caption{Higher-level version}
    \label{fig:ex:bloom-v1}
  \end{subfigure}
  \hfill
  \begin{subfigure}[b]{0.48\textwidth}
    \[
      \begin{array}{l}
        \textsc{BloomArray}: \\
        \quad\Assn{bloom}{zero(N)}; \\
        \quad\Assn{m}{0}; \\
        \quad\DWhile{m < M}{} \\
        \quad\quad \Assn{h}{0} \\
        \quad\quad \DWhile{h < H}{} \\
        \quad\quad\quad \Rand{bin}{oh([N])}; \\
        \quad\quad\quad \Assn{n}{0}; \\
        \quad\quad\quad \DWhile{n < N}{} \\
        \quad\quad\quad\quad \Assn{upd}{bloom[n] \mathop{||} bin[n]}; \\
        \quad\quad\quad\quad \Assn{bloom[n]}{upd}; \\
        \quad\quad\quad\quad \Assn{n}{n + 1} \\
        \quad\quad\quad \Assn{h}{h+1}; \\
        \quad\quad \Assn{m}{m + 1}
      \end{array}
    \]
    \caption{Array version}
    \label{fig:ex:bloom-v2}
  \end{subfigure}

  \caption{Bloom filter examples}
\end{figure*}

Next, we revisit the Bloom filter example introduced in \Cref{sec:overview}.  We
show how to translate the informal argument in \Cref{sec:overview} into formal
proofs in our program logic. First, we will analyze the process of adding items
into a Bloom filter $bloom$ and prove that the entries in $bloom$ are negatively
associated at the end of the process. Second, we will analyze a program that
checks the membership of a new item in a given Bloom filter and show how
to bound its false positive rate. Last, we combine them together into one proof
that bounds the false positive rate of a Bloom filter with $M$ elements.

\paragraph{Proving NA of $bloom$}
We reproduce the code for \textsc{Bloom} in \Cref{fig:ex:bloom-v1}.
This program is a higher-level version of the program in \Cref{fig:ex:bloom-v2},
which performs array operations bit-by-bit.
We align the two versions so that the equivalent operations are side-by-side.
We will demonstrate our program logic on the higher-level version first and
analyze the array version later in~\Cref{sec:bloom_filter:lowlevel}.

Recall that the code models inserting $M$ distinct elements into a Bloom filter
backed by an array $bloom$ of length $N$, where each element is hashed by $H$
functions, each producing an element of $[N]$ uniformly at random.
We refer to the outer loop as $\mathit{outer}$, and the inner
loop as $\mathit{inner}$. For both the outer and the inner loop, we apply the
rule \textsc{Loop} with the loop invariant:
$ \bigneg_{\beta = 0}^N \apDist{bloom[\beta]} $.
We consider the inner loop first. We show that the invariant is preserved by the
body of $\mathit{inner}$. After the $oh$ sampling command, \textsc{RSamp*} gives:
\[ \left(\bigneg_{\beta = 0}^N \apDist{bloom[\beta]}\right) \indand \apOnehot{bin}{[N]} \]
By negative association of the one-hot distribution \eqref{ax:oh-pna}, we get
\[ \left(\bigneg_{\beta = 0}^N \apDist{bloom[\beta]}\right) \indand \left(\bigneg_{\gamma = 0}^N bin[\gamma]\right)\]
which implies
\[ \left(\bigneg_{\beta = 0}^N \apDist{bloom[\beta]}\right) \negand \left(\bigneg_{\gamma = 0}^N bin[\gamma]\right)\]
using \textsc{Weak}. Rearranging terms, this is equivalent to
\[ \bigneg_{\beta = 0}^N \apDist{bloom[\beta]} \negand \apDist{bin[\beta]} . \]
After the assignment to $upd$, we have:
\[
  \left(\bigneg_{\beta = 0}^N \apDist{bloom[\beta]} \negand \apDist{bin[\beta]} \right)
  \land \apEq{upd}{bloom \mathop{||} bin} .
\]
Because $\mathop{||}$ is monotone, applying the monotone mapping axiom
\eqref{ax:mono-map} gives us:
\[
  \bigneg_{\beta = 0}^N \apDist{upd[\beta]} .
\]
Using the assignment rule (\textsc{RAssn}) on the assignment to $\mathit{bloom}$
shows that the loop invariant is preserved by the inner loop. Thus,
\textsc{Loop} gives:
\[
  \hoare{\bigneg_{\beta = 0}^N \apDist{bloom[\beta]}}
  {\mathit{inner}}
  {\bigneg_{\beta = 0}^N \apDist{bloom[\beta]}}
\]
Next, we turn to the outer loop. The argument showing that the invariant is
preserved by the outer loop follows by a straightforward argument, since the
outer loop only modifies $bloom$ through the inner loop, so \textsc{Loop} gives:
\[
  \hoare{\bigneg_{\beta = 0}^N \apDist{bloom[\beta]}}
  {\mathit{outer}}
  {\bigneg_{\beta = 0}^N \apDist{bloom[\beta]}}
\]
Then, we have:
\[
  \hoare{\top}{\textsc{Bloom}}{\bigneg_{\beta = 0}^N \apDist{bloom[\beta]}}
\]
because initializing $bloom$ to the all-zeros vector, a deterministic value,
establishes the loop invariant. This judgment shows that the $bloom$ vector
satisfies NA at the end of the program.

\begin{wrapfigure}[12]{r}{0.4\textwidth}
  \begin{minipage}[c]{0.4\textwidth}
    \vspace{-1ex}
		\[
			\begin{array}{l}
				\CheckMem(H, bloom): \\
        \quad \Assn{h}{0}; \\
							\quad\Assn{allhit}{1} \\
        \quad \DWhile{h < H}{} \\
								\quad\quad \Rand{bin}{\Unif{[N]}}; \\
								\quad\quad \Assn{hit}{bloom[bin]}; \\
								\quad\quad\Assn{allhit}{hit \mathop{\&\&} allhit}; \\
       \quad\quad\Assn{h}{h+1};
						\end{array}
		\]
	\end{minipage}
	\caption{Check the membership of a new item}
	\label{fig:ex:bloom:checkmem}
\end{wrapfigure}

\paragraph{Bounding the false positive rate}
Now, we turn to verifying a bound on the false positive rate of the Bloom
filter. Recall that a false positive occurs if when querying with an element
that was not inserted, the filter returns true. We can encode the membership
check of a new element as a program \CheckMem$(H,bloom)$, listed
in~\Cref{fig:ex:bloom:checkmem}, which hashes the new element into $H$ uniformly
random positions and checks if these positions are all set to one in the filter.
If so, the Bloom filter will report that the new element is in set, when it was
never inserted---a false positive.

To verify the false positive rate, we place the program $\CheckMem(H, bloom)$
immediately after \textsc{Bloom}, and then verify a bound on the probability
that $allhit$ is 1 at the end of the combined program. We first apply the
Chernoff bound to the NA variables \eqref{ax:na-chernoff} to prove that, with
high probability, the number of occupied bins in \textsc{Bloom} is near its mean
with high probability:
\[
  \Bigg\{\top\Bigg\}{\textsc{Bloom}}
		\Bigg\{\Pr \left[ \ \bigmid \sum_{\beta = 0}^N bloom[\beta] - \EE\left[\sum_{\beta = 0}^N
  bloom[\beta] \right] \bigmid\ \geq T(\delta, N) \right] \leq \delta \Bigg\}.
\]
This concentration bound implies that a tail bound, which says with high
probability $\sum_{\beta = 0}^N bloom[\beta]$ is upper bounded by its expected
value plus $T(\delta, N)$,
\begin{align}
	\label{eq:ex:bloom}
	\Bigg\{\top\Bigg\} {\textsc{Bloom}}
		\Bigg\{\Pr \left[ \sum_{\beta = 0}^N bloom[\beta] < \EE\left[\sum_{\beta = 0}^N
		bloom[\beta] \right] + T(\delta, N) \right]  \geq 1 - \delta \Bigg\}.
	\end{align}
Then we analyze \CheckMem and show in
\iffull
\Cref{app:ex}
\else
the extended version
\fi{}
that
\begin{align*}
	\Big\{\sum_{\beta = 0}^N bloom[\beta] < K \Big\}{\CheckMem}
	\Big\{\Pr[allhit] \leq (K/N)^H \Big\}.
\end{align*}
Then, by the \textsc{ProbBound} rule and basic axioms about probabilities,
we have
\begin{align}
	\label{eq:ex:checkmem}
	\hoare{\Pr[\sum_{\beta = 0}^N bloom[\beta] < K] \geq 1 - \delta}{\CheckMem}{\Pr[allhit] \leq (K/N)^H + \delta}.
\end{align}
We then use \textsc{Seqn} to combine the proved judgements for \textsc{Bloom}
~\eqref{eq:ex:bloom} and \CheckMem~\eqref{eq:ex:checkmem} to derive that, for any
$\delta$,
\begin{align*}
	\hoare{\top}{\textsc{Bloom}; \CheckMem }
	{\Pr[allhit] \leq \left( \frac{\EE\big[\sum_{\beta = 0}^N bloom[\beta]\big] +
	T(\delta, N)}{N} \right)^H + \delta}.
\end{align*}
Since $allhit$ is 1 exactly when there is a false positive, this judgment
proves an upper bound on the false positive rate of the Bloom filter.\footnote{%
  The precise expected value is $N \cdot (1 - (1 - 1/N)^{M \cdot H})$, a fact
  which can also be shown in our logic. Roughly speaking, this fact follows
  because each element of $bloom$ is the logical-or of $M \cdot H$
  probabilistically independent bits, each $1$ with probability $1/N$ and $0$
otherwise. This argument does not rely on negative association.}

\subsection{Bloom filter, low-level}
\label{sec:bloom_filter:lowlevel}

The previous Bloom filter uses a vector operation $bloom \mathop{||} bin$ to
transform an array of negatively associated values. We next consider a
lower-level version of the previous example, \textsc{BloomArray}, in
\Cref{fig:ex:bloom-v2}, where the vector operation is replaced by a loop that
applies the Boolean-or.

Let $\mathit{outer}$ and $\mathit{mid}$ be the outer-most and second outer-most
loops, and let $\mathit{inner}$ be the inner-most loop.  Again, our goal is to
show that the vector $bloom$ is negatively associated at the end of the program.
We first prove the following judgment for $\mathit{inner}$:
\[
  \hoare{\bigneg_{\beta = 0}^N \apDist{bloom[\beta]}
    \indand \bigneg_{\gamma = 0}^N \apDist{bin[\gamma]} }
  {\mathit{inner}}
  {\bigneg_{\beta = 0}^N \apDist{bloom[\beta]}
	\negand
	\bigneg_{\gamma = n}^N \apDist{bin[\gamma]}}
\]
We will apply the rule \textsc{Loop} on $\mathit{inner}$ with the following loop
invariant:
\[
  \phi = \bigneg_{\beta = 0}^N \apDist{bloom[\beta]}
  \negand
  \bigneg_{\gamma = k}^N \apDist{bin[\gamma]}
\]
To show that the loop invariant is preserved by the body, we can first show:
\[
  \hoare{\apDist{bloom[n], bin[n]}}
  {\Assn{upd}{bloom[n] \mathop{||} bin[n]}}
  {\apEq{upd}{bloom[n] \mathop{||} bin[n]}}
\]
using \textsc{RAssn}. Noting that the boolean-or operator is a monotone operation, we
may apply the frame rule \textsc{NegFrame} to obtain:
\[
  \hoare{\apDist{bloom[n], bin[n]} \negand \eta}
  {\Assn{upd}{bloom[n] \mathop{||} bin[n]}}
  {\apDist{upd} \negand \eta}
\]
with the framing condition
\[
  \eta = \left( \bigneg_{\beta = 0}^n \apDist{bloom[\beta]} \right)
  \negand \left( \bigneg_{\beta = n + 1}^N \apDist{bloom[\beta]} \right)
  \negand \left( \bigneg_{\gamma = n + 1}^N \apDist{bin[\gamma]} \right) .
\]
Thus, by re-associating the separating conjunction and applying \textsc{RAssn}
for the remaining two assignments in the inner-most loop, we have:
\[
  \hoare{\phi}{\Assn{upd}{bloom[n] \mathop{||} bin[n]}; \Assn{bloom[n]}{upd};
  \Assn{n}{n + 1}}{\phi}
\]
and thus by \textsc{Loop}, we have:
\[
  \hoare{\bigneg_{\beta = 0}^N \apDist{bloom[\beta]}
    \negand
  \bigneg_{\gamma = n}^N \apDist{bin[\gamma]}}
  {\mathit{inner}}
  {\bigneg_{\beta = 0}^N \apDist{bloom[\beta]}
    \negand
  \bigneg_{\gamma = n}^N \apDist{bin[\gamma]}} .
\]
Now for loop $\mathit{mid}$, we establish the same loop invariant as we took
before:
\[
  \psi = \bigneg_{\beta = 0}^N \apDist{bloom[\beta]}
\]
If $\psi$ holds at the beginning of $\mathit{mid}$, then invariant for the
inner-most loop $\phi$ holds after assigning $0$ to $n$ and sampling $bin$,
since $bin$ is independent of $\psi$ (\textsc{RSamp*}) and $bin$ is distributed
as $\onehot(n)$, which implies entries in $bin$ are negatively associated
\eqref{ax:oh-pna}. Furthermore, $\phi$ implies $\psi$ at the exit of
$\mathit{inner}$, by dropping the conjunct describing $bin$. Thus, $\psi$ is a
valid invariant for $\mathit{mid}$, and the rest of the proof proceeds
unchanged.

\subsection{Permutation hashing}

\begin{wrapfigure}[12]{r}{0.4\textwidth}
  \begin{minipage}[c]{0.4\textwidth}
    \vspace{-3ex}
    \[
      \begin{array}{l}
        \textsc{PermHash}: \\
        \quad\Rand{g}{\perm([B \cdot K])}; \\
        \quad\Assn{n}{0}; \\
        \quad\Assn{ct}{0}; \\
        \quad\DWhile{n < N}{} \\
        \quad\quad \Assn{bin[n]}{mod(g[n], B)}; \\
        \quad\quad \Assn{hitZ[n]}{[bin[n] = Z]}; \\
        \quad\quad \Assn{ct}{ct + hitZ[n]}; \\
				\quad\quad \Assn{n}{n+1}
      \end{array}
    \]
  \end{minipage}
  \caption{Permutation hashing}
  \label{fig:ex:perm}
\end{wrapfigure}

Our second example considers a scheme for hashing using a random permutation.
Consider the program in \Cref{fig:ex:perm}, from an algorithm for fast set
intersection~\citep{DBLP:journals/pvldb/DingK11}. Letting $B$ be the number of
bins, and the data universe be $[B \cdot K] = \{ 1, \dots, B \cdot K \}$ where
$B \cdot K \geq N$, we
first draw a uniformly random permutation $g$ of the data universe. Then, we
hash the numbers $n \in [N]$ into $bin[n]$ by applying the hash function $g$ and
then taking the result modulo $B$. Then, we record whether the item landed in a
specific bucket $Z$ by computing the indicator $hitZ[n] = [bin[n] = Z]$, which
is $1$ if $bin[n] = Z$ and $0$ otherwise, and accumulate the result into the
count $ct$.

Our goal is to show that $ct$ is usually not far from its expected value, which
is $N/B$. If the quantities $\{[bin[n] = Z]\}_{n}$ were independent, we would be
able to apply a standard concentration bound to the sum $ct$. However,
$\{ bin[n] = Z \}_n$ are \emph{not} independent: for instance, since exactly
$K$ elements from $[B \cdot K]$ map to $Z$, if $bin[n] = Z$ for $n \in \{ 0,
1, \dots, K - 1 \}$, then $bin[K] = Z$ must be false.

Nevertheless, we can show that $\{ [bin[n] = Z] \}_n$ are negatively associated
random variables. Intuitively, $\{ g[n] \}_n$ are NA random variables because
the result of a uniformly random permutation is NA. Then, $\{ bin[n] \}_n$ is
computed by mapping the function $mod(-, B)$ over the array $g$; since this
produces another uniform permutation distribution, the vector $\{ bin[n] \}_n$
is also NA. By similar reasoning $\{ [bin[n] = Z] \}_n$ is also NA, as it is
obtained by mapping the function $[- \mathrel{=} Z]$ over $\{ bin[n] \}_n$.
Since this example is similar to the first Bloom filter example, except applying
the negative association of the permutation distribution~\eqref{ax:perm-pna} and
the permutation map axiom~\eqref{ax:perm-map}, we defer the details to the
appendix.

\subsection{Fully-dynamic dictionary}

\begin{figure*}
  \begin{subfigure}[b]{0.48\textwidth}
    \[
      \begin{array}{l}
        \textsc{FDDict}: \\
        \quad\Assn{binCt}{zero(C,P)}; \\
        \quad\Assn{n}{0}; \\
        \quad\DWhile{n < N}{} \\
        \quad\quad \Rand{crate[n]}{oh([C])}; \\
        \quad\quad \Rand{pocket[n]}{oh([P])}; \\
        \quad\quad \Assn{bin[n]}{crate[n]^\top \cdot pocket[n]}; \\
        \quad\quad \Assn{p}{0}; \\
        \quad\quad \DWhile{p < P}{} \\
        \quad\quad\quad \Assn{c}{0}; \\
        \quad\quad\quad \DWhile{c < C}{} \\
        \quad\quad\quad\quad \Assn{upd}{binCt[c][p] + bin[n][c][p]}; \\
        \quad\quad\quad\quad \Assn{binCt[c][p]}{upd}; \\
        \quad\quad\quad\quad \Assn{c}{c + 1}; \\
        \quad\quad\quad \Assn{p}{p + 1}; \\
        \quad\quad \Assn{n}{n + 1}; \\
        \quad\Assn{p}{0}; \\
        \quad\DWhile{p < P}{} \\
        \quad\quad \Assn{c}{0}; \\
        \quad\quad \DWhile{c < C}{} \\
        \quad\quad\quad \Assn{over[c][p]}{[binCt[c][p] > T_{bin}]}; \\
        \quad\quad\quad \Assn{upd}{overCt[c] + over[c][p]}; \\
        \quad\quad\quad \Assn{overCt[c]}{upd}; \\
        \quad\quad\quad \Assn{c}{c + 1}; \\
        \quad\quad \Assn{p}{p + 1}
      \end{array}
    \]
    \caption{Fully-dynamic dictionary~\citep{DBLP:journals/corr/abs-1911-05060}}
    \label{fig:ex:dictionary}
  \end{subfigure}
  \hfill
  \begin{subfigure}[b]{0.48\textwidth}
    \[
      \begin{array}{l}
        \textsc{RepeatBIB}: \\
        \quad\Assn{r}{0}; \\
        \quad\DWhile{r < R}{} \\
        \quad\quad \Assn{n}{0} \\
        \quad\quad \Assn{rem}{0}; \\
        \quad\quad \DWhile{n < N}{} \\
        \quad\quad\quad \Assn{ct[n]}{ct[n] - [ct[n] > 0]}; \\
        \quad\quad\quad \Assn{rem}{rem + [ct[n] > 0]}; \\
        \quad\quad\quad \Assn{n}{n + 1}; \\
        \quad\quad \Assn{j}{0}; \\
        \quad\quad \RWhile{j < rem}{} \\
        \quad\quad\quad \Rand{bin[j]}{oh([N])}; \\
        \quad\quad\quad \Assn{k}{0}; \\
        \quad\quad\quad \RWhile{k < N}{} \\
        \quad\quad\quad\quad \Assn{upd}{ct[k] + bin[j][k]}; \\
        \quad\quad\quad\quad \Assn{ct[k]}{upd}; \\
        \quad\quad\quad\quad \Assn{k}{k + 1}; \\
        \quad\quad\quad \Assn{j}{j + 1}; \\
        \quad\quad \Assn{n}{0}; \\
        \quad\quad \Assn{emptyCt[r]}{0}; \\
        \quad\quad \Assn{empty}{isZero(ct)}; \\
        \quad\quad \DWhile{n < N}{} \\
        \quad\quad\quad \Assn{upd}{emptyCt[r] + empty[n]}; \\
        \quad\quad\quad \Assn{emptyCt[r]}{upd}; \\
        \quad\quad\quad \Assn{n}{n + 1}; \\
        \quad\quad \Assn{r}{r + 1};
      \end{array}
    \]
    \caption{Repeated balls-into-bins~\citep{DBLP:journals/dc/BecchettiCNPP19}}
    \label{fig:ex:repeat:bib}
  \end{subfigure}

  \caption{Larger examples}
\end{figure*}

For our next example, we consider a hashing scheme for a
fully-dynamic dictionary, a space-efficient data structure that supports
insertions, deletions, and membership queries. The top level of the
data structure by \citet{DBLP:journals/corr/abs-1911-05060} uses a two-level
hashing scheme: elements are first hashed into a \emph{crate}, and then hashed
into a \emph{pocket dictionary} within each crate. As part of the space
analysis of their scheme, \citet{DBLP:journals/corr/abs-1911-05060} proves a
high-probability bound on the number of pocket dictionaries that overflow after
a given number of elements are inserted.

We extract the program \textsc{FDDict} in \Cref{fig:ex:dictionary} from the
scheme in \citet{DBLP:journals/corr/abs-1911-05060}. The program models the
insertion of $N$ elements. Each element is first hashed into one of $C$ possible
crates uniformly at random, and then hashed into one of $P$ possible pocket
dictionaries uniformly at random.  The variable $bin[n]$ is a $C$ by $P$ matrix,
with all entries zero except for the entry at $(crate[n], pocket[n])$, which is
set to $1$. Next, the program totals up the number of elements hashing to each
(crate, pocket) pair, storing the result in the $C$ by $P$ matrix $binCt$.
Finally, the program checks which (crate, pocket) pairs have count larger than
some concrete threshold $T_{bin}$ ($over$), and totals up the number of full
pocket dictionaries in each crate ($overCt$).

Our logic can prove a judgment of the following form:
\[
  \bigg\{\top\bigg\}
  {\textsc{FDDict}}
  \bigg\{\bigwedge_{\gamma = 0}^C \Pr [ overCt[\gamma] > P \cdot \rho_{bin} + T(\rho_{over}, P) ] \leq \rho_{over}\bigg\} ,
\]
where the logical variables $\rho_{bin}$ and $\rho_{over}$ represents the
parametric overflow properties.  This formalizes a result similar to
\citet[Claim 21]{DBLP:journals/corr/abs-1911-05060}, which states that except
with probability $\beta$, all crates have at most $T_{over}$ overfull pocket
dictionaries. The core of the proof shows that for every crate index $\gamma$,
the counts $binCt[\gamma][\beta]$ are negatively associated, using the
\textsc{NegFrame} rule as in the array version of the Bloom filter example.
Then, we show that vector $over[\gamma][\beta]$, which indicates whether each
pocket dictionary $\beta$ in crate $\gamma$ is overfull or not, is also
negatively associated. This holds because $over[\gamma][\beta]$ is obtained from
$binCt[\gamma][\beta]$ by applying a monotone function. Furthermore, the count
of overflows $overCt[\gamma]$ is obtained by another monotone function on
$over[\gamma][\beta]$ and thus its entries are also negatively associated.

\subsection{Repeated balls-into-bins process}
\label{sec:ex_repeated_bib}

Our final example considers a probabilistic protocol proposed by
\citet{DBLP:journals/dc/BecchettiCNPP19}, implemented as \textsc{RepeatBIB} in
\Cref{fig:ex:repeat:bib}. Intuitively, the program implements a repeated
balls-into-bins process. Initially, $N$ balls are distributed among $N$ bins
($ct[n]$). For $R$ rounds, in each round a ball is first removed from every
non-empty bin. Then, the $rem$ removed balls are randomly reassigned to bins.
This process is useful for distributed protocols and scheduling algorithms,
where the balls represent tasks and the bins represent computation nodes.
\citet{DBLP:journals/dc/BecchettiCNPP19} proposed and analyzed this algorithm
(e.g., bounding the maximum load, proving how long it takes for all balls to
visit all bins). We can verify the following lower-bound on the number of empty
bins, analogous to \citet[Lemma 1 and Lemma 2]{DBLP:journals/dc/BecchettiCNPP19}:
\[
  \Bigg\{N \geq 2 \land \apEq{\sum_{\alpha = 0}^N ct[\alpha]}{N} \Bigg\}
  {\textsc{RepeatBIB}}
  \Bigg\{
    \Pr\left[ \bigvee_{\beta = 0}^R (emptyCt[\beta] < N/15 - T(\rho_{empty}, N)) \right] \leq R \cdot \rho_{empty}
  \Bigg\}
\]

Two aspects of this program make it more difficult to verify.  First, there is
a loop with a randomized guard: the number of removed balls $rem$ is randomized
quantity. Reasoning about such loops is challenging, because our \textsc{Loop}
rule is not directly applicable and only far weaker rules are available for
loops with general randomized guards.  \citet{DBLP:journals/dc/BecchettiCNPP19}
sidestep this problem by \emph{conditioning} on the number of balls in each
bin, which also fixes $rem$ to be some value, proving the target property for
every fixed setting, and then combining the proofs together. \programlogic can
formalize this style of reasoning using the randomized case analysis rule
(\textsc{RCase}) to condition on $rem$'s value, and then apply the
\textsc{Loop} rule; however, the post-condition of \textsc{RCase} must be
closed under mixtures (CM), while independence and negative association are
known \emph{not} to satisfy this side-condition.  Thus, it is not possible to
prove negative association by first conditioning and then combining. To work
around this second problem, we use a technique from
\citet{DBLP:journals/dc/BecchettiCNPP19} and prove, on each conditional
distribution, a high-probability bound using the Chernoff bound. The
benefit of this approach is that high-probability bounds \emph{are} CM, so we
can apply \textsc{RCase} to combine the results. In our view, the fact that
\programlogic can handle this kind of subtle argument involving conditioning is
a strength of our approach.

\section{Related work}
\label{sec:rw}

\paragraph*{Bunched implications.}
The logic of bunched implications (BI)~\citep{OhP99,PymMono} is a well-studied
substructural logic. BI has a \emph{resource semantics}~\citep{POhY04}, where
states are resources and the separating conjunction combines compatible
resources together. We follow Docherty's uniform presentation and
investigation of BI~\citep{docherty:thesis}; in particular, our negative
association model relies on Docherty's non-deterministic frame conditions, and
we use his duality-theoretic framework to establish \LOGIC's metatheory.

\paragraph*{Separation logics.}
The first separation logic was developed to verify pointer-manipulating
programs~\citep{Reynolds01, IOh01, OhRY01}. There is long line of work on
separation logic for concurrency, starting
from~\citep{DBLP:journals/tcs/OHearn07, DBLP:journals/tcs/Brookes07} and
continuing to the present day~(e.g., \citep{DBLP:conf/pldi/SergeyNB15,
DBLP:journals/jfp/JungKJBBD18}).

More recently, separation logics have been
developed for probabilistic programs. \programlogic is an extension of
PSL~\citep{PSL}, a separation logic for probabilistic independence.
\citet{BDHS20} propose DIBI, an extension of BI with a non-commutative
conjunction, and developed a program logic with DIBI assertions that is capable
of proving conditional independence.  \citet{DBLP:journals/pacmpl/BatzKKMN19}
propose QSL, a separation logic where assertions have a \emph{quantitative}
interpretation, and used their logic to verify probabilistic and
heap-manipulating programs. \citet{TassarottiH19} develop a separation logic for
relational reasoning about probabilistic programs, using the coupling approach
of pRHL~\citep{BartheGB12}.

\paragraph*{Verifying approximate data structures and applying concentration bounds.}
Bloom filters are a data structure supporting \emph{approximate membership
queries} (AMQs). Ceramist~\citep{DBLP:conf/cav/GopinathanS20} is a recent
framework for verifying hash-based AMQ structures in the Coq theorem prover.
Besides handling Bloom filters, Ceramist supports subtle proofs of correctness
for many other AMQs. Compared with our approach, Ceramist proofs are more
precise but also more intricate, applying theorems about Stirling numbers to
achieve a precise bound on the false positive probability.  In contrast, our
approach reasons about negative dependence to achieve a substantially simpler
proof, albeit with less precise bounds.

Prior works in verification have also applied the Chernoff bound to bound sums
of independent random quantities~(e.g., \citep{DBLP:conf/pldi/WangS0CG21,
Chakarov-martingale}). While independence is easier to establish, the negative
association property that we need is more subtle.

\paragraph*{Negative dependence.}
There are multiple definitions of negative dependence in the literature, each
with their own strengths and weaknesses. We work with negative association
(NA)~\citep{joag1983negative,dubhashi-ranjan}, because it holds in many
situations where negative dependence should hold and it is closed under various
notions of composition. Recently, the notion of Strong Rayleigh (SR)~\citep{SR}
distribution has been proposed as an ideal definition of negative dependence.
The SR condition satisfies more closure properties than NA does; in particular,
it is preserved under various forms of conditioning. However, SR distributions
have mostly been studied for Boolean variables only, and we do not know if an
analogue of the monotone maps property of NA holds for SR.

Beyond theoretical investigations, negative dependence plays a useful role in
many practical applications. In machine learning, negative dependence can help
ensure diversity in predictions by a
model~\citep{DBLP:journals/ftml/KuleszaT12}, and fast algorithms are known to
learn and sample from negatively-dependent
distributions~\citep{DBLP:conf/colt/AnariGR16}. In algorithm design, negative
dependence is a useful tool to randomly round solutions of linear programs to
integral solutions~\citep{DBLP:conf/focs/Srinivasan01}. Negative dependence can
ensure that certain constraints are satisfied exactly after rounding, while
still allowing concentration bounds to be applied to analyze the quality of the
rounded solution.

\section{Conclusion and future direction}
\label{sec:conc}

We introduced \programlogic, a probabilistic program logic that can reason about
independence and negative association. Assertions in \programlogic are based on
a novel probabilistic model of \LOGIC, an extension of the logic of Bunched
Implications with multiple separating conjunctions. We demonstrated how to use
\programlogic to reason about probabilistic hashing schemes, and a repeated
balls-into-bins process. There are several natural directions for future work.

\paragraph*{Other models of \LOGIC, and non-deterministic frames.}
The assertion logic \LOGIC was primarily motivated by our NA model, but it is
general enough that we believe there are likely other natural models. Exploring
these directions could allow modeling finer notions of separation, and could
further justify \LOGIC as an interesting logic in its own right. It would also
be interesting to see if there are other models that use a non-deterministic
operator to combine resources, as proposed by~\citet{docherty:thesis}.

\paragraph*{Verifying negative association for sampling algorithms.}
We used NA to analyze probabilistic hashing schemes. Another classical
application of NA is in \emph{sampling schemes}, which generate a sample from a
target distribution while satisfying certain
constraints~\citep{DBLP:journals/cpc/DubhashiJR07, branden-jonasson}. NA
samplers are useful in algorithm design~\citep{DBLP:conf/focs/Srinivasan01} and
statistics, and it would be interesting to understand how to verify these
programs. Many samplers employ rejection sampling, which is not easily analyzed
in \programlogic but which could be expressed with an explicit conditioning
operator, as in probabilistic programming
languages~\citep{10.1145/2593882.2593900}.

\begin{acks}                            
  We thank the anonymous reviewers for their helpful feedback and suggestions.
  This work benefited from discussions with Simon Docherty. This work was
  supported in part by the
  \grantsponsor{GS100000001}{NSF}{http://dx.doi.org/10.13039/100000001} under
  Grant No.~\grantnum{GS100000001}{2035314},~\grantnum{GS100000001}{1943130},~\grantnum{GS100000001}{2040249},~\grantnum{GS100000001}{2040222}
  and~\grantnum{GS100000001}{2152831}.
\end{acks}

\bibliography{header.bib,refs.bib}

\iffull
\appendix
\section{Preliminaries}
\label{sec:app:prelim}

\begin{lemma}
	\label{lemma:pnaeq_indep}
	Say $\mathcal{S} = \{ S_i \mid 1 \leq i \leq N\}$ where $S_i$ are disjoint,
	$S = \cup \mathcal{S}$ and $\mu \in \Mem[S]$,

	Then, $S_i$ are independent in $\mu$
	if and only if for any family of all monotone or all antitone functions
	$f_i: \Mem[S_i] \to \mathbb{R}^+$,
	\begin{align}
	\EE_{x \sim \mu}\left[\prod_{S_i \in \mathcal{S}} f_i (\p_{S_i} x) \right] =  \prod_{S_i \in \mathcal{S}} \EE_{x \in \mu}[f_i (\p_{S_i} x) ].
		\label{eq:gpnaindep}
	\end{align}
\end{lemma}

\begin{proof}
	The forward direction is straightforward.
	The backward direction needs more careful analysis.
	In general, zero correlation does not imply independence,
	but here, we have the equality for all family of monotone or antitone functions,
	so that suffices for independence.

	We prove by induction on $\mathcal{T} = \{ S_i \mid 1 \leq i \leq K\}$ that for any family of $v_i \in \Mem[S_i]$,
	\begin{align} \label{eq:gpnainduction}
		\mathbb{E}_{x \in \mu} \left[ \left( \bigwedge_{S_i \in \mathcal{T}} \p_{S_i} x = v_i \right) \land \left(\bigwedge_{S_i \in \mathcal{S} \setminus \mathcal{T}} \p_{S_i} x < v_i \right) \right]
		= \prod_{S_i \in \mathcal{T}} \mathbb{E}_{x \in \mu} \left[\p_{S_i} x = v_i\right] \cdot \prod_{S_i \in \mathcal{S} \setminus \mathcal{T}} \mathbb{E}_{x \in \mu}\left[\p_{S_i} x < v_i \right].
	\end{align}
	\begin{description}
		\item [Case $|\mathcal{T}| = 1$:] Say $\mathcal{T} = \{ S_j \}$.
			Since indicator functions $S_i < v_i$ and $S_i \leq v_i$
			are both monotonically decreasing,
			\begin{align*}
			&\mathbb{E}_{x \in \mu} \left[\p_{S_j} x = v_j \land (\bigwedge_{S_i \in \mathcal{S} \setminus \mathcal{T} } \p_{S_i} x < v_i) \right]\\
			&=	\mathbb{E}_{x \in \mu} \left[ \p_{S_j} x \leq v_j \land (\bigwedge_{S_i \in \mathcal{S} \setminus \mathcal{T} } \p_{S_i} x < v_i) \right]
			- \mathbb{E}_{x \in \mu} \left[ \p_{S_j} x < v_j \land (\bigwedge_{S_i \in \mathcal{S} \setminus \mathcal{T} } \p_{S_i} x < v_i) \right]\\
			&=	\mathbb{E}_{x \in \mu} \left[\p_{S_j} x \leq v_j\right] \cdot \prod_{S_i \in \mathcal{S} \setminus \mathcal{T} } \mathbb{E}_{x \in \mu} \left[ \p_{S_i} x < v_i\right]
			- \mathbb{E}_{x \in \mu} \left[ \p_{S_j} x < v_j \right] \cdot \prod_{S_i \in \mathcal{S} \setminus \mathcal{T} } \mathbb{E}_{x \in \mu} \left[ \p_{S_i} x < v_i \right] \tag{By~\Cref{eq:gpnaindep}}\\
			&= 	(\mathbb{E}_{x \in \mu} \left[ \p_{S_j} x  \leq v_j \right] 	- \mathbb{E}_{x \in \mu} \left[ \p_{S_j} x < v_j \right]) \cdot \prod_{S_i \in \mathcal{S} \setminus \mathcal{T} } \mathbb{E}_{x \in \mu} \left[\p_{S_i} x < v_i \right]\\
			&= 	\mathbb{E}_{x \in \mu} \left[ \p_{S_j} x = v_j\right] \cdot \prod_{S_i \in \mathcal{S} \setminus \mathcal{T} } \mathbb{E}_{x \in \mu} \left[\p_{S_i} x < v_i\right]
			\end{align*}
		\item  [Case $|\mathcal{T}| > 1$] Let $S_j$ be an element in $\mathcal{T}$.
			\begin{align*}
&\mathbb{E}_{x \in \mu} \left[ (\bigwedge_{S_i \in \mathcal{T}} \p_{S_i} x = v_i ) \land (\bigwedge_{S_i \in \mathcal{S} \setminus \mathcal{T}} \p_{S_i} x < v_i) \right] \\
&= \mathbb{E}_{x \in \mu} \left[ \p_{S_j} x \leq v_j \land (\bigwedge_{S_i \in \mathcal{T} \setminus \{S_j\} } \p_{S_i} x = v_i) \land (\bigwedge_{S_i \in \mathcal{S} \setminus \mathcal{T}} \p_{S_i} x < v_i) \right] \\
&- \mathbb{E}_{x \in \mu} \left[ \p_{S_j} x < v_j \land (\bigwedge_{S_i \in \mathcal{T} \setminus \{S_j\} } \p_{S_i} x = v_i ) \land (\bigwedge_{S_i \in \mathcal{S} \setminus \mathcal{T}} \p_{S_i} x < v_i) \right]\\
&= \mathbb{E}_{x \in \mu} \left[ \p_{S_j} x \leq v_j\right] \cdot \prod_{S_i \in \mathcal{T} \setminus \{S_j\} } \mathbb{E}_{x \in \mu} \left[\p_{S_i} x = v_i\right] \cdot \prod_{S_i \in \mathcal{S} \setminus \mathcal{T}} \mathbb{E}_{x \in \mu} \left[\p_{S_i} x < v_i \right] \\
&- \mathbb{E}_{x \in \mu} \left[\p_{S_j} x < v_j\right] \cdot \prod_{S_i \in \mathcal{T} \setminus \{S_j\} } \mathbb{E}_{x \in \mu} \left[\p_{S_i} x = v_i\right] \cdot \prod_{S_i \in \mathcal{S} \setminus \mathcal{T}} \mathbb{E}_{x \in \mu} \left[\p_{S_i} x < v_i \right] \\
&=  \mathbb{E}_{x \in \mu} \left[ \p_{S_j} x = v_j\right] \cdot \prod_{S_i \in \mathcal{T} \setminus \{S_j\} } \mathbb{E}_{x \in \mu} \left[\p_{S_i} x = v_i\right] \cdot \prod_{S_i \in \mathcal{S} \setminus \mathcal{T}} \mathbb{E}_{x \in \mu} \left[\p_{S_i} x < v_i \right]
			\end{align*}
	\end{description}
	When $\mathcal{T} = \mathcal{S}$,~\Cref{eq:gpnainduction} implies
	\begin{align*}
		\mathbb{E}_{x \in \mu} \left[ \bigwedge_{S_i \in \mathcal{S}} \p_{S_i} x = v_i\right]
	&=
	\prod_{S_i \in \mathcal{S}} 	\mathbb{E}_{x \in \mu} \left[ \p_{S_i} x = v_i \right]
	\end{align*}
	for any $v_i$'s. Thus, components in $\mathcal{S}$ are independent.
\end{proof}

\subsection{Coarsening}
We prove some properties of coarsening.
In the following we will use an alternative definition of coarsening,
which will be shown to be equivalent to what we define in the main text.
\begin{definition}[Alternative definition of coarsening]
	\label{def:partition_coarsening}
We first index any partition $\mathcal{S}$ as $\mathcal{S}_1, \dots, \mathcal{S}_{|\mathcal{S}|}$.
	Say $|\mathcal{S}'| = m$, $|\mathcal{S}| = n$.
We say	 $\mathcal{S}'$ coarsens a partition $\mathcal{S}$
	there exists a function a $f: [m] \to \mathcal{P}([n])$ such that
	1) $\cup_{i \in [m]} f(i) = [n] $;
	2) for any $i, j \in [m]$, either $i = j$ or $f(i), f(j)$ are disjoint;
	3) $\mathcal{S}' = \{ \cup \{ \mathcal{S}_j \mid j \in f(i)  \} \mid i \in [m] \}$.
\end{definition}

\begin{lemma}
	\label{lemma:coarsening_eq}
  Let $\mathcal{S}$, $\mathcal{S}'$ be two partitions. Then $\mathcal{S}'$
  coarsens $\mathcal{S}$ according to~\Cref{def:partition_coarsening} if and
  only if $\mathcal{S}'$ coarsens $\mathcal{S}$ according to~\Cref{def:PNA} .
\end{lemma}

\begin{proof}
	We index  $\mathcal{S}$ as $\mathcal{S}_1, \dots, \mathcal{S}_{n}$ and
	$\mathcal{S}'$ as $\mathcal{S}'_1, \dots, \mathcal{S}'_{|m|}$ .

	Backward direction:
 By that definition, we know
	a) for any $\mathcal{S}'_i \in \mathcal{S}'$, $\mathcal{S}'_i = \cup \mathcal{R}$ for some $\mathcal{R} \subseteq \mathcal{S}$;
	b) $\cup \mathcal{S} = \cup \mathcal{S}'$.

	We define the function $g: [m] \to \mathcal{P}([n])$ as
	$g(i) = \{j \mid \mathcal{S}_j \subseteq \mathcal{S}'_i\}$.
	This $g$ would satisfies all the conditions required:
	\begin{enumerate}
		\item By substitution, $\cup_{i \in [m]} g(i) =  \cup_{i \in [m]} \{j \mid \mathcal{S}_j \subseteq \mathcal{S}'_i\} = \cup_{s' \in \mathcal{S}'}\{j \mid \mathcal{S}_j \subseteq s'\} $.
			By b), for any $j \in [n]$, $\mathcal{S}_j \subseteq \cup \mathcal{S}'$.
			Then by a) and that $\mathcal{S}$ is a partition,	if $s'$ covers any of $\mathcal{S}_j$, it must
			covers all of $\mathcal{S}_j$,
then $\mathcal{S}_j \subseteq \cup \mathcal{S}'$ implies there exists $s' \in \mathcal{S}'$ such that $\mathcal{S}_j \subseteq s'$.
			Thus, $j \in  \{j \mid \mathcal{S}_j \subseteq s'\} \subseteq \cup_{s' \in \mathcal{S}'}\{j \mid \mathcal{S}_j \subseteq s'\} $.
			For any $j \not\in [n]$, $\mathcal{S}_j$ is undefined,
			so it is impossible that $\mathcal{S}_j \subseteq s'$ for some $s' \subseteq \mathcal{S}'$.
			Therefore, $\cup_{s' \in \mathcal{S}'}\{j \mid \mathcal{S}_j \subseteq s'\} = [n]$.
		\item
			For any $k \in g(i)$, $\mathcal{S}_k \subseteq \mathcal{S}'_i$.
			If $i \neq j$, then $\mathcal{S}'_i $ and $\mathcal{S}'_j$ are disjoint since $\mathcal{S}'$ is a partition.
			Thus, $\mathcal{S}_k \not\subseteq \mathcal{S}'_j$, and $k \not\in g(j)$.
			So for any $i \neq j$, $g(i), g(j)$ are disjoint.
		\item  By substitution,
	\begin{align*}
		\{ \cup \{ \mathcal{S}_j \mid j \in g(i)  \} \mid i \in [m] \}
			&=  \{ \cup \{ \mathcal{S}_j \mid \mathcal{S}_j \subseteq \mathcal{S}'_i  \} \mid i \in [m] \}
			&=  \{ \cup \{ \mathcal{S}_j \mid \mathcal{S}_j \subseteq s' \} \mid s' \in \mathcal{S}' \}
      .
	\end{align*}
	Again, by a) and that $\mathcal{S}$ is a partition, if $s' \in \mathcal{S}$ covers any part of of $\mathcal{S}_j$, it must
			covers all of $\mathcal{S}_j$, so  $\cup \{ \mathcal{S}_j \mid \mathcal{S}_j \subseteq s' \}  = s'$.
			Thus, $\{ \cup \{ \mathcal{S}_j \mid \mathcal{S}_j \subseteq s' \} \mid s' \in \mathcal{S}' \} =\mathcal{S}'$.
	\end{enumerate}

	Forward direction:
	By 3), we know that  $\mathcal{S}' = \{ \cup \{ \mathcal{S}_j \mid j \in f(i)  \} \mid i \in [m] \}$.
	So for any $\mathcal{S}'_i \in \mathcal{S}'$, we have $s' = \cup \{ \mathcal{S}_j \mid j \in f(i)  \} $, which is a subset of $\mathcal{S}'$ by construction. So we proved a).
	Also, $\cup \mathcal{S}' = \cup \{ \cup \{ \mathcal{S}_j \mid j \in f(i)  \} \mid i \in [m] \}
	= \cup \{ \mathcal{S}_j \mid j \in f(i)  \mid i \in [m] \} $,
	and by 1), that is equivalent to
	$\cup \{ \mathcal{S}_j \mid j \in [n] \}  $, which is equivalent to $\cup \mathcal{S}$.

\end{proof}

We can prove that coarsening commute with projections.
\begin{lemma}
	\label{lemma:coarsen_project}
	Given a partition $\mathcal{S} = \{ \mathcal{S}_i \}_i$ and a set $X$,
	let $\mathcal{S}_X = \{\mathcal{S}_i \cap X \mid \mathcal{S}_i \in \mathcal{S}\}$.
	For any $\mathcal{T}$ coarsening $\mathcal{S}_X$,
	there exists a coarsening $\mathcal{S}'$ of $\mathcal{S}$ such that
	$\mathcal{T} = \{\mathcal{S}_i \cap X \mid \mathcal{S}_i \in \mathcal{S}'\}$;
	conversely, for any $\mathcal{S}'$ coarsening $\mathcal{S}$,
	and $\mathcal{S}'_X = \{\mathcal{S}_i \cap X \mid \mathcal{S}_i \in \mathcal{S}'\}$,
  we have $\mathcal{S}'_X$ coarsens $\mathcal{S}_X$.

\end{lemma}
\begin{proof}
	Forward direction:
	By~\Cref{def:partition_coarsening}, there exists a coarsening function $f$ such that
	\begin{align*}
		\mathcal{T} &= \{ \cup \{ {(\mathcal{S}_X)}_j \mid j \in f(i)  \} \mid i \in [|\mathcal{T}|] \}\\
														&= \{ \cup \{ \mathcal{S}_j \cap X \mid j \in f(i)\} \mid i \in [|\mathcal{T}|] \}\\
														&= \{ (\cup \{ \mathcal{S}_j \mid j \in f(i)\}) \cap X \mid i \in [|\mathcal{T}|] \}\\
														&= \{ S' \cap X \mid S' \in \mathcal{S}' \} \tag{where $\mathcal{S}' = \{\cup \{ \mathcal{S}_j \mid j \in f(i)\} \mid i \in [|\mathcal{T}|]\}$}
	\end{align*}
	$\mathcal{S}'$ has the same size as $\mathcal{T}$, so
	$\mathcal{S}' = \{\cup \{ \mathcal{S}_j \mid j \in f(i)\} \mid i \in [|\mathcal{S}'|]\}$,
	and thus $\mathcal{S}'$ coarsens $\mathcal{S}$.

Backward direction:
$\mathcal{S}'$ coarsens $\mathcal{S}$, so there exists a coarsening function $f$ such that
\begin{align*}
	\mathcal{S}' &= \{ \cup \{ S_j \mid j \in f(i)\} \mid i \in [|\mathcal{S}'|] \} .
\end{align*}
Thus,
\begin{align*}
	\mathcal{S}'_X &= \{(\cup \{ S_j \mid j \in f(i)\}) \cap X \mid i \in [|\mathcal{S}'|] \} \\
	&= \{\cup \{ S_j \cap X \mid j \in f(i)\} \mid i \in [|\mathcal{S}'|] \} \\
	&= \{\cup \{ {S_X}_j \mid j \in f(i)\} \mid i \in [|\mathcal{S}'|] \} .
\end{align*}
Therefore, $\mathcal{S}'_X$ coarsens $\mathcal{S}_X$.

\end{proof}

\section{The logic \LOGIC}
\label{sec:app:logic}

\Cref{fig:hilbert} gives a Hilbert-style proof system for $\idxset$-BI. If
one erases the subscripts on $\sepand$ and $\sepimp$, then the rules without the
last three form a Hilbert proof system for BI.

\begin{figure}
\begin{mathpar}
	\inferrule*
	{~}
	{P \vdash P}
%

	\inferrule* [right = $\top$]
	{~}
	{P \vdash \top}

	\inferrule* [right = $\bot$]
	{~}
	{\bot \vdash P}

	\inferrule* [right=$\lor$-E]
	{P \vdash R \\ Q \vdash R}
	{P \lor Q \vdash R}

	\inferrule* [right=$\lor$-I]
	{P \vdash Q_i}
	{P \vdash Q_1 \lor Q_2}

	\inferrule* [right=$\land$-I1]
	{P \vdash Q \\ P \vdash R}
	{P \vdash Q \land R}

	\inferrule* [right=$\land$-I2]
	{Q \vdash R}
	{P \land Q \vdash R}

	\inferrule* [right=$\land$-E]
	{P \vdash Q_1 \land Q_2}
	{P \vdash Q_i}

	\inferrule* [right=$\rightarrow$-I]
	{P \land Q \vdash R}
	{P \vdash Q \rightarrow R}

 	\inferrule* [right=$\rightarrow$-E]
 	{P \vdash Q \rightarrow R \\ P \vdash Q}
 	{P \vdash R}

 	\inferrule* [right=$\sepimp$]
 	{P \sepand_\idx Q \vdash R}
 	{P \vdash Q \sepimp_\idx R}

 	\inferrule* [right=$\sepimp$-E]
 	{ P \vdash Q \sepimp_\idx R \and S \vdash Q }
 	{ P \sepand_\idx S \vdash R }

 	\inferrule* [right=$\sepand$-\textsc{Unit}]
 	{~}
 	{P \dashv \vdash P \sepand_\idx \sepid_{\idx}}

 	\inferrule* [right=$\sepand$-\textsc{Conj}]
 	{P \vdash R \\ Q \vdash S}
 	{P \sepand_\idx Q \vdash R \sepand_\idx S}

 	\inferrule* [right=$\sepand$-\textsc{Comm}]
 	{~}
 	{P \sepand_\idx Q \vdash Q \sepand_\idx P}

 	\inferrule* [right= $\sepand$-\textsc{Assoc}]
 	{~}
 	{(P \sepand_\idx Q) \sepand_\idx R \dashv\vdash P \sepand_\idx (Q \sepand_\idx R)}

        \inferrule* [right= $\sepand$-\textsc{Inclusion}]
 	{\idx_1 \leq \idx_2}
 	{P \sepand_{\idx_1} Q \vdash P \sepand_{\idx_2} Q}

%
\end{mathpar}
\caption{Hilbert system for $\idxset$-BI}
\label{fig:hilbert}
\end{figure}

Using the \idxset-BI rules, we can derive the rule $\textsc{Cut}$:
\begin{prooftree}\small
\AxiomC{$P \vdash Q$}
\AxiomC{$Q \vdash R$}
\RightLabel{$\land 2$}
\UnaryInfC{$P \land Q \vdash R$}
\RightLabel{$\rightarrow$}
\UnaryInfC{$P \vdash Q \rightarrow R$}
\RightLabel{\textsc{Cut}}
\BinaryInfC{$P \vdash R$}
\end{prooftree}

\InclusionOthers*

\begin{proof}
	\label{proof:InclusionOthers}
		Let $\idx_1 \leq \idx_2$. For the first rule, it suffices to show that $(P
		\sepimp_{\idx_2} Q) \sepand_{\idx_1} P \vdash Q$. By
		$\sepand$-\textsc{Weakening}, we have that $(P \sepimp_{\idx_2} Q)
		\sepand_{\idx_1} P \vdash (P \sepimp_{\idx_2} Q) \sepand_{\idx_2} P$ so the
		result follows from \textsc{Cut} and $\sepimp-\textsc{E}$.

		For the second rule, $\sepand_{1}$-\textsc{Unit} implies $ \sepid_{\idx_2} \vdash \sepid_{\idx_2} \sepand_{\idx_1}
		\sepid_{\idx_1}$;
		$\sepand$-\textsc{Weakening} implies $\sepid_{\idx_2} \sepand_{\idx_1}
		\sepid_{\idx_1} \vdash \sepid_{\idx_2} \sepand_{\idx_2} \sepid_{\idx_1}$;
		and $\sepand_{2}$-\textsc{Unit} implies $\sepid_{\idx_2} \sepand_{\idx_2} \sepid_{\idx_1} \vdash \sepid_{\idx_1}.$
		Then by \textsc{Cut}, we have $ \sepid_{\idx_2} \vdash \sepid_{\idx_1} $
\end{proof}

\section{The Model of Negative Dependence and Independence}
\label{sec:app:model}

%
%

\subsection{A BI model for negative association}
\PNAoplusinterpolate*
\begin{proof}
	Let $S$ denote $\dom(\mu_1)$ and $T$ denote $\dom(\mu_2)$.

	For any $\mu \in \mu_1 \oplus_s \mu_2$, we have
	$\pi_S \mu = \mu_1$,
	$\pi_T \mu = \mu_2$, and $\mu$ satisfies NA.
	$\mu$ being NA
implies $\mu$ is $\mathcal{R}$-\CPNA for any partition $\mathcal{R}$ on
$\dom(\mu)$
	So for any partition $\mathcal{S}$ on $S$, partition $\mathcal{T}$ on $T$,
	$\mu$ is $\mathcal{S} \cup \mathcal{T}$-\CPNA.
	Therefore, $\mu \in \mu_1 \oplus \mu_2$.

	For any $\mu \in \mu_1 \oplus \mu_2$,
		$\pi_S \mu = \mu_1$,
	$\pi_T \mu = \mu_2$,
	and $\mu$ is $\{S, T\}$-\CPNA since $\mu_1$ is $\{S\}$-\CPNA,
	$\mu_2$ is $\{T\}$-\CPNA.
	Thus, $\mu \in \mu_1 \oplus_w \mu_2$.
\end{proof}

\PNAisBI*
	\begin{proof}
		\label{proof:PNAisBI}
		We sketch the conditions, using the notation from the definition:
		\begin{description}
			\item[Down-Closed.]
				Let $dom(x) = S, dom(x') = S', dom(y) = T, dom(y') = T'$. We claim that we
				can take $z' = \pi_{S' \cup T'} z$. We evidently have $z \sqsupseteq z'$,
				and $\pi_{S'} z' = \pi_{S'} \pi_{S} z = x'$ and $\pi_{T'} z' = \pi_{T'}
				\pi_T z = y'$.

				What remains to show is that $z'$ is $\mathcal{S} \cup \mathcal{T}$-\CPNA
				for any $\mathcal{S}$, $\mathcal{T}$
				such that $x'$ is $\mathcal{S}$-\CPNA, $y'$ is $\mathcal{T}$-\CPNA,
				and $(\cup \mathcal{S}) \cap (\cup \mathcal{T}) = \emptyset$.

				If $x'$ is $\mathcal{S}$-\CPNA, then $x$ is $\mathcal{S}$-\CPNA;
				if $y'$ is $\mathcal{T}$-\CPNA, then $y$ is $\mathcal{T}$-\CPNA;
				then $z \in x \oplus y$ must be $\mathcal{S} \cup \mathcal{T}$-\CPNA.
				Since $z' := \pi_{S' \cup T'} z$, and $(\cup \mathcal{S}) \cup (\cup \mathcal{T}) \subseteq S' \cup T'$, we have $z'$ is $\mathcal{S} \cup \mathcal{T}$-\CPNA too.
				And evidently, $dom(z') = S' \cup T' = dom(x') \cup dom(y')$.
				So $z' \in x' \oplus y'$.

			\item[Commutativity.]
				Immediate.
			\item[Associativity.]
				Let $dom(x) = R, dom(y) = S, dom(z) = T$. We can assume that these sets
				are all disjoint, otherwise there is nothing to prove. We claim that we
				can take $s = \pi_{S \cup T} w$.
				For any $w$ in $t \oplus z$, $t \in x \oplus y$,
				we want to show that $w \in x \oplus s$ and $s \in y \oplus z$.
				\begin{itemize}
					\item
						For any partition $\mathcal{R}, \mathcal{S}$ such that
						$(\cup \mathcal{R})  \cap (\cup \mathcal{S}) = \emptyset$
						and $x$ is $\mathcal{R}$-\CPNA, $s$ is $\mathcal{S}$-\CPNA.
						For set $X \in \Var$,
						write $\{ Y \cap X \mid Y \in \mathcal{S}\}$ as $\mathcal{S}_X$.
						Then, by~\Cref{lemma:coarsen_project},
						$s$ is $\mathcal{S}$-\CPNA implies $y$ must be $\mathcal{S}_S$-\CPNA.
						Similarly,
						$s$ is $\mathcal{S}$-\CPNA implies $z$ must be $\mathcal{S}_T$-\CPNA.

						Then, $t \in x \oplus y$ must be $\mathcal{R} \cup (\mathcal{S}_S) $-\CPNA,
						and $w \in t \oplus z$ must be
						$\mathcal{R} \cup \mathcal{S}_S \cup \mathcal{S}_T$-\CPNA.
						Note that $\mathcal{S}$ coarsens $\mathcal{S}_S \cup \mathcal{S}_T$
						so $w$ is $\mathcal{R} \cup \mathcal{S}_S \cup \mathcal{S}_T$-\CPNA
						implies that $w$ is $\mathcal{R} \cup \mathcal{S} $-\CPNA.

						Also, $\pi_{R} w = \pi_R \pi_{R \cup S} w = \pi_R t = x$,
						and $dom(w) = R \cup S \cup T = dom(x) \cup dom(s)$.

						Hence, $w \in x \oplus s$.

					\item
						Note that $x$ is trivially $\{R\}$-\CPNA.
						Then, for any partition $\mathcal{S}, \mathcal{T}$ such that
						$R \cap (\cup \mathcal{S})  \cap (\cup \mathcal{T}) = \emptyset$
						and $y$ is $\mathcal{S}$-\CPNA and $z$ is $\mathcal{T}$-\CPNA,
						first $t$ must be $(\{R\} \cup \mathcal{S})$-\CPNA,
						and then $w$ must be $(\{R\} \cup \mathcal{S} \cup \mathcal{T})$-\CPNA.
						By projection, $s = \pi_{S \cup T}$ must be $\mathcal{S} \cup \mathcal{T} z$-\CPNA.

						Also, $\pi_{S} s = \pi_{S} \pi_{S \cup T} w = \pi_S w = \pi_S \pi_{R \cup S} w = \pi_S t = y$,
						and similarly, $\pi_T s = z$.
						Also, $dom(s) = S \cup T = dom(y) \cup dom(z)$.

						Hence, $s \in y \oplus z$.
				\end{itemize}

			\item[Unit Existence.]
				Take $e$ to be $\mu$ where $\mu$ is the (unique) distribution in $\DMem{\emptyset}$.
			\item[Unit Closure.]
				Immediate as we take $E = M$.
			\item[Unit Coherence.]
				Immediate: $x \in y \oplus e$ entails $y = \pi_{dom(y)} x$, which implies $y \sqsubseteq x$.
				\qedhere
		\end{description}
	\end{proof}

	\begin{theorem}
	Given a set of variables $S$, $S$ satisfies NA in $\mu$
	iff $\mu$ satisfies $\mathcal{S}$-\CPNA for any
$\mathcal{S}$ partitioning $S$ iff $\mu$  satisfies $\{ \{s\} \mid s \in S\}$-\CPNA.
\end{theorem}

\begin{proof}
The second equivalence is straightforward:
\begin{itemize}
	\item  $\{\{s\} \mid s \in \mathcal{S} \}$ is a partition of $S$, so we have the backward direction.
	\item Any $\mathcal{S}$
		partitioning $S$ coarsens $\{\{s\} \mid s \in \S \}$, so we have the first direction.
\end{itemize}

For the forward direction of the first equivalence,
it suffices to prove that for any partition $\mathcal{S}$ of $S$,
any family of all monotone or all antitone functions
		$f_i: \Mem[S_i] \to \mathbb{R}^+$
		\begin{align}
			\mathbb{E}_{x \sim \mu} \left[ \prod_{S_i \in \mathcal{S}} f_i(\p_{S_i} m) \right]
			&\leq  \prod_{S_i \in \mathcal{S}} \mathbb{E}_{x \sim \mu} \left[ f_i(\p_{S_i} m) \right].
		\end{align}
We prove that by induction on the size of $\mathcal{S}$.

\begin{description}
	\item [Base case $|\mathcal{S}| = 1$: ]$\mathcal{S}$-\CPNA is trivial.
	\item [Base case $|\mathcal{S}| = 2$: ] $\mathcal{S}$-\CPNA is straightforward from NA.
	\item [Inductive case: ]
		Assuming $\mu$ satisfies $\mathcal{S}$-\CPNA for any partition
		with size less than $K$,
		we want to show that $\mu$ satisfies $\mathcal{S}$-\CPNA for any partition
	 with size equals to $K$.

		Say $\mathcal{S} = \{S_1, \dots, S_K\}$.
		For any family of all monotone or all antitone functions
		$f_i: \Mem[S_i] \to \mathcal{R}^+$,
		either both $ m \mapsto \prod_{i = 1}^{K-1} f_i(\p_{S_i} m)$ and $f_K$ are monotone,
		or both $ m \mapsto \prod_{i = 1}^{K-1} f_i(\p_{S_i} m)$ and $f_K$ are antitone.
		Thus, by the inductive hypothesis
		\begin{align*}
			\mathbb{E}_{x \sim \mu} \left[ \prod_{i = 1}^K f_i(\p_{S_i} m) \right]
			&\leq	\mathbb{E}_{x \sim \mu} \left[ \prod_{i = 1}^{K - 1} f_i(\p_{S_i} m) \right]
			\cdot \mathbb{E}_{x \sim \mu} \left[ f_K(\p_{S_K} m)\right] \\
			&\leq  \prod_{i = 1}^{K - 1} \mathbb{E}_{x \sim \mu} \left[ f_i(\p_{S_i} m) \right]
			\cdot \mathbb{E}_{x \sim \mu} \left[ f_K(\p_{S_K} m)\right] \\
			&=  \prod_{i = 1}^{K} \mathbb{E}_{x \sim \mu} \left[ f_i(\p_{S_i} m) \right] .
		\end{align*}
\end{description}

The backward direction of the first equivalence is more involved.
For any two disjoint $A, B \subseteq S$, we know $\mu$ satisfies $\{A, B\}$-\CPNA,
so for every pair of both monotone or both antitone functions
$f: \Mem[A] \to \mathbb{R}^+$,
$g: \Mem[B] \to \mathbb{R}^+$,
	\[
		\mathbb{E}_{m \sim \mu} [ f(\p_A m) \cdot g(\p_B m) ]
		\leq \mathbb{E}_{m \sim \mu} [ f(\p_A m) ] \cdot \mathbb{E}_{m \sim \mu} [ g(\p_B m) ] .
	\]
But the problem is to show this inequality when $f, g$ are not both
non-negative. We prove that in three steps:
\begin{enumerate}
	\item If $f, g$ are lower-bounded by $-L$, i.e., $f(x) \geq  -L$
and $g(x) \geq  -L$ for any $x$. Then $x \to f(x) + L$ and
$x \to g(x) + L$ are both non-negative functions.
Thus,
	\begin{align}
 \label{eq:PNAeqNA_plusM}
		\mathbb{E}_{m \sim \mu} [ (f(\p_A m) + L) \cdot (g(\p_B m) + L) ]
		\leq \mathbb{E}_{m \sim \mu} [ f(\p_A m) + L] \cdot \mathbb{E}_{m \sim \mu} [ g(\p_B m) + L] .
	\end{align}
	Meanwhile,
	\begin{align*}
	\mathbb{E} [ (f(\p_A m) + L) \cdot (g(\p_B m) + L) ]
	&= \mathbb{E} [ f(\p_A m) \cdot g(\p_B m) ]
	+ L \cdot \mathbb{E} [ f(\p_A m)]
	+  L \cdot \mathbb{E} [ g(\p_B m)] + L^2 \\
	\mathbb{E} [ f(\p_A m) + L] \cdot \mathbb{E}[ g(\p_B m) + L]
	&=(\mathbb{E} [f(\p_A m)] + L) \cdot (\mathbb{E}[g(\p_B m)] + L) \\
	&= \mathbb{E} [ f(\p_A m) ] \cdot \mathbb{E} [g(\p_B m) ]
	+ L \cdot \mathbb{E} [ f(\p_A m)]
	+  L \cdot \mathbb{E} [ g(\p_B m)]  + L^2 .
	\end{align*}

	So~\Cref{eq:PNAeqNA_plusM} implies that
	\begin{align*}
\mathbb{E}_{m \sim \mu} [ f(\p_A m) \cdot g(\p_B m)]
		\leq \mathbb{E}_{m \sim \mu} [ f(\p_A m)] \cdot \mathbb{E}_{m \sim \mu} [ g(\p_B m)] .
	\end{align*}

\item If the codomain of $f$ or $g$ does not range across both negative and positive numbers,
	then we can also prove the desired inequality by applying the monotone convergence theorem
	on the result for lower-bounded functions.
	\begin{itemize}
		\item Say $f$ is non-negative and $g$ is non-positive.
			For any natural number $n$,  $m \in \Mem[A \cup B]$,
			we define $g_n(\p_B m) = \max(g(\p_B m), -n)$,
			$h_n(m) = f(\p_A m) \cdot g_n(\p_B m) $.
			Then for any $n$, $g_n$ and $h_n$ are lower-bounded non-positive functions;
			and for any $m$, $\{g_n(m)\}_{n \in \mathbb{N}}$ is a monotonically decreasing sequence
			converging to $g(m)$,
			$\{h_n(m)\}_{n \in \mathbb{N}}$ is a monotonically decreasing sequence
			converging to $f(\p_A m) \cdot g(\p_B m) $.
			By the monotone convergence theorem,
			\begin{align*}
				\mathbb{E}_{m \sim \mu} f(\p_A m) \cdot g(\p_B m) &= \lim_{n \to \infty} \mathbb{E}_{m \sim \mu} h_n (m) \\
				\mathbb{E}_{m \sim \mu} g(\p_B m) &= \lim_{n \to \infty} \mathbb{E}_{m \sim \mu} g_n (\pi_B m).
			\end{align*}
			By what we proved above, for any $n$, we have
	\begin{align*}
		\mathbb{E}_{m \sim \mu} [ h_n(m) ]
		\leq \mathbb{E}_{m \sim \mu} [ f(\p_A m)] \cdot \mathbb{E}_{m \sim \mu} [ g_n(\p_B m)]
	\end{align*}

	Taking that to the limit $n \to \infty$,
	\begin{align*}
		\lim_{n \to \infty} \mathbb{E}_{m \sim \mu} [ h_n(m) ]
		&\leq \lim_{n \to \infty} \left( \mathbb{E}_{m \sim \mu} [ f(\p_A m)] \cdot \mathbb{E}_{m \sim \mu} [ g(\p_B m)] \right) \\
		&= \lim_{n \to \infty}  \mathbb{E}_{m \sim \mu} [ f(\p_A m)] \cdot \lim_{n \to \infty}  \mathbb{E}_{m \sim \mu} [ g(\p_B m)]
	\end{align*}
	Therefore, for any distribution $\mu \in \DMem{A \cup B}$,
	\begin{align*}
\mathbb{E}_{m \sim \mu} [ f(\p_A m) \cdot g(\p_B m)]
		\leq \mathbb{E}_{m \sim \mu} [ f(\p_A m)] \cdot \mathbb{E}_{m \sim \mu} [ g(\p_B m)] .
	\end{align*}
\item The case where $f$ is non-positive and $g$ is non-negative is symmetric.
\item The case where $f$ and $g$ are both non-positive is also similar.
	We will define $f_n(\p_A m) = \max(f(\p_A m), -n)$,  $g_n(\p_B m) = \max(g(\p_B m), -n)$,
$h_n(m) = f_n(\p_A m) \cdot g_n(\p_B m) $. Then we have
			\begin{align*}
				\mathbb{E}_{m \sim \mu} f(\p_A m) \cdot g(\p_B m) &= \lim_{n \to \infty} \mathbb{E}_{m \sim \mu} h_n (m) \\
				\mathbb{E}_{m \sim \mu} g(\p_B m) &= \lim_{n \to \infty} \mathbb{E}_{m \sim \mu} g_n (\pi_B m)\\
				\mathbb{E}_{m \sim \mu} f(\p_B m) &= \lim_{n \to \infty} \mathbb{E}_{m \sim \mu} f_n (\pi_A m).
			\end{align*}
			And the rest follows.
	\end{itemize}

\item Now we consider the general case where we only know both $f$ and $g$
	are either lower-bounded or upper bounded.
	\begin{itemize}
		\item If both $f$ and $g$ are lower-bounded, reduce to the first case.
		\item If $f$ is lower-bounded by $L$, $g$ is upper-bounded by $U$,
			then we can consider function $f' = f + L$ and $g' = g - U$.
			Then $f'$ is non-negative and $g'$ is non-positive, so by step 2, we have
				\begin{align*}
\mathbb{E}_{m \sim \mu} [ f'(\p_A m) \cdot g'(\p_B m)]
		\leq \mathbb{E}_{m \sim \mu} [ f'(\p_A m)] \cdot \mathbb{E}_{m \sim \mu} [ g'(\p_B m)] .
	\end{align*}
	By calculations analogous to what we did in step 1, that implies
	\begin{align*}
\mathbb{E}_{m \sim \mu} [ f(\p_A m) \cdot g(\p_B m)]
		\leq \mathbb{E}_{m \sim \mu} [ f(\p_A m)] \cdot \mathbb{E}_{m \sim \mu} [ g(\p_B m)] .
	\end{align*}
\item If $f$ is upper-bounded and $g$ is lower-bounded: analogous to above.
\item If both $f$ and $g$ are upper-bounded: also, analogous to above.
	\end{itemize}
	\end{enumerate}

	Thus, $\mu$ satisfies $\{A, B\}$-\CPNA implies $\mu$ satisfies $(A, B)$-NA.
	And therefore, $\mu$ satisfies $\{A, B\}$-\CPNA for any $A, B \subseteq S$
	implies $S$  satisfies strong NA in $\mu$.

\end{proof}

\subsection{A $\mathbf{2}$-BI model for independence and negative association}

	The proof that independence implies \CPNA will use the following lemma.
\begin{lemma} \label{groupNA}
	In a distribution $\mu$, if $\mu$ satisfies $\{S_1, S_2\}$-\CPNA,
	$\mu$ satisfies	$\{T_1, T_2\}$-\CPNA,
	and $S_1 \cup S_2$ is independent from $T_1 \cup T_2$ in $\mu$
	then $\mu$ is $\{S_1 \cup T_1, S_2 \cup T_2\}$-\CPNA.
\end{lemma}

\begin{proof}
	By the definition of \CPNA and independence, $S_1, S_2$ are disjoint,
	$T_1, T_2$ are disjoint, and $S_1 \cup T_1, S_2 \cup T_2$ are disjoint.
	For any monotonically decreasing/increasing functions
	$f: \Mem[S_1 \cup T_1] \to \mathbb{R}^+, g: \Mem[S_2 \cup T_2] \to \mathbb{R}^+$,
	\begin{align*}
		&\mathbb{E}_{m \sim \mu}[f(\p_{S_1 \cup T_1} m) \cdot g(\p_{S_2 \cup T_2} m)] \\
		&=\mathbb{E}_{s \sim \pi_{S_1 \cup S_2} \mu} \mathbb{E}_{t \sim \pi_{T_1 \cup T_2}\mu}[f(\p_{S_1} s, \p_{T_1} t) \cdot g(\p_{S_2} s, \p_{T_2} t)] \tag{By independence of $S_1 \cup S_2$ and $T_1 \cup T_2$}\\
		&\leq \mathbb{E}_{s \sim \pi_{S_1 \cup S_2} \mu} \left( \mathbb{E}_{t_1 \sim \pi_{T_1} \mu}[f(\p_{S_1} s, t_1)] \cdot  \mathbb{E}_{t_2 \sim \pi_{T_2} \mu} [g(\p_{S_2} s, t_2)]\right) \tag{$\diamondsuit$} \\
		&\leq \mathbb{E}_{s_1 \sim \pi_{S_1} \mu} \mathbb{E}_{t_1 \sim \pi_{T_1} \mu}[f(s_1, t_1)] \cdot  \mathbb{E}_{s_2 \sim \pi_{S_2} \mu} \mathbb{E}_{t_2 \sim \pi_{T_2} \mu} [g(s_2, t_2)] \tag{$\heartsuit$} \\
		&\leq \mathbb{E}_{m \sim \mu}[f(\p_{S_1 \cup T_1} m)] \cdot \mathbb{E}_{m \sim \mu}[g(\p_{S_2 \cup T_2} m)] \tag{$\clubsuit$}
	\end{align*}
	where $\diamondsuit$ is because  $\pi_{T_1 \cup T_2} \mu$ is $T_1, T_2$-\CPNA
	and $f(\p_{S_1} s, t_1), g(\p_{S_2} s, t_2)$ are both monotonically
	decreasing/increasing in $T_1, T_2$; $\heartsuit$ is because $\pi_{S_1 \cup
	S_2} \mu$ is $S_1, S_2$-\CPNA and that $ \mathbb{E}_{t_1 \sim \pi_{T_1}
\mu}[f(\p_{S_1} s, t_1)]$, and $\mathbb{E}_{t_2 \sim \pi_{T_2} \mu}
[g(\p_{S_2} s, t_2)]$  are both monotonically decreasing/increasing in $S_1,
S_2$; $\clubsuit$ is by independence of $S_1$ and $T_1$ and the independence of
$S_2$ and $T_2$ in $\mu$.
\end{proof}

\IndepClosure*

\begin{proof}
	\label{proof:IndepClosure}
	Fix $\mathcal{S}$ and $\mathcal{T}$.
	Say $\mathcal{S} = \{S_1, \dots, S_p\}$ and $\mathcal{T} = \{T_1, \dots, T_q\} $.  For any
	$\mathcal{R}$ coarsening $\mathcal{S} \cup \mathcal{T}$, indexing $\mathcal{S}
	\cup \mathcal{T}$ as $\{U_1, \dots, U_{p + q}\}$, indexing $\mathcal{R}$ as
	$\{R_1, \dots, R_n\}$, we have:
	\begin{align*}
		\mathcal{R} = \{ \cup \{ U_j \mid j \in f(i)  \} \mid i \in [n] \} .
	\end{align*}
	Then, given a family of monotonically increasing/decreasing functions
	$g_i: R_i \to \mathbb{R}^{+}$
	\begin{align*}
		\mathbb{E}_{m \sim \mu} \left[\prod_{R_i \in \mathcal{R}} g_i(\p_{R_i} m)\right]
				&=
				\mathbb{E}_{m \sim \mu} \left[\prod_{i \in [n]} g_i(\p_{\cup \{U_j \mid j \in f(i)\}} m ) \right].
	\end{align*}
	For each $i$, $\cup \{U_j \mid j \in f(i)\}$ can be divided into the part in $S$ and the part in $T$. We refer to them as
	$S'_i$ and $T'_i$. (Some of $S'_i$ and $T'_i$ may be empty).
	Thus, for each $i$,
	\begin{align*}
		g_i(\p_{\cup \{U_j \mid j \in f(i)\}} m )
				&= g_i(\p_{S'_i \cup T'_i} m ) .
	\end{align*}

	By~\Cref{lemma:coarsen_project}, $\mathcal{S}' = \{S'_1, \dots, S'_n\}$
	coarsens $\mathcal{S}$, and $\mathcal{T}' = \{T'_1, \dots, T'_n\}$
	coarsens $\mathcal{T}$. So $\mu$ is $\mathcal{S}'$-\CPNA and
	$\mathcal{T}'$-\CPNA.

	We prove by induction on $k \in [n]$ that
	\begin{align*}
		\mathbb{E}_{m \sim \mu} \left[\prod_{i \in [k]} g_i(\p_{S'_i \cup T'_i} m ) \right] &\leq \prod_{i \in [k]} \mathbb{E}_{m \sim \mu} \left[g_i(\p_{S'_i \cup T'_i} m )\right] .
	\end{align*}
	\begin{description}
		\item [Base case] When $k = 1$, trivial.
		\item [Inductive case]
			For $k < n$, assume
				\begin{align*}
		\mathbb{E}_{m \sim \mu} [\prod_{i \in [k]} g_i(\p_{S'_i \cup T'_i} m )] &\leq \prod_{i \in [k]} \mathbb{E}_{m \sim \mu} [g_i(\p_{S'_i \cup T'_i} m )].
			\end{align*}
			Note that $\mu$ is $\mathcal{S}'$-\CPNA implies that
			$\mu$ is $\{\cup_{i \in [k]} (S'_i), S'_{k+1}\}$-\CPNA, and
			$\mu$ is $\mathcal{T}'$-\CPNA implies that $\{\cup_{i \in [k]} (T'_i), T'_{k+1}\}$-NA.
			Thus, by~\Cref{groupNA}, $\mu$ is also
			$\{\{\cup_{i \in [k]} (S'_i) \cup \{\cup_{i \in [k]} (T'_i) , S'_{k+1} \cup T'_{k+1}\} $-NA.
					Also, since all $g_i$ is monotonically increasing (decreasing) and non-negative,
					$m \mapsto \prod_{i \in [k]} g_i(m)$ is also a monotonically increasing (decreasing) function from $ \cup_{i \in [k]} S'_i \cup \cup_{i \in [k]} T'_i$ to $\mathbb{R}^+$.
					Therefore,
			\begin{align*}
					\mathbb{E}_{m \sim \mu} [\prod_{i \in [k + 1]} g_i(\p_{S'_i \cup T'_i} m )]
					&\leq \mathbb{E}_{m \sim \mu} [\prod_{i \in [k]}g_i(\p_{S'_i \cup T'_i} m )] \cdot \mathbb{E}_{m \sim \mu} [g_{k + 1}(\p_{S'_{k + 1} \cup T'_{k + 1}} m )] \\
					&\leq \prod_{i \in [k+1]} \mathbb{E}_{m \sim \mu} [g_i(\p_{S'_i \cup T'_i} m )],
			\end{align*}
	\end{description}
	where the second inequality follows from the inductive hypothesis.

	Thus, the desired inequality holds for any $\mathcal{R}$ coarsening $\mathcal{S} \cup \mathcal{T}$ and any family of monotonically increasing (decreasing) functions on $\mathcal{R}$.
	Thus, $\mu$ is $\mathcal{S} \cup \mathcal{T}$-\CPNA.
\end{proof}

\ProductIsBI*

\begin{proof}
	\label{proof:productBI}
	Let $\mathcal{X}_{i, \idx} = (X_i, \sqsubseteq_i, \oplus_{i, \idx}, E_{i, \idx})$.
	For any $i \in \{1, 2\}$, $\idx \in \idxset$, $\mathcal{X}_{i, \idx}$ is a BI frame.

First, for any $\idx \in \idxset$,
we prove that $(X, \sqsubseteq, \oplus_{\idx}, E_{\idx})$ is a BI frame.

	\begin{description}
		\item[Down-Closed.] Let $(z_1, z_2) \in (x_1, x_2) \oplus_{\idx} (y_1, y_2)$ with
			$(x_1, x_2) \sqsupseteq (x_1', x_2')$ and $(y_1, y_2) \sqsupseteq (y_1', y_2')$.
			Then, from the Down-Closed property of $\mathcal{X}_{1, \idx}$ and $\mathcal{X}_{2, \idx}$ respectively,
			we have that there exists $z_1'$ and $z_2'$ such that
			$z_i \sqsupseteq_i z_i'$ and $z_i' \in x_i' \oplus y_i'$ for $i \in \{1, 2\}$.
			Hence $(z_1, z_2) \sqsupseteq (z_1', z_2')$ and $(z_1', z_2') \in (x_1', x_2') \oplus_{\idx} (y_1', y_2')$.

		\item[Commutativity.] Immediate.

		\item[Associativity.] Let $(w_1, w_2) \in (t_1, t_2) \oplus_{\idx} (z_1, z_2)$ and $(t_1, t_2) \in (x_1, x_2) \oplus_{\idx} (y_1, y_2)$.
			Then for $i \in \{1, 2\}$ there exists $s_i \in y_i \oplus_{i, \idx} z_i$ such that $w_i \in x_i \oplus_{i, \idx} s_i$.
			Thus, $(s_1, s_2) \in (y_1, y_2) \oplus_{\idx} (z_1, z_2)$ and $(w_1, w_2) \in (x_1, x_2) \oplus_{\idx} (s_1, s_2)$.

		\item[Unit Existence.] Immediate.
		\item[Unit Closure.]
			If $(e_1, e_2) \in E_{\idx}$ and $(e_1', e_2') \sqsupseteq (e_1, e_2)$, then $e_i' \in E_{i, \idx}$, so
			$(e_1', e_2') \in E_{\idx}$.
		\item[Unit Coherence.]
		Let $(e_1, e_2) \in E_{\idx}$ and $(x_1, x_2) \in (y_1, y_2) \oplus_{\idx} (e_1, e_2)$. Then
			$x_i \sqsupseteq_{i} y_i$, which implies that $(x_1, x_2) \sqsupseteq (y_1, y_2)$.
			\qedhere
	\end{description}

Second, we show
if $\idx \leq \idx' \in \idxset$, then $(\mu_1, \mu_1') \oplus_{\idx} (\mu_2, \mu_2') \subseteq (\mu_1, \mu_1') \oplus_{\idx'} (\mu_2, \mu_2') $:

$\mathcal{X}_1$ and $\mathcal{X}_2$ are \LOGIC frames, so $\idx \leq \idx'$ implies that  $\mu_1 \oplus_{1, \idx} \mu_2 \subseteq_1 \mu_1 \oplus_{1, \idx'} \mu_2 $ and  $\mu_1' \oplus_{2, \idx'} \mu_2' \subseteq_2 \mu_1' \oplus_{2, \idx'} \mu_2'$.
Therefore, $(\mu_1, \mu_1') \oplus_{\idx} (\mu_2, \mu_2') \subseteq (\mu_1, \mu_1') \oplus_{\idx'} (\mu_2, \mu_2') $.
\end{proof}

\section{Program Logic}
\label{sec:app:seplog}

\subsection{The semantics of \Lang}

\begin{definition}[Semantics of expressions]
	\label{def:semantics_expressions}
We assume all expressions are well-typed.
We interpret them as $\Mem[\DetVar] \times \Mem[T] \to \Val$, which can be naturally lifted to $\Mem[\DetVar] \times \DMem{T} \to \Dist(\Val)$. For $\sigma \in \Mem[\DetVar]$, $m \in \Mem[T]$,
\begin{align*}
	\Sem{x_d} (\sigma, m) &= \sigma(x_d) \\
	\Sem{x_r} (\sigma, m) &= m(x_r)  \\
\Sem{[e_1, \dots, e_n]} (\sigma, m) &= [\Sem{e_1}(\sigma, m) , \dots, \Sem{e_n}(\sigma, m) ] \\
\Sem{e + f} (\sigma, m)  &=
\Sem{e}(\sigma, m)  + \Sem{f}(\sigma, m)  &\text{when $\Sem{e}(\sigma, m) $ and $\Sem{f_r}(\sigma, m) $ are numbers} \\
	\Sem{e \land f} (\sigma, m) &=
	\Sem{e}(\sigma, m)  \land \Sem{f}(\sigma, m)  &\text{when $\Sem{e}(\sigma, m) $ and $\Sem{f_r}(\sigma, m) $ are booleans}
\end{align*}
\end{definition}

	\label{def:convexcomb}
\begin{definition}[Convex combination of distributions]
Let the binary operator $\comb_{\rho}$ takes a convex combination of two distributions, i.e.,
for any $\mu_1, \mu_2 \in \DMem{S}$, for any $x \in \Mem[S]$,
\begin{align*}
	\mu_1 \comb_{\rho} \mu_2 (x) \triangleq \rho \cdot \mu_1(x) + (1- \rho) \cdot \mu_2(x).
\end{align*}
\end{definition}

\begin{definition}[Semantics of \Lang]
	\label{def:semantics_lang}
Let $\text{Unif}_S$ denotes the uniform distribution on finite set $S$, i.e., $\mu: s \mapsto \frac{1}{|S|}$ for any $s \in S$.

We also assume that for any $\sigma \in \Mem[\DetVar]$, any $\mu_1 \in \DMem{T_1}, \mu_2 \in \DMem{T_2}$ where $T_1, T_2 \subseteq \RanVar$, the first component of
$\Sem{c}(\sigma, \mu_1)$ and $\Sem{c}(\sigma, \mu_2)$ are the same.
	For any $(\sigma, \mu) \in \Config$, let
\begin{align*}
	\Sem{\Skip}(\sigma, \mu)	&\triangleq (\sigma, \mu) \\
	\Sem{\Assn{x_d}{e_d}}(\sigma, \mu)	&\triangleq (\sigma[x_d \mapsto \Sem{e_d} (\sigma, \mu)], \mu) \\
	\Sem{\Assn{x_r}{e_r}}(\sigma, \mu)	&\triangleq (\sigma, \dbind(\mu, m \mapsto \dunit(m[x_r \mapsto \Sem{e_r}(\sigma, m)]))) \\
	\Sem{\Rand{x_r}{\Unif{S}}}(\sigma, \mu)	&\triangleq (\sigma, \dbind(\mu, m \mapsto \dbind(\text{Unif}_S, v \mapsto \dunit(m[x_r \mapsto v]))))  \\
	\Sem{\Seq{c}{c'}}(\sigma, \mu)	&\triangleq \Sem{c'}(\Sem{c} (\sigma, \mu)) \\
	\Sem{\RCond{b}{c}{c'}(\sigma, \mu)} &\triangleq
	\begin{cases}
		\Sem{c}(\sigma, \mu) & \text{ if } \Sem{b}(\sigma, \mu) = \delta(\ktt)\\
		\Sem{c'}(\sigma, \mu) & \text{ if } \Sem{b}(\sigma, \mu) = \delta(\kff)\\
		(\sigma, \mu_1 \comb_{\rho} \mu_2) & \text{otherwise, where } (\sigma, \mu_1) =  \Sem{c}(\sigma, \mu \mid \Sem{b} \sigma = \ktt), \\
		\qquad																																			& \rho = \mu(\Sem{b} \sigma = \ktt),
																																				 \text{ and } (\sigma, \mu_2) =  \Sem{c'}(\sigma, \mu \mid \Sem{b} \sigma = \kff).
	\end{cases}\\
\Sem{\RWhile{b}{c}} (\sigma, \mu) &\triangleq
\lim_{n \to \infty} \Sem{(\RCondt{b}{c})^n; \RCondt{b}{\mathbf{abort}}} (\sigma, \mu)
	\end{align*}
\end{definition}
In the last case of the conditional, we write $\Sem{b} \sigma$ for the partial
evaluation of $b$ on the deterministic memory and think of it as another
expression.  So $\mu \mid \Sem{b} \sigma = \ktt$ and  $\mu \mid \Sem{b} \sigma
= \kff$ are both conditional distribution.  We assume that there is no update
to the deterministic memory when branching on randomized expressions, so the
deterministic memory of both $\Sem{c} (\sigma, \mu \mid \Sem{b} \sigma = \ktt)$
and $\Sem{c} (\sigma, \mu \mid \Sem{b} \sigma = \kff)$ remains to be $\sigma$.

In the semantics for the while loop, we can see $\mathbf{abort}$ as a command
that programmers does not have access to: for any $\sigma, \mu$,
$\Sem{\mathbf{abort}} (\sigma, \mu)$ returns the zero sub-distribution.  The
limit, taken with the point-wise order, exists because the sub-distribution's
mass monotonically increases as $n$ increases and is upper bounded by 1.  In
practice, because we assumed that all loops terminates in finite steps, the
limit is always a full distribution, so all commands in \Lang can still be
interpreted as transformers from configurations to configurations.

\subsection{The atomic propositions and axioms}
\begin{definition}
	\label{def:bignotation}
	For any operation $\odot \in \{ \land, \lor, \negand, \indand\}$, we use the
corresponding big-operation $\bigodot \in \left \{ \bigwedge, \bigvee, \bigneg, \bigind \right \}$.
\begin{itemize}
	\item For any constant or logical variable $N \geq 1$, let $\bigodot_{i
= 0}^{N} P_i = P_0$ abbreviate $((P_{0} \odot P_{1}) \odot \cdots ) \odot
P_{N-1}$. Formally, let $\bigodot_{i = 0}^{N} P_i = P_0$ if $N = 1$, and
let $\bigodot_{i = 0}^{N} P_i \triangleq \left(\bigodot_{i =
0}^{N - 1} P_i \right) \odot P_{N-1}$ for $N > 1$.

\item For a finite multi- set of formula $\{P_i\}_{i \in S}$, let $\bigodot_{s \in S}
P_s$ abbreviate $((P_{s_0} \odot P_{s_1}) \odot \cdots ) \odot P_{s_{k}}$,
where $s_0, \dots, s_{k}$ is an arbitrary ordering of $S$.  The satisfaction
is not ambiguous since $\odot$ is associative and commutative.

\item For any program variable $v \in \DetVar \cup \RanVar$, for any state
  $(\sigma, \mu) \models \apEq{v}{N}$, we want $\bigodot_{i = 0}^{v} P_i$ to be equivalent to
$\bigodot_{i = 0}^{N} P_i$.  Formally,  $\bigodot_{i = 0}^{v} P_i$
abbreviates $\bigvee_{N \in \Val} ( \apEq{v}{N} \land \bigodot_{i = 0}^{N} P_i )$.
\end{itemize}
\end{definition}

\starcapturesNA*

\begin{proof}
	Forward direction: we fix $\sigma \in \Mem[\DetVar]$.
 Let $Y_j = \{y_i \mid 0 \leq i \leq j\}$
	we prove by induction on $j$ that
		$(\sigma, \pi_{Y_{j} } \mu)\models  \bigneg_{i=0}^{j + 1} \apIn{y_i}  $.

		If $j = 0$: then $y_1 \in \dom(\mu)$, and $(\sigma, \pi_{\{y_1\}} \mu) \models \apIn{y_1}$.

		If $j \geq 1$: Assuming $(\sigma, \pi_{Y_{j - 1} } \mu)\models  \bigneg_{i=0}^{j 1} \apIn{y_i}  $.
	Since $Y$ satisfies NA in $\mu$, by~\Cref{theorem:PNAeqNA}, $\mu$ is $\mathcal{T}$-\CPNA for any partition $\mathcal{T}$ of $Y$.
	In particular, for any partition $\mathcal{T}_1$ on $Y_{j-1}$ and any partition $\mathcal{T}_2$ on $\{y_j\}$,
		$\mu$ must be $\mathcal{T}_1 \cup \mathcal{T}_2$-\CPNA.
		Thus,
		$\pi_{Y_j} \mu \in \pi_{Y_{j - 1}} \mu \oplus \pi_{\{y_j\}} \mu$.
		Since $(\sigma, \pi_{Y_{j - 1} } \mu)\models  \bigneg_{i=0}^{j } \apIn{y_i}  $
		and $(\sigma, \pi_{\{y_j\} } \mu)\models  \apIn{y_j}  $,
		that implies $(\sigma, \pi_{Y_{j} } \mu)\models  \bigneg_{i=0}^{j + 1} \apIn{y_i}  $.

		Thus, $(\sigma, \pi_{Y_{j} } \mu)\models  \bigneg_{i=0}^{j + 1} \apIn{y_i}  $.
		Take $j = K$, we have $(\sigma, \pi_{Y} \mu) \models  \bigneg_{i=0}^{K} \apIn{y_i}$.
		By persistence, $(\sigma, \mu) \models  \bigneg_{i=0}^{K} \apIn{y_i}$.

	Backward direction: for any $A, B$ being  disjoint subsets of $Y$,
	by commutativity and associativity of $\negand$, we can
	reorder formula and get
	$(\sigma, \mu) \models \left( \bigneg_{y_i \in A} \apIn{y_i} \negand \bigneg_{y_i \in B} \apIn{y_i} \right) \negand \bigneg_{y_i \in (T \setminus (A \cup B))} \apIn{y_i}  $.
	By satisfaction rules and the definition of $\oplus$, there exists $\mu' \sqsubseteq \mu$ such that
	$(\sigma,  \mu') \models  \bigneg_{y_i \in A} \apIn{y_i} \negand \bigneg_{y_i \in B} \apIn{y_i}$.
	By satisfaction rules again,
	there exists $\mu_1, \mu_2, \mu''$ such that
	$ \mu' \sqsupseteq  \mu'' \in  \mu_1 \oplus \mu_2$,
	and $(\sigma, \mu_1) \models  \bigneg_{y_i \in A} \apIn{y_i} $,
	and $(\sigma, \mu_2) \models  \bigneg_{y_i \in B} \apIn{y_i} $.
	Note that $\mu_1$ is trivially $\{ A \}$-\CPNA,
	and $\mu_2$ is trivially $\{ B \}$-\CPNA.
	Thus, $\mu''$ satisfies $\{A, B \}$-\CPNA.

	Therefore, $\mu$ satisfies $(A, B)$-NA for any $A, B$ being disjoint subsets of $Y$,
	i.e., $\mu$ satisfies NA on $Y$.
\end{proof}

\PNAIntro*

\begin{proof}
	For any state $(\sigma, \mu)$ satisfying $\apOnehot{[x_1, \dots, x_{N-1}]}{N}$,
	by~\Cref{thm:buildingblock},
	$\{x_1, \dots, x_N\}$ satisfies	NA in $\mu$,
	and by~\Cref{theorem:starcapturesNA},
	$(\sigma, \mu) \models \bigneg_{\gamma = 1}^N \apIn{x_{\gamma}} $.
	Similarly,
	by~\Cref{thm:buildingblock},
	$\{x_1, \dots, x_N\}$ satisfies	NA in $\mu$,
	and by~\Cref{theorem:starcapturesNA},
	$(\sigma, \mu) \models \bigneg_{\gamma = 1}^N \apIn{x_{\gamma}} $.

\end{proof}

\begin{lemma}
	\label{lemma:monotoneclosure2CPNA}
	Given a distribution $\mu$ with domain $S$.
	Let $\mathcal{S} = \{S_1, \dots, S_N \}$ be a partition of $S$,
	$\{f_i: \Mem[S_i] \to \Mem[T_i]\}_{\mathcal{S}_i \in \mathcal{S}}$ be a family of non-decreasing functions (or a family of non-increasing functions),
	and $\mathcal{T} = \{ T_1, \dots, T_N\}$ be a partition of another set $T$.
	Let
	\[\mu' = 	\dbind\Big(\mu, m \mapsto  \bigoplus_{S_i \in \mathcal{S}} \dunit(f_i(\p_{S_i} m))  \Big). \]
	If $\mu$ satisfies $\mathcal{S}$-\CPNA, then $\mu'$ satisfies
	$\mathcal{T}$-\CPNA.
\end{lemma}
\begin{proof}
	It suffices to show that for any $\mathcal{R} = \{ R_1, \dots, R_m\}$ that coarsens $\mathcal{T}$,
	for any family of non-negative non-increasing (or non-decreasing) functions $\{g_i: \Mem[R_i] \to \mathbb{R}^+\}_{R_i \in \mathcal{R}}$,
\[
	\mathbb{E}_{m' \sim \mu'}\left[ \prod_{i} g_i(\p_{R_i} m')\right] \leq \prod_i \mathbb{E}_{m' \sim \mu' } \left[  g_i(\p_{R_i} m') \right].
\]

Our first step is to show that if we can obtain a distribution on $\mathcal{R}$ by first applying monotone map and then coarsening, then we can also obtaining that by first coarsening and then applying monotone map:
For any $\mathcal{R}$ that coarsens $\mathcal{T} = \{T_1, \dots, T_n \}$, there exists some coarsening function $g$ such that
\begin{align*}
	\mathcal{R} = \{ \cup\{T_j \mid j \in g(i)\} \mid i \in [n] \}.
\end{align*}
Let $R_i =\cup \{T_j \mid j \in g(i)\} $.
Since $g$ is a coarsening function,
for each $i \in [n]$,
there exists exactly one $k_i$ (i.e., the index for the component in the coarsened function) such that $g(k_i) \ni i$.
Then,
$T_i \subseteq \cup \{T_j \mid j \in g(k_i)\} = R_{k_i}$.

Let $R'_i =\cup \{S_j \mid j \in g(i)\} $,
and $\mathcal{R}' = \{ R'_i \mid i \in [n]\}$.
Then $\mathcal{R}'$ coarsens $\mathcal{S}$.
	Define $h_i: \Mem[R'_i] \to  \Mem[R_i]$ by having
	$h_i(m) = \bowtie_{j \in g(i)} f_j(\p_{S_j} m)$.
Since each $f_i$ is monotone, then each $h_i$ also monotone in the point-wise order. Then,
	\begin{align*}
		\mu' =	\dbind\Big(\mu,   m \mapsto \bigoplus_{i \in [n]} \dunit(f_i(\p_{S_i} m)) \Big)
		&= 	\dbind\Big(\mu, m \mapsto \bigoplus_{R'_i \in \mathcal{R}'} \dunit(h_i(\p_{R'_i} m))  \Big),
	\end{align*}
	so $\mu'$ is equivalent to applying $h_i$ on each component of $\mathcal{R}'$.

	Also, since $\mu$ is  $\mathcal{S}$-\CPNA and $\mathcal{R}'$ coarsens $\mathcal{S}$,
for any family of non-negative non-increasing (or non-decreasing) functions $\{g_i: \Mem[R'_i] \to \mathbb{R}^+\}_{R'_i \in \mathcal{R}'}$,
\begin{align} \label{eq:oldPNAreq}
\mathbb{E}_{m \sim \mu} \left[ \prod_{i} g_i(\p_{R'_i} m)\right] \leq \prod_i \mathbb{E}_{m \sim \mu}   g_i(\p_{R'_i} m) .
\end{align}
Our second step is to show that this inequality is preserved under monotone maps.
\begin{itemize}
	\item If every $h_i$ is non-increasing and $\{g_i: \Mem[R_i] \to \mathbb{R}^+\}_{R_i \in \mathcal{R}}$ are non-decreasing,
 is non-increasing,
	\begin{align*}
				& \mathbb{E}_{m' \sim \mu'}[ \prod_{R_i \in \mathcal{R}} g_i(\p_{R_i} m')]\\
				&= \sum_{m' \in \Mem[T]} \mu'(m') \cdot \prod_{R_i \in \mathcal{R}}  g_i(\p_{R_i} m') \\
				&=\sum_{m' \in \Mem[T]} \sum_{m \in \Mem[S]} \mu(m) \cdot (\prod_{R'_i \in \mathcal{R}'} \dunit(h_i(\p_{R'_i} m))(\p_{R_i} m')) \cdot \prod_{R_i \in \mathcal{R}}  g_i(\p_{R_i} m') \\
				&=\sum_{m \in \Mem[S]} \mu(m) \cdot \prod_{R'_i \in \mathcal{R}'}  g_i(h_i(\p_{R'_i} m)) \\
				&= \mathbb{E}_{m \sim \mu} [ \prod_{R'_i \in \mathcal{R}'} g_i(h_i(\p_{R'_i} m))] .
	\end{align*}
	Note that $g_i \circ f$ is non-negative non-decreasing,
	then since $\mu$ satisfies~\Cref{eq:oldPNAreq}, we have
	\begin{align*}
				\mathbb{E}_{m \sim \mu} [\prod_{R'_i \in \mathcal{R}'} g_i(h_i(\p_{R'_i} m))]
				&\leq
				\prod_{R'_i \in \mathcal{R}'} \mathbb{E}_{m \sim \mu} [g_i(h_i(\p_{R'_i} m))]\\
				&= \prod_{R'_i \in \mathcal{R}'} \sum_{m \in \Mem[S] } \mu(m) \cdot g_i(h_i(\p_{R'_i} m)) \\
				&= \prod_{R_i \in \mathcal{R}} \mathbb{E}_{m' \sim \mu' }  g_i(\p_{R_i} m') .
	\end{align*}
	Combined, we have $ \mathbb{E}_{m' \sim \mu'}[ \prod_{i} g_i(\p_{R_i} m')] \leq \prod_i \mathbb{E}_{m' \sim \mu' }   g_i(\p_{R_i} m')$.

\item When $f_i$ is non-increasing and $\{g_i\}$ are non-negative non-decreasing,
	or when $f_i$ is non-decreasing and $\{g_i\}$ are non-negative non-decreasing/non-increasing,
	the proof is analogous.
\end{itemize}
\end{proof}

\MonotoneMapAxiom*
\begin{proof}
 Abbreviate $ \{ \bigcup_{\alpha = 0}^{K_{\gamma} + 1} \{x_{\gamma, \alpha}\} \mid 1 \leq \gamma < M\}$ as $X_M$.

	For any state $(\sigma, \mu)$ satisfying $\bigneg_{\gamma = 0}^N (\bigwedge_{\alpha = 0}^{K_{\gamma} + 1} \apIn{x_{\gamma, \alpha}}) $,
		we show by induction on $M$ that: for $1 \leq M \leq N$,
		\begin{itemize}
			\item $(\sigma, \mu) \models \bigneg_{\gamma = 0}^M (\bigwedge_{\alpha = 0}^{K_{\gamma}} \apIn{x_{\gamma, \alpha}}) $;
			\item And $\mu$ satisfies $X_M$-\CPNA implies
		 $\mu$ satisfies $X_N$-\CPNA.
	\end{itemize}

		Assuming that it is true for $M$, we show that for $M-1$.
		By assumption, $(\sigma, \mu) \models \bigneg_{\gamma = 0}^M (\bigwedge_{\alpha = 0}^{K_{\gamma} + 1} \apIn{x_{\gamma, \alpha}}) $,
		so there exists $\mu', \mu_1, \mu_2$ such that $\mu \sqsupseteq \mu' \in \mu_1 \oplus \mu_2$,
		$(\sigma, \mu_1) \models \bigneg_{\gamma = 0}^{M-1} (\bigwedge_{\alpha = 0}^{K_{\gamma} + 1} \apIn{x_{\gamma, \alpha}})$ and $(\sigma, \mu_2) \models \bigwedge_{\alpha = 0}^{K_M + 1} \apIn{x_{M, \alpha}}$. Thus, $\mu'$ is $X_M$-\CPNA if $\mu_1$ is $X_{M-1}$-\CPNA and $\mu_2$ is $\cup_{\alpha = 0}^{K_M + 1} \{ x_{M, \alpha} \}$-\CPNA. Thus,
		\begin{itemize}
			\item By the definition of $\oplus$, we have $\mu \sqsupseteq \mu_1$,
				and by persistence $(\sigma, \mu) \models \bigneg_{\gamma = 0}^{M-1} (\bigwedge_{\alpha = 0}^{K_{\gamma} + 1} \apIn{x_{\gamma, \alpha}}) $;
			\item Since 	$(\sigma, \mu_1) \models \bigneg_{\gamma = 0}^{M-1} (\bigwedge_{\alpha = 0}^{K_{\gamma} + 1} \apIn{x_{\gamma, \alpha}})$, we have  $\cup X_{M-1}$ is inside $\dom(\sigma) \cup \dom(\mu_1)$. Thus,
				if $\mu$ satisfies $X_{M-1}$-\CPNA, then $\mu_1 \sqsubseteq \mu$  satisfies $X_{M-1}$-\CPNA.
				Trivially, $\mu_2$ is $\cup_{\alpha = 0}^{K_M} \{ x_{M, \alpha} \}$-\CPNA.
				Thus, $\mu'$ is $X_M$-\CPNA, and $\mu$ is $X_M$-\CPNA.
				By inductive assumption, $\mu$ is $X_N$-\CPNA.
		\end{itemize}
	 Then, letting $M =1$ would give us: $\mu$ satisfies $X_1$-\CPNA implies
		 $\mu$ satisfies $X_N$-\CPNA.
			$X_1$ is a partition of one component, so $\mu$ satisfying $X_1$-\CPNA is trivial.
			Therefore, $\mu$ satisfies $X_N$-\CPNA.

		Since $ \bigwedge_{\gamma = 0}^N y_{\gamma} = f_{\gamma}(x_{\gamma, 1}, \dots, x_{\gamma, K_{\gamma}})$ and $f_{\gamma}$ are all antitone or monotone,
		by~\Cref{lemma:monotoneclosure2CPNA},
	$\mu$ satisfies $\{ \{ y_{\gamma} \} \mid 1 \leq \gamma < N\}$-\CPNA.
	Then we can show $(\sigma, \mu) \models \bigneg_{\gamma=1}^N \apIn{y_{\gamma}}$
	through another simple induction or by applying the existing theorems.
 If we do that by applying the theorems,
	~\Cref{theorem:PNAeqNA} implies that
	if $\mu$ satisfies $\{ \{ y_{\gamma} \} \mid 1 \leq \gamma \leq N\}$-\CPNA,
	then $\{ \{ y_{\gamma} \} \mid 1 \leq \gamma \leq N\}$ satisfies NA in $\mu$.
	Then, by~\Cref{theorem:starcapturesNA}, $(\sigma, \mu) \models \bigneg_{\gamma=1}^N \apIn{y_{\gamma}}$.
\end{proof}

\subsection{The restriction property}
We prove the restriction on deterministic memories and on randomized memories
by separate induction, and then combine them.

\begin{lemma}[Restriction on deterministic memories]
	\label{lemma:det_restriction}
	Let $(\sigma, \mu)$ be any configuration in $\Config$,
		and let $\phi$ be a \AssertLOGIC formula interpreted on $\CombModel$,
		Then, for any $m \in \Mem[\DetVar \setminus \FV(\phi)]$,
		\[(\sigma, \mu) \models \phi \iff (\p_{\FV(\phi)}\sigma \bowtie m, \mu) \models \phi. \]
	\end{lemma}

	\begin{proof}
		Note that the two directions are symmetric, so we only prove the forward direction.

		We prove it by induction on the syntax of formula.
		Most cases are straightforward, so we only show three cases.
		\begin{description}
			\item [$\phi = P \to Q$: ]  Assuming $(\sigma, \mu) \models P \to Q$,
				that says for any $(\sigma', \mu') \sqsupseteq (\sigma, \mu)$,
				if $(\sigma', \mu') \models P$, then $(\sigma', \mu') \models Q$.

				For any $(\widehat{\sigma}, \widehat{\mu}) \sqsupseteq (\p_{\FV(\phi)}\sigma \bowtie m, \mu)$,
				it must $\widehat{\sigma} = \p_{\FV(\phi)}\sigma \bowtie m$ and $\widehat{\mu} \sqsupseteq \mu$.
				If $ (\widehat{\sigma}, \widehat{\mu} ) \models P$,
				then by $\widehat{\sigma} = \p_{\FV(\phi)}\sigma \bowtie m$
				and the inductive hypothesis, $ (\sigma, \widehat{\mu}) \models P $;
				since $(\sigma, \mu) \models P \to Q$ and $\widehat{\mu} \sqsupseteq \mu$,
				$ (\sigma, \widehat{\mu}) \models P $ implies $ (\sigma, \widehat{\mu}) \models Q$;
				and by inductive hypothesis again, $ (\p_{\FV(\phi)}\sigma \bowtie m, \widehat{mu}) \models Q $,
				Thus, $ (\p_{\FV(\phi)}\sigma \bowtie m, \mu) \models P \to Q$.

			\item [$\phi = P \negand Q$: ]
				Assuming $(\sigma, \mu) \models P \negand Q$,
				then there exists
				$(\widehat{\sigma}, \widehat{\mu}), (\sigma_1, \mu_1), (\sigma_2, \sigma_2)$ such that
				$(\sigma, \mu) \sqsupseteq (\widehat{\sigma}, \widehat{\mu}) \in (\sigma_1, \mu_1) \oplus (\sigma_2, \mu_2)$,
				$(\sigma_1, \mu_1) \models P$,
				and $(\sigma_2, \mu_2) \models Q$.
				By the definition of the pre-order and $\oplus$,
				it must $\sigma = \widehat{\sigma} = \sigma_1 = \sigma_2$,
				$\mu \sqsupseteq \widehat{\mu} \in \mu_1 \oplus \mu_2$.

				By inductive hypothesis,
				$(\p_{\FV(\phi)}\sigma_1 \bowtie m, \mu_1) = (\p_{\FV(\phi)}\sigma \bowtie m, \mu_1) \models P$,
				and $(\p_{\FV(\phi)}\sigma_2 \bowtie m, \mu_2) = (\p_{\FV(\phi)}\sigma \bowtie m, \mu_2) \models Q$.
				Also,
				\[(\p_{\FV(\phi)}\sigma \bowtie m, \mu) \sqsupseteq (\p_{\FV(\phi)}\sigma \bowtie m, \widehat{\mu}) \in  (\p_{\FV(\phi)}\sigma \bowtie m, \mu_1)  \oplus (\p_{\FV(\phi)}\sigma \bowtie m, \mu_2)  . \]
				So $(\p_{\FV(\phi)}\sigma \bowtie m, \mu) \models P \negand Q$.

			\item [$\phi = P \indand Q$]
				Analogous as the case for $P \negand Q$.

			\item [$\phi = P \indimp Q$: ]
				Assuming $(\sigma, \mu) \models P \indimp Q$,
				that says for any
				$(\sigma'', \mu'') \in (\sigma, \mu) \otimes (\sigma', \mu')$,
				if $(\sigma', \mu') \models P$,
				then $(\sigma'', \mu'') \models Q$.

				For any $(\sigma'', \mu'') \in (\p_{\FV(\phi)}\sigma \bowtie m, \mu) \otimes (\sigma', \mu')$,
				it must $\sigma'' = \sigma' = \p_{\FV(\phi)}\sigma \bowtie m $, $\mu'' \in \mu \otimes \mu'$.
				Thus, $(\sigma', \mu') \models P$ is equivalent to $( \p_{\FV(\phi)}\sigma \bowtie m , \mu') \models P$, and by inductive hypothesis,
				that is equivalent to $(\sigma, \mu') \models P$.
				It also follows that $(\sigma, \mu'') \in (\sigma, \mu) \otimes (\sigma, \mu')$.
				Since $(\sigma, \mu) \models P \indimp Q$ and $(\sigma, \mu') \models P$,
				we have $(\sigma, \mu'') \models Q$.
				By inductive hypothesis, $ (\p_{\FV(\phi)}\sigma \bowtie m, \mu'') \models Q $,
				and equivalently $(\sigma'', \mu'') \models Q$.

				Thus, $(\p_{\FV(\phi)}\sigma \bowtie m, \mu) \models P \indimp Q$.

		\end{description}
	\end{proof}

	\begin{lemma}[Restriction on randomized memories]
		\label{lemma:random_restriction}
	Let $(\sigma, \mu)$ be any configuration in $\Config$,
		and let $\phi$ be a \AssertLOGIC formula interpreted on $\CombModel$,
		Then,

		\[(\sigma, \mu) \models \phi \iff (\sigma, \pi_{\FV(\phi)}\mu) \models \phi. \]
	\end{lemma}
	\begin{proof}
		The reverse direction follows by the Kripke monotonicity,
		and the forward direction follows by induction on $\phi$.
		The proof for most of the inductive cases is very similar to the proof that
		the probabilistic model in~\citet{PSL} satisfies restriction.
		The new inductive case is:
		\begin{itemize}
			\item $\phi \equiv  P \negand Q$.
				Then, $\mu \models \phi$ iff there exists  $\mu', \mu_1, \mu_2$ s.t. $\mu \sqsupseteq \mu' \in \mu_1 \oplus \mu_2$, $\mu_1 \models P$ and  $\mu_2 \models Q$.
				By induction, $\pi_{FV(P)} \mu_1 \models P$ and  $\pi_{FV(Q)} \mu_2 \models Q$.
				Note that $\pi_{FV(P)} \mu_1 \sqsubseteq \mu_1$ and  $\pi_{FV(Q)} \mu_2 \sqsubseteq \mu_2$.
				By Down-closure, there exists $\mu'' \sqsubseteq  \mu'$  such that
				$\mu'' \in \pi_{FV(P)} \mu_1 \oplus \pi_{FV(Q)} \mu_2$.
				This $\mu''$ satisfies $P \negand Q$.
				By definition of $\oplus$ in the \CPNA model,
				\[dom(\mu'') = dom(\pi_{FV(P)} \mu_1) \cup dom(\pi_{FV(Q)} \mu_2 ) = FV(P) \cup FV(Q) = FV(P \negand Q) .\]
				Also, by the definition of the pre-order, $\mu'' \sqsubseteq \mu' \sqsubseteq \mu$
				implies that $\mu'' = \pi_{dom(\mu'')} \mu = \pi_{FV(\phi)} \mu$.
				Thus, $\pi_{FV(\phi)} \mu = \mu'' \models \phi$.
\end{itemize}
\end{proof}

	\MBIRestriction*
	\begin{proof}
		Based on~\Cref{lemma:det_restriction} and~\Cref{lemma:random_restriction},
\[(\sigma, \mu) \models \phi \iff (\sigma, \pi_{\phi} \mu) \models \phi  \iff (\p_{\FV(\phi)}\sigma \bowtie m, \pi_{\FV(\phi)}\mu) \models \phi. \]
	\end{proof}

For the counterexample of the restriction property, we prove a lemma.
\begin{lemma}
	\label{lemma:counterex}
	Let $\sigma \in \Mem[\DetVar]$ be ``empty'' -- let every deterministic variable be undefined,
	$\mu $ be the uniform distribution over one hot vectors on $A, B$,
and $\phi = (\apUnif{C}{\{0, 1\}}) \negimp (\apIn{B} \indand \apIn{C})$.
Then, $(\sigma, \mu) \models \phi$.
\end{lemma}

\begin{proof}
			Fix any $\mu_C$ such that $ (\sigma, \mu_C) \models \apUnif{C}{\{0, 1\}}$,
			which implies that $\pi_C \mu_C (0) = 0.5$ and $\pi_C \mu_C(1) = 0.5$.
			Fix $\mu_e \in \mu \oplus \mu_C$.

			Since $B \in \dom(\mu)$, $\mu$ is trivially $\{\{B\}\}$-\CPNA.
			Similarly, $\mu_C$ is trivially $\{\{C\}\}$-\CPNA.
				Thus, $\mu_e \in \mu \oplus \mu_C$ must be $\{\{B\}, \{C\}\}$-\CPNA.
				Then for any two both monotone or antitone functions $f: \Mem[B] \to \mathbb{R}^+,
				g: \Mem[C] \to \mathbb{R}^+$,
				\begin{align*}
					\mathbb{E}_{m \sim \mu_e}[f(\p_B m) \cdot g(\p_C m)] \leq \mathbb{E}_{m \sim \mu_e}[f(\p_B m)] \cdot \mathbb{E}_{m \sim \mu_e}[g(\p_C m)].
				\end{align*}
				Similarly, 	$\mu_e \in \mu \oplus \mu_C$ must be $\{\{A\}, \{C\}\}$-\CPNA,
				and for any two both monotone or antitone functions $f: \Mem[A] \to \mathbb{R}^+,
				g: \Mem[C] \to \mathbb{R}^+$,
				\begin{align}
					\label{eq:AC-NA}
					\mathbb{E}_{m \sim \mu_e}[f(\p_A m) \cdot g(\p_C m)] \leq \mathbb{E}_{m \sim \mu_e}[f(\p_A m)] \cdot \mathbb{E}_{m \sim \mu_e}[g(\p_C m)].
				\end{align}

				Suppose variables $B$ and $C$ are not independent in $\mu_e$,
				then by~\Cref{lemma:pnaeq_indep},
				there must exists some both monotone or both antitone functions $f: \Mem[B] \to \mathbb{R}^+,
				g: \Mem[C] \to \mathbb{R}^+$ such that
				\begin{align*}
					& \mathbb{E}_{m \sim \mu_e}[f(\p_B m) \cdot g(\p_C m)] < \mathbb{E}_{m \sim \mu_e}[f(\p_B m)] \cdot \mathbb{E}_{m \sim \mu_e}[g(\p_C m)] \\
					\iff & 0.5 \cdot f(0) \cdot \left( g(1) \cdot P(C = 1 \mid B = 0) + g(0) \cdot P(C = 0 \mid B = 0) \right) \\
										&+  0.5 \cdot f(1) \cdot \left( g(1) \cdot P(C = 1 \mid B = 1) + g(0) \cdot P(C = 0, B = 1) \right) \\
					<& (0.5 \cdot f(0) + 0.5 \cdot f(1) )  \cdot ( 0.5 \cdot g(1) + 0.5 \cdot g(0)) ,
				\end{align*}
				where $ P(\dots)$ denotes the respective probability that in $\mu_e$.
				Since $\mu_e \in \mu \oplus \mu_C$,
				we have $\mu_e \sqsupseteq \mu$,
				and $\mu$ being a uniform distribution over one-hot vectors on $A, B$ indicates that
				for any $m$ in the support of $\mu_e$, $A = 1$ iff $B = 0$, and $A = 0$ iff $B = 1$.
				Therefore,
				$P(C = v \mid B = 0) = P(C = v \mid A = 1)$ and $P(C = v \mid B = 1) = P(C = v \mid A = 0). $
				Also, by Bayes theorem, we have
				\begin{align*}
					&P(C = v \mid A = 0) \cdot P(A = 0) + P(C = v \mid A = 1) \cdot P(A = 1) = P(C = v)  \\
					\implies & P(C = v \mid A = 0) \cdot 0.5 + P(C = v \mid A = 1) \cdot 0.5 = 0.5 \\
					\iff & P(C = v \mid A = 0) = 1 - P(C = v \mid A = 1) .
			\end{align*}
			Let $X = P(C = 1 \mid A = 1)$, $Y = P(C = 0 \mid A = 1)$, then we have
			\begin{align*}
			& 0.5 \cdot f(0) \cdot \left( g(1) \cdot X + g(0) \cdot Y \right)
			+  0.5 \cdot f(1) \cdot \left( g(1) \cdot (1-X) + g(0) \cdot (1 - Y) \right) \\
					<& (0.5 \cdot f(0) + 0.5 \cdot f(1) )  \cdot ( 0.5 \cdot g(1) + 0.5 \cdot g(0)) \\
					\iff & f(0) \cdot  g(1) \cdot (X - 0.5) + f(0) \cdot  g(0) \cdot (Y - 0.5)
			+ f(1) \cdot g(1) \cdot (0.5-X) + f(1) \cdot  g(0) \cdot (0.5 - Y) < 0 \\
					\iff & (f(0) - f(1)) \cdot  g(1) \cdot (X - 0.5) + (f(0) - f(1)) \cdot  g(0) \cdot (Y - 0.5) < 0 \\
					\iff 			& 0.5 \cdot f(1) \cdot \left( g(1) \cdot X + g(0) \cdot Y \right)
			+  0.5 \cdot f(0) \cdot \left( g(1) \cdot (1-X) + g(0) \cdot (1 - Y) \right) \\
					<& (0.5 \cdot f(0) + 0.5 \cdot f(1) )  \cdot ( 0.5 \cdot g(1) + 0.5 \cdot g(0)) .
			\end{align*}
				Viewing $f$ as a function from $\Mem[A]$ to $\mathbb{R}^+$,
				this is equivalent to
				\begin{align*}
					\mathbb{E}_{m \sim \mu_e}[f(\p_A m) \cdot g(\p_C m)] < \mathbb{E}_{m \sim \mu_e}[f(\p_A m)] \cdot \mathbb{E}_{m \sim \mu_e}[g(\p_C m)] .
				\end{align*}
				The last inequality contradicts~\Cref{eq:AC-NA}.

				Therefore, $B$ and $C$ must be independent in $\mu_e$.
				Hence, $\mu_e \models \apIn{B} \sepand \apIn{C}$, and $\mu \models \phi$.

			\end{proof}

\CounterEx*
	\begin{proof}
		\label{proof:counterex}
Let $A, B, C$ be three variables in $\RanVar$.
Let $\phi = (\apUnif{C}{\{0, 1\}}) \negimp (\apIn{B} \indand \apIn{C})$.
Let $\sigma$ be a deterministic memory where every deterministic variable is undefined,
and $\mu $ be the uniform distribution over one hot vectors on $A, B$.
Then, we claim $(\sigma, \mu) \models \phi$
but $(\sigma, \pi_{\{B,C\}} \mu ) \not\models \phi$.
For $(\sigma, \mu) \models \phi$, it suffices to show that for any $\mu_C$
where $C$'s value is the uniform distribution on $\{0, 1\}$, for any $\mu'
\sqsupseteq \mu$, and $\mu_e \in \mu' \oplus \mu_C$, $B$ and $C$ are
independent in $\mu_e$ according to~\Cref{lemma:counterex}.
The intuition is that $\mu_e$ must satisfies $\{B, C\}$-\CPNA and $\{A, C\}$-\CPNA,
and since $A$'s value is always the opposite of $B$'s value,
$(B, C)$ has to satisfy pairwise independence in $\mu_e$.
To show $(\sigma, \pi_{\{B,C\}} \mu) \not\models \phi$,
we first note that $\pi_{\{B,C\}} \mu = \pi_{\{B\}} \mu$ is a uniform distribution of $0$ and $1$ on $B$.
Let  $\mu_C' \in \DMem{\{C\}}$ be the uniform distribution on $\{0, 1\}$,
$\mu' \in \DMem{\{B, C\}}$ be the uniform distribution over one-hot vectors on $B,C$.
Clearly, $B, C$ are not independent in $\mu'$, so $\mu' \not\models \apIn{B} \indand \apIn{C}$.
Also, $\mu'$ is in $\pi_{\{B, C\}} \mu \oplus \mu_C'$. So $\pi_{B,C} \mu \not\models \apUnif{C}{\{0, 1\}}  \negimp (\apIn{A} \indand \apIn{C})$.
\end{proof}

			\subsection{The proof system of the program logic}
			\text{ }
%

The proof for the soundness of the frame rule relies on the following corollary of
~\Cref{lemma:monotoneclosure2CPNA}.

	\begin{lemma}[Monotone map closure \CPNA (specific case)]
	\label{lemma:monotoneclosureCPNA}
	Let $S, T \subseteq \Var$ be two disjoint sets of variables, and $\mathcal{T}$
	is sub-partition of $T$.  Suppose that $f : \Mem[S] \to \Mem[S']$ is a
	monotonically non-decreasing non-negative function, and $S'$ disjoint from
	$T$.  If $\mu$ with domain $S \cup T$ satisfies $\{S\} \cup
	\mathcal{T}$-\CPNA, then $\dbind(\mu, m \mapsto \dunit(f(\pi_S m)) \oplus
	\dunit(\pi_T m)) \in \DMem{S' \cup T}$ satisfies $\{S'\} \cup
	\mathcal{T}$-\CPNA.
\end{lemma}
\begin{proof}
	We can reduce to~\Cref{lemma:monotoneclosure2CPNA}:
Let $S_1 = S$, $T_1 = S'$,
and $\{S_2, \dots, S_n\} = \{T_2, \dots, T_n\} = \mathcal{T}$.
Let $f_1 = f$ and the rest of $f_i$ to be the identity map.
Note that we have assumed to only work with non-negative values, so identity maps are also monotonically non-decreasing non-negative function.
\end{proof}

\SoundnessProgramlogic*
\begin{proof}
	\label{proof:soundness_programlogic}
	We only prove the cases not already in~\citet{PSL}.
	\begin{description}
		\item [Rule: \textsc{Cond}.]

		For any configuration $(\sigma, \mu) \models \phi$,
		by side-condition that $\models \phi \rightarrow \apDetm{b}$,
		$(\sigma, \mu) \models \apDetm{b}$.
		Thus, $\Sem{b} (\sigma, \mu)$ must be a Dirac distribution.
		Since the commands are well-typed, $\Sem{b} (\sigma, \mu)$ is a distribution over booleans.
		So $\Sem{b} (\sigma, \mu)$ is either a Dirac distribution of truthful value $\ktt$
		or a Dirac distribution of false value $\kff$.

		If $\Sem{b} (\sigma, \mu)$ is a Dirac distribution of truthful value $\ktt$,
		then for any $m$ in the support of $\mu$, $\Sem{b}(\sigma, m) = \ktt$,
    and thus, $(\sigma, \mu) \models \phi \land \apEq{b}{\ktt}$.
    By the side-condition $\vdash \hoare{\phi \land \apEq{b}{\ktt} }{c}{\psi}$
		and inductive hypothesis that this judgement is sound,
		$\Sem{c} (\sigma, \mu) \models \psi$.
		When $\Sem{b} (\sigma, \mu) = \delta(\ktt)$, the semantics say that $ \Sem{\RCond{b}{c}{c'}} (\sigma, \mu) = \Sem{c} (\sigma, \mu) \models \psi$.

		Symmetrically, when $\Sem{b} (\sigma, \mu) = \delta(\kff)$,
 $ \Sem{\RCond{b}{c}{c'}} (\sigma, \mu) = \Sem{c'} (\sigma, \mu) \models \psi$.

\item [Rule: \textsc{Loop}.]
	For any $(\sigma, \mu) \models \phi$, the side condition implies $(\sigma, \mu) \models \apDetm{b}$. We show by induction that for any $n$,
	$\Sem{(\RCondt{b}{c})^n} (\sigma, \mu) \models \phi \land \apDetm{b}$,
	and
  $\Sem{(\RCondt{b}{c})^n} (\sigma, \mu) \models \phi \land \apEq{b}{\kff}$ implies
  $\Sem{(\RCondt{b}{c})^{n + 1}} (\sigma, \mu) \models \phi \land \apEq{b}{\kff}$.

Say $(\sigma', \mu') = \Sem{(\RCondt{b}{c})^n} (\sigma, \mu)$. Assuming 	$(\sigma', \mu') \models \phi \land \apDetm{b}$, there are two possibilities:
	\begin{itemize}
    \item $(\sigma', \mu') \models \phi \land \apEq{b}{\kff}$, then
			\[ \Sem{(\RCondt{b}{c})^{n+1} } (\sigma, \mu) = \Sem{\RCondt{b}{c} }
      (\sigma', \mu') = (\sigma', \mu') \models \phi \land \apEq{b}{\kff}. \]
    \item $(\sigma', \mu') \models \phi \land \apEq{b}{\ktt}$, then
	\[ \Sem{(\RCondt{b}{c})^{n+1} } (\sigma, \mu) = \Sem{\RCondt{b}{c} } (\sigma', \mu')
	= \Sem{c}  (\sigma', \mu') \models \phi, \]
	where the last satisfaction is guaranteed by $\vdash \hoare{\phi \land \apEq{b}{\ktt}}{c}{\phi}$.
	Since $\models \phi \to \apDetm{b}$, so $\Sem{\RCondt{b}{c} } (\sigma', \mu') \models \phi \land \apDetm{b}$.
	\end{itemize}
	Since we assumed that the loop ends in finite step,
	there exists a finite number $N$ such that
  $\Sem{(\RCondt{b}{c})^N} (\sigma, \mu) \models \phi \land \apEq{b}{\kff}$ and
  also $\Sem{(\RCondt{b}{c})^{N -1}} (\sigma, \mu) \models \phi \land \apEq{b}{\ktt}$
	if $N > 1$.

	Then
		$\Sem{\RWhile{b}{c}} (\sigma, \mu) = \Sem{(\RCondt{b}{c})^{N}} (\sigma, \mu)
    \models \phi \land \apEq{b}{\kff}$.

	\item [Rule: \textsc{RCase}.]
		For any $(\sigma, \mu) \models \phi \indand \eta$,
		there exists $\mu', \mu_1, \mu_2$ such that $\mu \sqsupseteq \mu' \in \mu_1 \oplus \mu_2$
		and $(\sigma, \mu_1) \models \phi$, $(\sigma, \mu_2) \models \eta$.

	Say $\mu_2$ is in $\DMem{T}$, then for any $m$ in the support of $\mu_2$,
	the conditional distribution $\mu_2 \mid T = m$ is a Dirac distribution of $m$,
	i.e., $\delta(m)$.
	Since $\eta \in \CCond$ (closed under conditioning),
	so $(\sigma, \delta(m)) \models \eta$,
	and thus $(\sigma, \delta(m)) \models \bigvee_{\alpha \in S} \eta_\alpha$.
	Then, there exists $\alpha$ such that $(\sigma, \delta(m)) \models \eta_\alpha$.
	Since $\dom(\mu_1)$ and $\dom(\mu_2)$ are independent in $\mu'$,
	the conditional distribution $\mu' \mid T = m$ is in $\mu_1 \oplus \delta(m)$.
	So $\mu' \mid T = m \models \phi \indand \eta_\alpha$.
	By the side-condition $\vdash \phi \indand \eta_{\alpha}$ and inductive hypothesis that
	it is sound, we have $\Sem{c} (\sigma,\mu' \mid T = m) \models \psi$.

	For any command $c$, any condition $b$,
	let $(\sigma_b, \mu_b) = \Sem{c} (\sigma, \mu' \mid b)$ and $ (\sigma_b, \mu_{\neg b}) = \Sem{c}(\sigma, \mu' \mid \neg b)$.
	We can show by induction on the semantics of commands that
	$\Sem{c} (\sigma, \mu') $ is a convex combination of $(\sigma, \mu' \mid b)$
	and $(\sigma, \mu' \mid \neg b) $:
	$\Sem{c} (\sigma, \mu') = (\sigma, (\mu' \mid b) \comb_{\mu'(b = \ktt)} (\mu' \mid \neg b) ) $.
	Thus, $\Sem{c} (\sigma,\mu')$ is a convex combination of
	all $\Sem{c} (\sigma,\mu' \mid T = m) $, where $m$ in the support of $\mu'$.
	Since each of $\Sem{c} (\sigma,\mu' \mid T = m)$ satisfies $\psi$ and $\psi$ is
	closed under mixture,
	we have  $\Sem{c} (\sigma,\mu') \models \psi$.
	We can also show  by induction on the semantics of commands that
	$\Sem{c}(\sigma, \mu) \sqsupseteq \Sem{c}(\sigma, \mu')$ if $\mu \sqsupseteq \mu'$.
	So by persistence $\Sem{c}(\sigma, \mu) \models \psi$.

\item [Rule: \textsc{ProbBound}.]
	For any program state $(\sigma, \mu) \models \Pr[ev_1] \geq 1 - \epsilon$,
	let event $ev_1^{\sigma}: \Mem[\dom(\mu)] \to \{0, 1\}$ be the result of partially
	interpreting $ev_1$ on $\sigma$, i.e.,
	$ev_1^{\sigma} = \text{curry}(\Sem{ev_1}) (\sigma)$,
	and denote the function $\lambda x. 1 - ev_1^{\sigma}(x)$ as $\neg ev_1^{\sigma}$.
	We also write $\Pr_{\mu}[ev_1^{\sigma}]$ for
	$\sum_{m \in \DMem{\dom(\mu)}} \mu(m) \cdot ev_1^{\sigma}(m)$.

	We can express $\mu$ as the convex combination of two conditional distributions,
	i.e.,
	\[
		\mu = \Pr_{\mu}[ev_1^{\sigma}] \cdot (\mu \mid ev_1^{\sigma})
	+ \Pr_{\mu}[\neg ev_1^{\sigma}] \cdot (\mu \mid \neg ev_1^{\sigma}).
	\]
	Let $(\sigma', \mu_{ev_1^{\sigma}}) = \Sem{c}(\sigma, (\mu \mid ev_1^{\sigma}))$.
	Since assignments to deterministic memories can only use variables in the
	deterministic memories, there exists probabilistic memories $\mu_{\neg ev_1^{\sigma}}$
	such that
	$(\sigma', \mu_{\neg ev_1^{\sigma}}) = \Sem{c}(\sigma, (\mu \mid \neg ev_1^{\sigma}))$.
	Then, by induction on the denotational semantics,
	we can prove that
	$\Sem{c} (\sigma, \mu) = (\sigma', \mu_{ev_1^{\sigma}} \circ_{\Pr_{\mu}[ev_1^{\sigma}]} \mu_{\neg ev_1^{\sigma}})$.

	By construction, $\Sem{ev_1} (\sigma, (\mu \mid ev_1^{\sigma})) = 1$,
	so $ (\sigma, (\mu \mid ev_1^{\sigma})) \models ev_1$.
	Also, by the assumption and inductive hypothesis, we have
	$\models \hoare{ev_1}{c}{\Pr[ev_2] \geq 1- \delta}$,
	which implies
	\[\Sem{c}(\sigma, (\mu \mid ev_1^{\sigma})) = (\sigma', \mu_{ev_1^{\sigma}}) \models \Pr[ev_2] \geq 1 - \delta.\]
	By definition, that means
	$\Pr_{(\sigma', \mu_{ev_1^{\sigma}})}[ev_2] \geq 1 - \delta $.
	Then, by the law of total probability,
	\begin{align*}
		\Pr_{(\sigma', \mu_{ev_1^{\sigma}} \circ_{\Pr_{\mu}[ev_1^{\sigma}]} \mu_{\neg ev})} [ev_1^{\sigma}]
		&\leq \Pr_{(\sigma', \mu_{ev_1^{\sigma}})} [ev_2] + \Pr_{\mu}[\neg ev_1^{\sigma}] \\
		&\leq \delta +  \Pr_{\sigma, \mu}[\neg ev_1] \tag{because $\Pr_{\mu}[\neg ev_1^{\sigma}] = \Pr_{\sigma, \mu}[\neg ev_1]$}\\
		&\leq \delta +  \epsilon  \tag{because $(\sigma, \mu) \models \Pr[ev_1] \geq 1 - \epsilon$,}
	\end{align*}

	Therefore,
	$\Sem{c} (\sigma, \mu)\models \Pr[ev_2] \leq \delta + \epsilon$.

		\item [Rule: \textsc{NegFrame}.]  For any $(\sigma, \mu) \models ev_1 \sepand \eta$,
	there exists $k_1, k_2, \mu'$ such that $\mu \sqsupseteq \mu' \in k_1 \oplus k_2$, and
	$(\sigma, k_1)	\models \phi$ and $(\sigma, k_2) \models \eta$.

	Let $S_1 \triangleq \dom(k_1)$,
	and note that $(\sigma, k_1) \models \phi$ and $\models \phi \rightarrow \apIn{\RV(c)}$
	implies $(\sigma, k_1) \models \apIn{\RV(c)}$,
	and thus $\RV(c) \cap \RanVar \subseteq S_1$.
	Let $S_2 \triangleq \dom(k_2) \cap \FV(\eta)$.
	Then, $\MV(c)$  is disjoint from $S_2$ because $S_2 \subseteq \FV(\eta)$ and $\FV(\eta) \cap \MV(c) = \emptyset$;
	also, by restriction,  $(\sigma, \pi_{S_2} k_2) \models \eta$.

	Since $k_1 \oplus k_2$ is non-empty, $dom(k_1), dom(k_2)$ are disjoint;
	since $S_1 = dom(k_1), S_2 \subseteq dom(k_2)$, $S_1, S_2$ are disjoint.

	Let $(\sigma_e, \mu_e) = \denot{c} (\sigma, \mu)$.
	Denote $\RV(c) \cap \DetVar$ as $R_1$, $\MV(C) \cap \DetVar$ as $M_1$,
 $\RV(c) \cap \RanVar$ as $R_2$, $\MV(C) \cap \RanVar$ as $M_2$.
	By the soundness of $\RV, \WV$, and $\MV$,
	there exists $G: \Mem[R_1] \to \Mem[M_1]$,
	$F : \Mem[\RV(c)]	\to \DMem{M_2}$ such that:
	\begin{align*}
		\sigma_e &= G(\p_{R_1} \sigma) \bowtie \p_{\DetVar \setminus M_1} \sigma \\
		\mu_e &= \dbind(\mu, m \mapsto F(\p_{R_1} \sigma \bowtie \p_{R_2} m) \otimes  \dunit(\p_{\dom(\mu) \setminus M_2} m) ) .
	\end{align*}
	Since $R_2 \subseteq S_1$, and $S_2$ is disjoint of $S_1$ and $\MV(c)$,
	we have
	\begin{align}
							\pi_{M_2 \cup S_1 \cup S_2} \mu_e
							&= \dbind(\pi_{S_1 \cup S_2}  \mu, (m_1, m_2) \mapsto F(\p_{R_1} \sigma \bowtie \p_{R_2} m_1) \otimes \dunit(\p_{S_1 \setminus M_2} m_1) \otimes \dunit(m_2)) . \label{eq:Fmap}
	\end{align}

	One implication is that $(\sigma_e, \pi_{S_2} \mu_e) \models \eta$:
	~\Cref{eq:Fmap}  implies that $\pi_{S_2} \mu_e = \pi_{S_2} \mu$.
	We $(\sigma, \pi_{S_2} \mu ) \models \eta$, so  $(\sigma, \pi_{S_2} \mu_e ) \models \eta$.
			By restriction, for any $m \in \Mem[\DetVar \setminus \FV(\eta)]$,
			$(\p_{\FV(\eta)} \sigma \bowtie m, \pi_{S_2} \mu_e) \models \eta$.
			Since $	\sigma_e = G(\p_{R_1} \sigma) \bowtie \p_{\DetVar \setminus M_1} \sigma $,
			and $\FV(\eta) \cap \MV(c) = \emptyset$ implies that $\FV(\eta) \cap \DetVar \subseteq \DetVar \setminus M_1$,
			and $(\sigma_e, \pi_{S_2} \mu_e) \models \eta$.

	If $y \in \DetVar$, then $ \Sem{y} (\sigma_e, m_F \bowtie \p_{X \cap \RanVar} m ) = \sigma_e (y)$.
	Thus, $(\sigma_e, t) \models \apIn{y}$ where $t$ is the trivial distribution in $\DMem{\emptyset}$.
	Also, $(\sigma_e, \mu_e) \in (\sigma_e, t) \circ (\sigma_e, \mu_e)$,
	where $ (\sigma_e, \mu_e) \sqsupseteq (\sigma_e, \pi_{S_2} \mu_e) \models \eta$.
	So $(\sigma_e, \mu_e) \models \eta$.

	If $y \notin \MV(c)$, then $\{y\} \cup S_2 \subseteq \dom(\mu) \setminus \MV(c)$,
	and thus $\pi_{\{y\} \cup S_2}  \mu_e = \pi_{\{y\} \cup S_2} \mu$.
	\begin{itemize}
    \item Since $\vdash \hoare{\phi}{c}{\apEq{y}{f(X)}}$, by inductive assumption we have $\denot{c} (\sigma, k_1) \models \apEq{y}{f(X)}$.
	Updates to deterministic variables only depend on deterministic program state, so $\denot{c} (\sigma, k_1)  = (\sigma_e, \mu_k)$ for some $\mu_k$.
	And $y \notin \MV(c)$ implies that $ \pi_{\{y \}} \mu_k =  \pi_{\{y \}} \mu = \pi_{\{y\}} k_1$.
  The restriction property and $\denot{c} (\sigma, k_1) \models \apEq{y}{f(X)}$ implies $(\sigma_e, \pi_{\{y\}} k_1) = (\sigma_e, \pi_{\{y\}} \mu_k) \models \apIn{y}$.
\item $(\sigma_e, \pi_{S_2} \mu) \models \eta$.
\item  $\mu \sqsupseteq \mu' \in k_1 \oplus k_2$, so $\pi_{\{y\} \cup S_2} \mu \in \pi_{\{ y \}}k_1 \oplus \pi_{S_2} k_2$ too.
	\end{itemize}
	Therefore, $(\sigma_e, \pi_{\{y\} \cup S_2} \mu) \models \apIn{y} \negand \eta $.
	Since $ (\sigma_e, \mu_e)  \sqsupseteq (\sigma_e,  \pi_{\{y\} \cup S_2} \mu_e)  = (\sigma_e,  \pi_{\{y\} \cup S_2} \mu)$,
	by persistence, $  (\sigma_e, \mu_e) \models \apIn{y} \negand \eta$.

	If $ y \in \RanVar$ and $y \in \MV(c)$,
	our overall strategy is to first connect $F$ with $f$ and show the operation on variable $y$ is a monotone map,
	and then apply monotone map closure to establish the NA between $y$ and $\eta$.

	Since $(\sigma, \pi_{S_1}\mu)
	= (\sigma, k_1) \models \phi$, by persistence $(\sigma, \mu) \models \phi$.
  By side-condition that $\vdash \{\phi\} c \{ \apEq{y}{f(X)} \}$ and by induction
	that the proof rules are sound, it must $\denot{c}(\sigma,
  \mu) =  (\sigma_e, \mu_e)  \models \apEq{y}{f(X)}$.  By restriction, $(\sigma_e, \pi_{X \cup \{y\}} \mu_e)
  \models \apEq{y}{f(X)}$.

	Since $X \cap \RanVar \subseteq (\RV(c) \cap \RanVar) \setminus \MV(c) \subseteq S_1
	\setminus \MV(c)$ and $y \in \MV(c)$,
	\begin{align}
		\pi_{X \cup \{y\}}  \mu_e &= \pi_{X \cup \{y\}} \pi_{ S_1 \cup \MV(c)} \denot{c} \mu \notag\\
																																				&= \pi_{X \cup \{y\}} \dbind(\pi_{S_1}  \mu, m \mapsto F(\p_{R_1} \sigma \bowtie \p_{R_2} m) \otimes \dunit(\p_{S_1 \setminus M_2} m)) \notag \\
																																				&=\dbind(\pi_{S_1}  \mu,  m \mapsto \pi_{\{y\}} F(\p_{R_1} \sigma \bowtie \p_{R_2} m) \otimes \dunit(\p_{X \cap \RanVar} m)) 																										\label{eq:projecty}
	\end{align}

  Since $(\sigma_e, \pi_{X \cup \{y\}}  \mu_e)  \models \apEq{y}{f(X)}$, for every $m_e$ in the support of $ \pi_{X \cup \{y\}}  \mu_e$,
	we have $\Sem{y} (\sigma_e, m_e) = \Sem{f(X)} (\sigma_e, m_e)$.
	By~\Cref{eq:projecty}, a memory $m_e$ is in the support of $ \pi_{X \cup \{y\}}  \mu_e$ if and only if there exists some $m, m_F$
	such that $m$ is in the support of $\mu$, and $m_F$ is in the support of $\pi_{\{y\}} F(\p_{R_1} \sigma \bowtie \p_{R_2} m)$,
and $m_e = m_F \bowtie \p_{X \cap \RanVar} m$.
	Thus, the condition we have is:
	for every  $m$ is in the support of $\mu$ and $m_F$ is in the support of $\pi_{\{y\}} F(\p_{R_1} \sigma \bowtie \p_{R_2} m)$,
	\[ \Sem{y} (\sigma_e,m_F \bowtie \p_{X \cap \RanVar} m )
	= \Sem{f(X)} (\sigma_e, m_F \bowtie \p_{X \cap \RanVar} m ) . \]
Since $f$ does not depend on states and $X$ do not depend on $m_F$,
we also have that $\Sem{f(X)} (\sigma_e, m_F \bowtie \p_{X \cap \RanVar} m ) =   \Sem{f(X)} (\sigma, m)$.

	If $y \in \RanVar$, then $\Sem{y} (\sigma_e,m_F \bowtie \p_{X \cap \RanVar} m) = m_F(y)$,
	so it must $m_F(y) = \Sem{f(X)} (\sigma, m) $.
		so although $F$ is a randomized function according to its type,
	for any $m$ in the support of $\mu$,
$\pi_{\{y\}} F(\p_{R_1} \sigma \bowtie \p_{R_2} m) $  is a Dirac distribution:
\[ \pi_{\{y\}} F(\p_{R_1} \sigma \bowtie \p_{R_2} m) = \delta(\Sem{f(X)} (\sigma, m)). \]
Fixing $\sigma$, then there exists $f'$ such that $ \Sem{f(X)} (\sigma, m) = f'(\p_{X \cap \RanVar} m)$
and $f'$ is monotone as $f$ is monotone. Since $X \cap \RanVar \subseteq S_1 $,
we can also make $f'$ to have type $\Mem[S_1] \to \Val$.

Thus,
	\begin{align*}
							\pi_{\{y\} \cup S_2} \mu_e
							&= \dbind(\pi_{S_1 \cup S_2}  \mu, (m_1, m_2) \mapsto  \pi_y F(\p_{R_1} \sigma \bowtie \p_{R_2} m_1) \otimes \dunit(m_2))\\
							&= \dbind(\pi_{S_1 \cup S_2}  \mu, (m_1, m_2) \mapsto  \dunit(\Sem{f(X)} (\sigma, m_1)) \otimes \dunit(m_2))\\
							&= \dbind(\pi_{S_1 \cup S_2}  \mu, (m_1, m_2) \mapsto  \dunit(f'(m_1)) \otimes \dunit(m_2)) .
	\end{align*}
	Let $	\mu_c = \pi_{\{y\} \cup S_2} (\mu_e)$.
	We will then show that
	$(\sigma_e, \mu_c) \models \apIn{y} \sepand \eta. $

	Let $g_1 = \pi_{y} (\mu_e)$, $g_2 = \pi_{S_2} \mu $,
	it suffices to show that $\mu_c \in g_1 \oplus g_2$.
	and $g_1 \models \apIn{y} $,
	$g_2 \models \eta$:
	\begin{itemize}
    \item $(\sigma_e, \mu_e) \models \apEq{y}{f(X)}$, so $(\sigma_e, \mu_e) \models \apIn{y}$.
			By restriction, $(\sigma_e, g_1) \models \apIn{y}$.
		\item $g_2 = \pi_{S_2} \mu = \pi_{FV(\eta)} \pi_{dom(k_2)} \mu = \pi_{FV(\eta)} k_2 $,
			so $(\sigma, g_2) \models \eta$.
			By~\Cref{lemma:det_restriction}, for any $m \in \Mem[\DetVar \setminus \FV(\eta)]$,
			$(\p_{\FV(\eta)} \sigma \bowtie m, g_2) \models \eta$.
			Since $	\sigma_e = G(\p_{R_1} \sigma) \bowtie \p_{\DetVar \setminus M_1} \sigma $,
			and $\FV(\eta) \cap \MV(c) = \emptyset$ implies that $\FV(\eta) \subseteq \DetVar \setminus M_1$,
			$(\sigma_e, g_2) \models \eta$.

		\item
			First,
			$\pi_{\{ y\}} \mu_c = \pi_{\{ y\}} \mu_e = g_1$,
			and $\pi_{S_2} \mu_c = \pi_{S_2} \mu_e = \pi_{S_2} \mu = g_2$.

			Second, $\pi_{dom(k_1) \cup dom(k_2)} \mu \in k_1 \oplus k_2$
		 implies that $\pi_{dom(k_1) \cup dom(k_2)} \mu$ is
			$\mathcal{S}' \cup \mathcal{T}'$-\CPNA for any $\mathcal{S}', \mathcal{T}'$ such that $k_1$ is $\mathcal{S}'$-\CPNA, $k_2$ is $\mathcal{T}'$-\CPNA.
			Because $\pi_{S_1} \mu$ is always $\{S_1\}$-\CPNA,
			$\pi_{S_1 \cup S_2}\mu$ is $\{S_1\} \cup \mathcal{T}$-\CPNA for  any $\mathcal{T}$ such that $g_2$ is $\mathcal{T}$-\CPNA.
			Recall that
			\begin{align*}
				\mu_c = \dbind(\pi_{S_1 \cup S_2}  \mu, (m_1, m_2) \mapsto  \dunit(f'(m_1)) \otimes \dunit(m_2)) .
			\end{align*}
			Thus, by the monotonicity map closure~\Cref{lemma:monotoneclosureCPNA},
			$\mu_c$ is $\{y\} \cup \mathcal{T}$ -\CPNA for  any $\mathcal{T}$ such that $g_2$ is $\mathcal{T}$-\CPNA.
			Thus, $\mu_c \in g_1 \oplus g_2$, and therefore $(\sigma_e, \mu_c) \in (\sigma_e, g_1) \oplus (\sigma_e, g_2) $.
	\end{itemize}
	Therefore, $(\sigma_e, \mu_c) \models \apIn{y} \indand \eta$.

 By persistence,  $(\sigma_e, \mu_e) \models \apIn{y} \indand \eta$.
	\end{description}
\end{proof}

\section {Completeness of $M$-BI}
\label{sec:app:completeness}

\subsection{$\idxset$-BI Algebras}

\begin{definition}[BI Algebra] \label{def:bi-algebra}
  An BI algebra is an algebra $\mathcal{A} = (A, \land, \lor, \rightarrow, \top, \bot, \sepand, \sepimp, \top^{\sepand})$ such that
  \begin{itemize}
    \item $(A, \land, \lor, \rightarrow, \top, \bot)$ is a Heyting algebra
    \item $(A, \sepand, \top^{\sepand})$ is a commutative monoid
    \item $a \sepand b \leq c$ if and only if $a \leq b \sepimp c$
  \end{itemize}
  where $\leq$ is the ordering associated with the Heyting algebra.
\end{definition}

\begin{definition}[$\idxset$-BI Algebra] \label{def:mbi-algebra}
  An $\idxset$-BI algebra is an algebra $\mathcal{A} = (A, \land, \lor, \rightarrow, \top, \bot, \sepand_{\idx \in \idxset}, \sepimp_{\idx \in \idxset}, \top^{\sepand}_{\idx \in \idxset})$ such that
  \begin{itemize}
    \item For each $\idx \in \idxset$, the structure $(A, \land, \lor, \rightarrow, \top, \bot, \sepand_{\idx}, \sepimp_{\idx}, \top^{\sepand}_{\idx})$ is a BI-algebra
    \item If $\idx_1 \leq \idx_2$ then $a \sepand_{\idx_1} b \leq a \sepand_{\idx_2} b$
  \end{itemize}
\end{definition}

We can interpret $M$-BI in an $M$-BI algebra $A$. Let $\valuation :
\mathcal{AP} \rightarrow A$ be a map assigning atomic propositions to elements
of $A$. We extend $\valuation$ to an interpretation $\llbracket - \rrbracket_{\valuation}$
mapping propositions to elements of $A$, defined by:

\newcommand{\alginterp}[1]{\llbracket #1 \rrbracket}

\begin{align*}
  \alginterp{p} &= \valuation(p) \\
  \alginterp{\top} &= \top \\
  \alginterp{\sepid_\idx} &= \top^{\sepand}_{\idx} \\
  \alginterp{\bot} &= \bot \\
  \alginterp{P \land Q}_\valuation &= \alginterp{P}_\valuation \land \alginterp{Q}_\valuation \\
  \alginterp{P \lor Q}_\valuation &= \alginterp{P}_\valuation \lor \alginterp{Q}_\valuation \\
  \alginterp{P \rightarrow Q}_\valuation &= \alginterp{P}_\valuation \rightarrow \alginterp{Q}_\valuation \\
  \alginterp{P \sepand_{\idx} Q}_\valuation &= \alginterp{P}_\valuation \sepand_{\idx} \alginterp{Q}_\valuation \\
  \alginterp{P \sepimp_{\idx} Q}_\valuation &= \alginterp{P}_\valuation \sepimp_{\idx} \alginterp{Q}_\valuation \\
\end{align*}

\begin{theorem}[Algebraic Soundness] If $P \vdash Q$ is provable, then for all $\valuation$, $\alginterp{P}_\valuation \leq \alginterp{Q}_\valuation$.\end{theorem}
\begin{proof}
  By induction on the derivation of $P \vdash Q$. The cases for everything except the \textsc{Inclusion} rules show follow from the exact same argument as for standard BI and BI-algebra, as in Simon Docherty's thesis.

  For the remaining case of $\sepand$-\textsc{inclusion}, let $\idx_1 \leq \idx_2$. Then we have
    \[\alginterp{P \sepand_{\idx_1} Q}_\valuation = \alginterp{P}_\valuation \sepand_{\idx_1} \alginterp{Q}_\valuation \leq \alginterp{P}_\valuation \sepand_{\idx_2} \alginterp{Q}_\valuation \leq \alginterp{P \sepand_{\idx_2} Q}_\valuation\]
\end{proof}

\begin{definition}[Lindenbaum-Tarski Algebra] The Lindenbaum-Tarski algebra corresponding to $\idxset$-BI is the set of all equivalence classes of interprovable propositions. That is, define the equivalence relation $P \sim Q$ as $P \vdash Q$ and $Q \vdash P$. We will show that the set of equivalence classes of this relation forms an $\idxset$-BI algebra. Let $[P]_\sim$ be the equivalence class of $P$ under $\sim$.
  Take $\sepid_{\idx}$, $\top$, and $\bot$ to be $[\sepid_{\idx}]_{\sim}$, $[\top]_{\sim}$, and $[\bot]_{\sim}$, respectively. Then we define:
  \begin{align*}
    [P]_{\sim} \land [Q]_\sim &= [P \land Q]_\sim \\
    [P]_{\sim} \lor [Q]_\sim &= [P \lor Q]_\sim \\
    &\vdots \\
    [P]_{\sim} \sepand_{\idx} [Q]_\sim &= [P \sepand_{\idx} Q]_\sim \\
    [P]_{\sim} \sepimp_{\idx} [Q]_\sim &= [P \sepimp_{\idx} Q]_\sim \\
  \end{align*}
  The fact that these operations are well-defined and form a $\idxset$-BI algebra follows almost entirely from the corresponding result for normal BI outlined in Docherty's thesis. The only remaining case is to check that if $\idx_1 \leq \idx_2$ then $[P]_{\sim} \sepand_{\idx_1} [Q]_{\sim} \leq [P]_{\sim} \sepand_{\idx_2} [Q]_\sim$. We have
  \begin{align*}
    [P]_{\sim} \sepand_{\idx_1} [Q]_\sim
      &= [P \sepand_{\idx_1} Q]_{\sim} \\
						&\leq [P \sepand_{\idx_2} Q]_{\sim} \tag{Since $P \sepand_{\idx_1} Q \leq P \sepand_{\idx_2} Q$} \\
      &= [P]_{\sim} \sepand_{\idx_2} [Q]_{\sim}
  \end{align*}
\end{definition}

\begin{lemma} $P \vdash Q$ if and only if $[P]_{\sim} \leq [Q]_{\sim}$. \end{lemma}
\begin{proof}
  In the proof that the Lindenbaum-Tarski algebra indeed formed an $\idxset$-BI algebra, we already showed that $P \vdash Q$ implies $[P]_{\sim} \leq [Q]_{\sim}$. Consider the opposite direction. Then we have that $[P]_{\sim} \land [Q]_{\sim} = [P]_{\sim}$, hence $[P \land Q]_{\sim} = [P]_{\sim}$. This implies that $P \land Q \dashv\vdash P$.  Since $P \land Q \vdash Q$, by transitivity we have $P \vdash Q$.
\end{proof}

\begin{theorem}[Algebraic Completeness] If $\alginterp{P}_\valuation \leq \alginterp{Q}_\valuation$ for all $\valuation$, then $P \vdash Q$.\end{theorem}
\begin{proof}
  Consider the valuation $\valuation$ which maps $p \in \mathcal{AP}$ to $[p]_\sim$. Then $\alginterp{P}_\valuation = [P]_{\sim}$ and $\alginterp{Q}_\valuation = [Q]_{\sim}$. Hence we have
  $[P]_{\sim} \leq [Q]_{\sim}$ which implies $P \vdash Q$.
\end{proof}

\subsection{$\idxset$-BI Frames}
$\idxset$-BI formulas are interpreted on Down-Closed $\idxset$-BI frames.
We define a complex algebra on \LOGIC frames.

\begin{definition}[Complex Algebra] \label{def:complex-algebra}
  If $\mathcal{X}$ is an $\idxset$-BI frame, then the \emph{complex algebra} of $\mathcal{X}$, written $\complex(\mathcal{X})$ is the structure $(\mathcal{P}_{\sqsubseteq}(X), \cap, \cup, \rightarrow_{\mathcal{X}}, X, \emptyset, \sepand_{\idx \in \idxset}, \sepimp_{\idx \in \idxset}, E_{\idx \in \idxset})$ where
  \begin{align*}
    \mathcal{P}_{\sqsubseteq}(X) &= \{ A \subseteq X \ | \ a \in A \land a \sqsubseteq b \rightarrow b \in A \} \\
    A \rightarrow_{\mathcal{X}} B &= \{ a \ | \ \forall b.\, a \sqsubseteq b \land b \in A \rightarrow b \in B \} \\
    A \sepand_{\idx} B &= \{ x \ | \ \exists w, y, z.\, w \sqsubseteq x \land w \in y \oplus_{\idx} z \land y \in A \land z \in B \} \\
    A \sepimp_{\idx} B &= \{ x \ | \ \forall w, y, z.\, (x \sqsubseteq w \land z \in w \oplus_{\idx} y \land y \in A) \rightarrow z \in B \}
  \end{align*}

\end{definition}

\begin{lemma} If $\mathcal{X}$ is an $\idxset$-BI frame, then $\complex(\mathcal{X})$ is an $\idxset$-BI algebra. \end{lemma}
\begin{proof}
  Each $(X, \sqsubseteq, \oplus_\idx, E_{\idx})$ is a
  BI frame. \citet{docherty:thesis} shows that the complex of a BI frame is a
  BI algebra. Thus the only thing to check is that the ordering on $\sepand$
  respects the ordering on $\idxset$.  Let $\idx_1 \leq \idx_2$. We must show
  that $A \sepand_{\idx_1} B \subseteq A \sepand_{\idx_2} B$. Let $x \in A
  \sepand_{\idx_1} B$. Then there exists $w, y, z$ such that $w \sqsubseteq x$
  and $w \in y \oplus_{\idx_1} z$, with $y \in A$ and $z \in B$.
  by the \text{Operation Inclusion} property, we have that $w \in y \oplus_{\idx_2} z$, hence
  $x \in A \sepand_{\idx_2} B$.
\end{proof}

\begin{definition}[Prime Filter]
  If $(L, \land, \lor)$ is a bounded distributive lattice, a prime filter on $F$ is a non-empty proper subset of $A$ such that:
  \begin{itemize}
  \item If $x \in F$ and $x \leq y $ then $y \in F$.
  \item If $x \in F$ and $y \in F$ then $x \land y \in F$.
  \item If $x \lor y \in F$ then $x \in F$ or $y \in F$.
  \end{itemize}
\end{definition}
We write $\prf(L)$ for the set of prime filters on $L$.

\begin{definition}[Prime Filter Frame]
  If $\mathcal{A}$ is an $\idxset$-BI algebra, then the prime filter $\idxset$-frame of $\mathcal{A}$ is defined as $\prf(\mathcal{A}) = (\prf(A), \subseteq, \oplus_{\idx \in \idxset}, E_{\idx \in \idxset})$
  where
  \begin{align*}
    F_1 \oplus_{\idx} F_2 &= \{ F \in \prf(A) \ | \ \forall a_1 \in F_1.\, \forall a_2 \in F_2.\, a_1 \sepand_{\idx} a_2 \in F \} \\
    E_{\idx} &= \{ F \in \prf(A) \ | \ \top^{\sepand}_{\idx} \in F \}
  \end{align*}
\end{definition}

\begin{lemma}
  If $\mathcal{A}$ is an $\idxset$-BI algebra, then $\prf(\mathcal{A})$ is an $\idxset$-BI frame.
\end{lemma}
\begin{proof}
  \citet{docherty:thesis} shows that for each $\idx \in \idxset$, $(\prf(A), \subseteq, \oplus_{\idx}, E_{\idx})$ is a BI frame. Therefore, we only need to check the Operation Inclusion property. Let $\idx_1 \leq \idx_2$
  and let $F, G, H \in \prf(A)$ with $F \in G \oplus_{\idx_1} H$.
  Let $a \in G$ and $b \in H$. Then $a \sepand_{\idx_1} b \in F$. Since $a \sepand_{\idx_1} b \leq a \sepand_{\idx_2} b$, and filters are upward-closed,
  $a \sepand_{\idx_2} b \in F$, hence $F \in G \oplus_{\idx_2} H$.

\end{proof}

\begin{theorem}[Representation Theorem]
  Every $\idxset$-BI algebra is isomorphic to a subalgebra of a complex algebra. In particular, if $\mathcal{A}$ is an $\idxset$-BI algebra,
  then the map $\theta : \mathcal{A} \rightarrow \complex(\prf(\mathcal{A}))$ defined as
  \[ \theta(x) = \{ F \in \prf(\mathcal{A}) \ | \ x \in F \} \]
  is an embedding.
\end{theorem}
\begin{proof}
  \citet{docherty:thesis} proves that for each $\idx \in \idxset$, this map $\theta$ is an embedding of $(A, \land, \lor, \rightarrow, \top, \bot, \sepand_{\idx}, \sepimp_{\idx}, \top^{\sepand}_{\idx})$ as a BI algebra into the complex algebra, viewed as a BI algebra for the operations indexed by $\idx$. Hence, $\theta$ is injective and a homomorphism with respect to all of the $\idxset$-BI algebra operations.
\end{proof}

\begin{theorem}[Equivalence of Algebras and Frames]
  Let $\mathcal{A} = (A, \dots)$ be an $\idxset$-BI algebra and let $\valuation_{\textsf{a}} : \mathcal{AP} \rightarrow A$ be an interpretation of atomic propositions. Let $\mathcal{X} = (X, \dots)$ be an $\idxset$-BI frame and let $\valuation_{\textsf{f}} : \PROP \rightarrow \mathcal{P}(X)$ be a persistent valuation on $\mathcal{X}$. Let $\theta$ be the embedding from the previous result. Define the persistent valuation $\valuation_{\textsf{a}}' : \mathcal{AP} \rightarrow \mathcal{P}(\prf(A))$ and the interpretation $\valuation_{\textsf{f}}' : \mathcal{AP} \rightarrow \complex(\mathcal{X})$ by:
  \begin{align*}
    \valuation_{\textsf{a}}'(p) &= \theta(\valuation_{\textsf{a}}(p)) \\
    \valuation_{\textsf{f}}'(p) &= \valuation_{\textsf{f}}(p) .
  \end{align*}
Then we have
\begin{enumerate}
\item $x \models_{\valuation_{\textsf{f}}} P$ if and only if $x \in \alginterp{P}_{\valuation_{\textsf{f}}'}$
\item $F \models_{\valuation_{\textsf{a}}'} P$ if and only if $\alginterp{P}_{\valuation_{\textsf{a}}} \in F$ .
\end{enumerate}
\end{theorem}
\begin{proof}
For the first part, we proceed by induction on $P$.
\newcommand{\lra}{\leftrightarrow}
\begin{itemize}
\item Case $P = p$: We have:
  \begin{align*}
    x \models_{\valuation_{\textsf{f}}} p
    & \lra x \in \valuation_{\textsf{f}}(p) \\
    & \lra x \in \valuation_{\textsf{f'}}(p) \\
    & \lra x \in \alginterp{p}_{\textsf{f'}}
   \end{align*}
\item Case $P = \top$: Then $x \models_{\valuation_{\textsf{f}}} \top$ holds for all $x$, and $\top$ in $\complex(\mathcal{X})$ is defined to be $X$, so $x \in \alginterp{\top}_{\valuation_{\textsf{f}}'}$ holds for all $x$.
\item Case $P = \sepid_{\idx}$: Similar to $P = \top$.
\item Case $P = \bot$: Similar to $P = \top$.
\item Case $P = Q_1 \land Q_2$:
  \begin{align*}
			x \models_{\valuation_{\textsf{f}}} Q_1 \land Q_2
    & \lra x \models_{\valuation_{\textsf{f}}} Q_1 \textrm{ and } x \models_{\valuation_{\textsf{f}}} Q_2  \tag{By satisfication rule}\\
				& \lra x \in \alginterp{Q_1}_{\valuation_{\textsf{f}'}} \textrm{ and } x \in \alginterp{Q_2}_{\valuation_{\textsf{f}'}} \tag{$\valuation_{\textsf{f}'}$ and $\valuation_{\textsf{f}'}$ the same} \\
				& \lra x \in \alginterp{Q_1}_{\valuation_{\textsf{f}'}} \cap \alginterp{Q_2}_{\valuation_{\textsf{f}'}} \\
				& \lra x \in \alginterp{Q_1 \land Q_2}_{\valuation_{\textsf{f}'}} \tag{By the $\land$ operation in Complex algebra and the recursive definition of $\valuation$}
  \end{align*}
\item Case $P = Q_1 \lor Q_2$:
  \begin{align*}
    x \models_{\valuation_{\textsf{f}}} Q_1 \lor Q_2
    & \lra x \models_{\valuation_{\textsf{f}}} Q_1 \textrm{ or } x \models_{\valuation_{\textsf{f}}} Q_2 \\
    & \lra x \in \alginterp{Q_1}_{\valuation_{\textsf{f}'}} \textrm{ or } x \in \alginterp{Q_2}_{\valuation_{\textsf{f}'}} \\
    & \lra x \in \alginterp{Q_1}_{\valuation_{\textsf{f}'}} \cup \alginterp{Q_2}_{\valuation_{\textsf{f}'}} \\
    & \lra x \in \alginterp{Q_1 \lor Q_2}_{\valuation_{\textsf{f}'}}
  \end{align*}
\item Case $P = Q_1 \rightarrow Q_2$: Let $x \models_{\valuation_{\textsf{f}}} Q_1 \rightarrow Q_2$. Then, for all $y$ such that $x \sqsubseteq y$,
  if $y \models_{\valuation_{\textsf{f}}} Q_1$, then $y \models_{\valuation_{\textsf{f}}} Q_2$. Applying the induction hypothesis, we have that for all $y$ such that $x \sqsubseteq y$,
  if $y \in \alginterp{Q_1}_{\valuation_{\textsf{f}'}}$, then
   $y \in \alginterp{Q_2}_{\valuation_{\textsf{f}'}}$.
  Hence, $x \in \alginterp{Q_1 \rightarrow Q_2}_{\valuation_{\textsf{f}'}}$. The reverse direction is similar.
\item Case $P = Q_1 \sepand_{\idx} Q_2$: Let $x \models_{\valuation_{\textsf{f}}} Q_1 \sepand_{\idx} Q_2$. Then there exists $x'$, $y$, and $z$ such that $x' \sqsubseteq x$ and $x' \in y \oplus_{\idx} z$, where $y \models_{\valuation_{\textsf{f}}} Q_1$ and $z \models_{\valuation_{\textsf{f}}} Q_2$. By the induction hypothesis, we have that $y \in \alginterp{Q_1}_{\valuation_{\textsf{f}'}}$ and $z \in \alginterp{Q_2}_{\valuation_{\textsf{f}'}}$. Hence, $x \in \alginterp{Q_1 \sepand_{\idx} Q_2}_{\valuation_{\textsf{f}'}}$.
\item Case $P = Q_1 \sepimp_{\idx} Q_2$: Let $ x \models_{\valuation_{\textsf{f}}} Q_1 \sepimp_{\idx} Q_2$. Then for all $x'$, $y$, and $z$ such that $x \sqsubseteq x'$ and $z \in x' \oplus_{\idx} y$, if $y \models_{\valuation_{\textsf{f}}} Q_1$,  then $z \models_{\valuation_{\textsf{f}}} Q_2$.

  To show that $x \in \alginterp{Q_1 \sepimp_{\idx} Q_2}_{\valuation_{\textsf{f}'}}$, let $w$, $y$, and $z$ be such that $x \sqsubseteq w$, $z \in w \oplus_{\idx} y$, and $y \in \alginterp{Q_1}_{\valuation_{\textsf{f}'}}$. We must show that $z \in \alginterp{Q_2}_{\valuation_{\textsf{f}'}}$. Applying the induction hypothesis, we have that $y \in \models_{\valuation_{\textsf{f}}} Q_1$. Thus, by the preceding paragraph, we have that $z \models_{\valuation_{\textsf{f}}} Q_2$. Applying the induction hypothesis again, we get that $z \in \alginterp{Q_2}_{\valuation_{\textsf{f}'}}$.
\end{itemize}
For the second part, assuming $F \in \prf(\mathcal{A})$,
then for any $P$, we have
		\begin{align*}
			F \models_{\valuation_{\textsf{a}}'} P
			& \lra F \in \valuation_{\textsf{a}'}(P) \\
			& \lra F \in  \prf(\mathcal{A}) \text{ and } \valuation_{\textsf{a}}(P) \in F \\
			& \lra \valuation_{\textsf{a}}(P) \in F \\
			& \lra \alginterp{P}_{\valuation_{\textsf{a}}} \in F .
		\end{align*}
\end{proof}

\begin{theorem}[Completeness]
  If $P \models_\valuation Q$ for all $\valuation$, then $P \vdash Q$
\end{theorem}
\begin{proof}
  Suppose $P \not\vdash Q$. Then by algebraic completeness, there exists some
  $\idxset$-BI algebra $\mathcal{A}$ and an interpretation
  $\valuation_{\textsf{a}}$ such that $ \alginterp{P}_{\valuation_{\textsf{a}}}
  \nleq \alginterp{Q}_{\valuation_{\textsf{a}}}$. By the prime filter theorem,
  there exists~\citep{docherty:thesis} a prime filter $F$ such that $\alginterp{P}_{\valuation_{\textsf{a}}} \in F$
and $\alginterp{Q}_{\valuation_{\textsf{a}}} \not\in F$. Let $\valuation_{\textsf{a}}'$ be as in the previous theorem, then we have
$F \models_{\valuation_{\textsf{a}}'} P$ and $F \not\models_{\valuation_{\textsf{a}}'} Q$ which contradicts the assumption that $P \models_{\valuation_{\textsf{a}}'} Q$.
\end{proof}

\section{Examples: Omitted Details}
\label{app:ex}

\subsection{Bound false positive rate in Bloom filter}
One detail we omitted is that, since the first line of the program \textsc{Bloom},
$bloom$ has been kept as a \emph{bit-array} throughout, i.e.,
all its entries are either 0 or 1. So it is easy to prove that
\[
	\hoare{\top}{\textsc{Bloom}}
		{\bigwedge_{\beta = 0}^{N}(bloom[\beta] = 0 \lor bloom[\beta] = 1 )}.
\]
Then, by the conjunction rule \textsc{Conj}, we have
\begin{align*}
  \hoare{\top}{\textsc{Bloom}}
	{\Pr \left[\sum_{\beta = 0}^N bloom[\beta] < K + T(\delta, N) \right] \geq 1 - \delta },
\end{align*}
where $K = \EE\left[\sum_{\beta = 0}^N bloom[\beta]\right]$.

In the following, we will abbreviate formulas that assert
$b$ is a bit-array where exactly $J$ of its first $N$ entries
are one,
\[
	\left( \sum_{\beta = 0}^N b[\beta] = J \right) \land
	\bigwedge_{\beta = 0}^{N}(b[\beta] = 0 \lor b[\beta] = 1),
\]
as $\bv{b}{J}{N}$.
Similarly, we will use $\bv{b}{< J}{N}$ to abbreviate
\[
	\left( \sum_{\beta = 0}^N b[\beta] < J \right) \land \bigwedge_{\beta =
	0}^{N}(b[\beta] = 0 \lor b[\beta] = 1).
\]
Now we restate our goal as
\begin{align*}
	\hoare{\bv{bloom}{< K}{N}}{\CheckMem}{\Pr[allhit] \leq(K/N)^H}.
\end{align*}

\CheckMem first initializes $h$ and $allhit$ deterministically to 1.
Then, using \textsc{RAssn} and \textsc{Frame}, we can show that
\[
	\vdash \hoare{\top}
	{\Assn{h}{0}; \Assn{allhit}{1}}
	{\left( \apEq{h}{0} \right) \sepand \left( \apEq{allhit}{1} \right)}.
\]
Using the~\eqref{ax:prob:one} axiom and the fact that $1 \leq (K/N)^0$ for any $K$ and $N$,
we can show $\models \left( \apEq{h}{0} \right) \sepand \left( \apEq{allhit}{1} \right)
\to \Pr[allhit] \leq (K/N)^h$.
Thus,
\[
	\vdash \hoare{\top}
	{\Assn{h}{0}; \Assn{allhit}{1}}
	{\Pr[allhit] \leq (K/N)^h}.
\]
Because the assignments $\Assn{h}{0}; \Assn{allhit}{1}$ do not modify the
Bloom filter array $bloom$, we can then apply the frame rule \textsc{Frame} to
derive
\begin{align}
	\label{eq:ex:checkmem:init}
	\vdash \hoare{\bv{bloom}{< K}{N}}
	{\Assn{h}{0};\Assn{allhit}{1}}
	{\bv{bloom}{< K}{N} \sepand \Pr[allhit] \leq (K/N)^h}.
\end{align}
We will abbreviate $\bv{bloom}{< K}{N} \sepand \Pr[allhit] \leq (K/N)^h$ as $\eta$.
Because $\sum_{\beta = 0}^N b[\beta]$ is an integer upper bounded by $N$,
\[
	\models_{\Mem} \eta \to \bigvee_{0 \leq J < K} \eta_J,
\]
where $\eta_J$ abbreviates $J < K \land \left(
	\bv{bloom}{J}{N} \sepand \Pr[allhit] \leq (K/N)^h \right)$.

We will then prove that for each $J$, the formula $\eta_J$
is a loop invariant of \CheckMem's loop body.
The loop body first uniformly sample an element from $[N]$,
so by \textsc{RSamp}${}^*$,
\begin{align}
	\eta_J \sepand \apUnif{bin}{[N]} .
	\label{eq:ex:checkmem:loop1}
\end{align}
Together with the axiom $\models P \sepand Q \to P$,~\Cref{eq:ex:checkmem:loop1}
implies
\[
	J < K \land \left(\bv{bloom}{J}{N} \sepand \left( \Pr[allhit] \leq
	(\frac{K}{N})^h \right) \sepand \apUnif{bin}{[N]} \right).
\]
Then, $hit$ gets assigned to $bloom[bin]$, so by \text{RAssn}, we have
\[
	\hoare{ \bv{bloom}{J}{N} \sepand \apUnif{bin}{[N]}
	}{\Assn{hit}{bloom[\beta]}}{ \left( \bv{bloom}{J}{N} \sepand \apUnif{bin}{[N]}
	\right) \land \apEq{hit}{bloom[bin]}}.
\]
Since the array $bloom$ only contains zero-one entries, when the sum of its entries is $J$,
an entry $bloom[bin]$ drawn uniformly at random has probability $\frac{J}{N}$ to be 1.
If the entry is in addition chosen independently from values in $bloom$,
then the bit $bloom[bin]$ is distributed independent from the distribution of $bloom$.
The~\eqref{ax:unif:samp} axiom encodes this fact:
\[
	\models \left( \left( \bv{b}{J}{N} \sepand  \apUnif{x}{[N]} \right)
	\land \apEq{hit}{b[x]} \right)
	\to \apBern{hit}{\frac{J}{N}}\sepand \bv{b}{J}{N} .
\]
Thus, we have
\[
	\hoare{ \bv{bloom}{J}{N} \sepand \apUnif{bin}{[N]}
	}{\Assn{hit}{bloom[\beta]}}{\apBern{hit}{\frac{J}{N}} \sepand
	\bv{bloom}{J}{N} }.
\]
Because $\Assn{hit}{bloom[bin]}$ does not modify $allhit$, we
can apply the frame rule for $\sepand$ \textsc{Frame} and get
\begin{align*}
	&	\left\{ \bv{bloom}{J}{N} \sepand \apUnif{bin}{[N]} \sepand \Pr[allhit] \leq \left(\frac{K}{N}\right)^h \right\}\\
& \quad\quad\quad\quad\quad\quad \Assn{hit}{bloom[\beta]}\\
& \left\{ \apBern{hit}{\frac{J}{N}} \sepand \bv{bloom}{J}{N} \sepand  \Pr[allhit] \leq \left(\frac{K}{N}\right)^h \right\}.
\end{align*}
Next, with the assignment $\Assn{allhit}{hit \mathop{\&\&} allhit}$,
by applying the axioms~\eqref{ax:ind:prob},~\eqref{ax:equal:prob}
and the \textsc{RAssn} rule,
we get:
\begin{align*}
	&\left\{ \apBern{hit}{\frac{J}{N}} \sepand \bv{bloom}{J}{N} \sepand  \Pr[allhit] \leq \left(\frac{K}{N}\right)^h\right\} \\
	& \quad\quad\quad\quad\quad \Assn{allhit}{hit \mathop{\&\&} allhit}\\
	&\left\{ \left(\Pr[allhit] \leq \frac{J}{N} \cdot \left(\frac{K}{N}\right)^h \right) \sepand \bv{bloom}{J}{N}\right\}
\end{align*}
We can then apply the rule of constancy \textsc{Const} and get
\begin{align*}
	&	\left\{ J < K \land \left( \bv{bloom}{J}{N} \sepand \apUnif{bin}{[N]} \sepand \Pr[allhit] \leq \left(\frac{K}{N}\right)^h \right)\right\}\\
	&  \quad\quad\quad\quad \Assn{hit}{bloom[\beta]}; \Assn{allhit}{hit \mathop{\&\&} allhit}\\
	&\left\{ J < K \land \left( \left(\Pr[allhit] \leq \frac{J}{N} \cdot \left(\frac{K}{N}\right)^h \right) \sepand \bv{bloom}{J}{N}\right) \right\}
\end{align*}
When we have $J <  K$, then $ (K/N)^h \cdot \frac{J}{N}  \leq  (K/N)^{h+1}$, so the post condition
implies
\[
	J < K \land \left( \left(\Pr[allhit] \leq  \left(\frac{K}{N}\right)^{h + 1} \right) \sepand \bv{bloom}{J}{N}\right)
\]
The last step in the loop body is the assignment $\Assn{h}{h + 1}$.
By the deterministic assignment rule \textsc{DAssn},
we can establish the post condition $\eta_J$ afterwards:
\[
	J < K \land \left(\bv{bloom}{J}{N} \sepand \Pr[allhit] \leq \left(\frac{K}{N}\right)^h \right).
\]
Thus, we have $\hoare{\eta_J}{\textit{loop body}}{\eta_J}$

By \textsc{Loop} rule, we can establish $\hoare{\eta_J}{loop}{\eta_J \land h \geq H}$.
Since $\eta_J \land h \geq H$ implies $\Pr[allhit] \leq (K/N)^H $, we then have
\[
	\hoare{\eta_J}{loop}{\Pr[allhit] \leq \left(\frac{K}{N}\right)^H}.
	\]
	Because $\Pr[allhit] \leq (K/N)^H$ is closed under mixtures,
and $\eta$ is closed under conditioning, we can then apply \textsc{RCase} to prove that
\begin{align}
	\label{eq:ex:checkmem:loop}
	\hoare{\eta}{loop}{ \Pr[allhit] \leq \left(\frac{K}{N}\right)^H}.
\end{align}
Using the \textsc{Seqn} rule to combine the proved judgments for \CheckMem's
initialization~\eqref{eq:ex:checkmem:init} and loop~\eqref{eq:ex:checkmem:loop}, we derive
\begin{align*}
	\hoare{\sum_{\beta = 0}^N bloom[\beta] < K}{\CheckMem}{\Pr[allhit] \leq \left(\frac{K}{N}\right)^H}.
\end{align*}

\subsection{Permutation Hashing}

We sketch how to replicate the informal reasoning in \programlogic. For the main
loop, we apply the rule \textsc{Loop} with the following loop invariant:
\[
  \bigwedge_{\alpha = 0}^n \apEq{hitZ[\alpha]}{[mod(g[\alpha], B) = Z]}
  \land \apPerm{g}{[B \cdot K]}
  \land ct = \sum_{\alpha = 0}^n hitZ[\alpha]
  \land (\neg (n < N) \to \apEq{n}{N})
\]
The loop invariant is preserved by the body of the loop, using the assignment
rule (\textsc{RAssn}) and the rule of constancy (\textsc{Const}). Thus we can
show the following judgment:
\[
  \hoare{\apEq{ct}{0} \land \apEq{n}{0}}{loop}{%
    \bigwedge_{\alpha = 0}^N \apEq{hitZ[\alpha]}{[mod(g[\alpha], B) = Z]}
    \land \apPerm{g}{[B \cdot K]}
  \land \apEq{ct}{\sum_{\alpha = 0}^N hitZ[\alpha]}}
\]
Applying \eqref{ax:perm-map}, the post-condition implies:
\[
  \bigwedge_{\alpha = 0}^N \apEq{hitZ[\alpha]}{[mod(g[\alpha], B) = Z]}
    \land \bigsep_{\alpha = 0}^N \apDist{hitZ[\alpha]}
    \land \apPerm{g}{[B \cdot K]}
    \land \apEq{ct}{\sum_{\alpha = 0}^N hitZ[\alpha]}
\]
Applying basic axioms about expected value and the permutation distribution
(\eqref{ax:perm:marg} \eqref{ax:prob:unif} \eqref{ax:biject:unif}), we have:
\[
  \bigsep_{\alpha = 0}^N \apDist{hitZ[\alpha]}
  \land \apEq{ct}{\sum_{\alpha = 0}^N hitZ[\alpha]}
  \land \apEq{\EE[ct]}{N / B}
\]
And we can apply the negative-association Chernoff bound \eqref{ax:na-chernoff}
to conclude:
\[
  \hoare{\top}
  {\textsc{PermHash}}
  {\Pr [ |ct - N / B| > T(\beta, N) ] < \beta}
\]
This conclusion corresponds to Proposition A.2 in
\citet{DBLP:journals/pvldb/DingK11} algorithm for fast set
intersection.\footnote{%
  \citet{DBLP:journals/pvldb/DingK11} apply a variant of the Chernoff bound to
  obtain a multiplicative, rather than an additive, error guarantee. We present
  the additive version since the bound is a bit simpler, but there is no
  difficulty to handling the multiplicative version in our framework.}

\subsection{Fully-Dynamic Dictionary}
We outline the main steps in the formal proof; the most interesting step is the
last one, where we use negative association, but all steps can be handled in our
framework.

We will refer to the two outer-most loops as (1) and (2), the next two
outer-most loops as (1.1) and (2.1), and the inner-most loop as (1.1.1).

\paragraph*{Computing $\EE[binCt[c][p]]$.}
For loop (1), we apply \textsc{Loop} with the following loop invariant:
\[
  \bigwedge_{\gamma = 0}^C \bigwedge_{\beta = 0}^P
  \apEq{\EE[binCt[\gamma][\beta]]}{n / (P \cdot C)}
  \land ( \neg (n < N) \to \apEq{n}{N} ) .
\]
To show that this invariant is preserved by the loop, by two applications of
\textsc{RSamp*} the following holds after the sampling commands:
\[
  \apOnehot{pocket[n]}{[P]} \indand \apOnehot{crate[n]}{[C]}
  \land \apEq{bin[n]}{crate[n]^\top \cdot pocket[n]} .
\]
Using an axiom about independence and products of one-hot vectors
\eqref{ax:ind:prod:oh}, this implies:
\[
  \apOnehot{bin[n]}{[C] \times [P]} .
\]
Using an axiom about the one-hot encoding~\eqref{ax:oh:marg}:
\[
  \apEq{\mathbb{E}[bin[\alpha][\gamma][\beta]]}{1/(P \cdot C)}
\]
for every $\alpha$, $\gamma$, and $\beta$. Standard loop invariants for loop
(1.1) and (1.1.1) show that:
\[
  \apEq{binCt[c][p]}{\sum_{\alpha = 0}^n bin[\alpha][c][p]} ,
\]
and linearity of expectation establishes the invariant condition for loop (1).
The invariant holds at the start of loop (1) since $binCt$ is zero-initialized,
and it also holds at the end of loop (1). Since $binCt$ is not modified further,
the expectation bound also holds at the end of the program (\textsc{Const}).

\paragraph*{Bounding $\Pr[binCt[c][p] > T_{bin}]$.}
For loop (1), we apply \textsc{Loop} with the following loop invariant:
\[
  \left( \bigind_{\alpha = 0}^n \apDist{bin[\alpha]} \right)
  \land
  \bigwedge_{\gamma = 0}^C \bigwedge_{\beta = 0}^P
  \apEq{binCt[\gamma][\beta]}{\sum_{\alpha = 0}^N bin[\alpha][\gamma][\beta]}
  \land ( \neg (n < N) \to \apEq{n}{N} ) .
\]
The first conjunction is an invariant, by applying the sampling rule
\textsc{Samp*} and the independence frame rule \textsc{Frame} from PSL. The rest
of the invariant is preserved, following standard invariants for loops (1.1) and
(1.1.1). By projection \eqref{ax:ind:map}, at the end of loop (1) we can
conclude:
\[
  \bigwedge_{\gamma = 0}^C \bigwedge_{\beta = 0}^P
  \left( \bigind_{\alpha = 0}^N \apDist{bin[\alpha][\gamma][\beta]} \right)
  \land \apEq{binCt[\gamma][\beta]}{\sum_{\alpha = 0}^N bin[\alpha][\gamma][\beta]} .
\]
Thus a (standard) Chernoff bound gives:
\[
  \bigwedge_{\gamma = 0}^C \bigwedge_{\beta = 0}^P
  \Pr [ binCt[\gamma][\beta] > \EE[binCt[\gamma][\beta]] + T(\rho_{bin}, N) ] \leq \rho_{bin} .
\]
where $\EE[binCt[\gamma][\beta]]$ is $N/(P \cdot C)$ by the previous step.
Again, property holds until the end of the program since $binCt$ is not modified
further (\textsc{Const}).

\paragraph*{Bounding $\EE[overCt[c]]$.}

Using standard loop invariants, at the end of loop (2) we have:
\[
  \bigwedge_{\gamma = 0}^C
  \apEq{overCt[\gamma]}{\sum_{\beta = 0}^P \sum_{\delta = 0}^C bin[\delta][\beta]}
  \land \bigwedge_{\beta = 0}^P \apEq{over[\gamma][\beta]}{[binCt[\gamma][\beta] > T_{bin}]} .
\]
Using linearity of expectation and the fact that $over[\gamma][\beta]$ is either
zero or one, we have:
\[
  \EE[overCt[\gamma]]
	\sim \sum_{\beta = 0}^P \mathbb{E}[over[\gamma][\beta]]
	\sim \sum_{\beta = 0}^P \Pr[over[\gamma][\beta] = 1]
  \sim \sum_{\beta = 0}^P \Pr[binCt[\gamma][\beta] > T_{bin}]
  \leq P \cdot \rho_{bin}
\]
since we have bound the probability in the previous step.

\paragraph*{Bounding $\Pr[overCt[c] > T_{over}]$.}

We want the following loop invariant for (1):
\[
  \bigwedge_{\gamma = 0}^C \bigwedge_{\beta = 0}^P \bigneg_{\beta = 0}^P \apDist{binCt[\gamma][\beta]}
  \land
  ( \neg(n < N) \to \apEq{n}{N} ) .
\]
We want the following loop invariant for (1.1):
\[
  \bigwedge_{\gamma = 0}^C
  \bigneg_{\beta = 0}^P \apDist{binCt[\gamma][\beta]}
	\negand
  \bigneg_{\beta = c}^P \apDist{bin[n][\gamma][\beta]}
  \land
  ( \neg(p < P) \to \apEq{p}{P} ) .
\]
And the following loop invariant for (1.1.1):
\[
  \bigwedge_{\gamma = 0}^C
  \bigneg_{\beta = 0}^P \apDist{binCt[\gamma][\beta]}
	\negand
  \bigneg_{\beta = p + [c > \gamma]}^P \apDist{bin[n][\gamma][\beta]} .
\]
We show the invariant post-conditions for a fixed $\gamma$; the big conjunction
then follows by applying \textsc{Conj}. Working from inside-to-outside, we start
with loop (1.1.1). To establish the invariant condition, the critical case is $c
= \gamma$. We can pull out:
\[
  \bigneg_{\beta = 0}^p \apDist{binCt[\gamma][\beta]}
	\negand
  \bigneg_{\beta = p + 1}^P \apDist{binCt[\gamma][\beta]}
	\negand
  \bigneg_{\beta = p + 1}^P \apDist{bin[n][\gamma][\beta]}
	\negand
	\underbrace{
    \apDist{binCt[\gamma][p]}
		\negand
    \apDist{bin[n][\gamma][p]}
	}_{\Phi}
\]
Now, we can use the assignment rule to show:
\[
  \hoare{\Phi}{\Assn{upd}{binCt[c][p] + bin[n][c][p]}}{\apEq{upd}{binCt[c][p] + bin[n][c][p]}}
\]
Since addition is a monotone function, the NA frame rule (\textsc{NegFrame})
gives:
\[
  \bigneg_{\beta = 0}^p \apDist{binCt[\gamma][\beta]}
	\negand
  \bigneg_{\beta = p + 1}^P \apDist{binCt[\gamma][\beta]}
	\negand
  \bigneg_{\beta = p + 1}^P \apDist{bin[n][\gamma][\beta]}
	\negand
  \apDist{upd}
\]
after the assignment to $upd$. After the assignment to $bin[c][p]$, we can fold:
\[
  \bigneg_{\beta = 0}^P \apDist{binCt[\gamma][\beta]}
	\negand
  \bigneg_{\beta = p + 1}^P \apDist{bin[n][\gamma][\beta]}
\]
to establish the invariant for loop (1.1.1).

To establish the invariant for loop (1.1), when the inner-most loop (1.1.1)
terminates we have $c > \gamma$, and so we have:
\[
  \bigneg_{\beta = 0}^P \apDist{binCt[\gamma][\beta]}
	\negand
  \bigneg_{\beta = p + 1}^P \apDist{bin[n][\gamma][\beta]} .
\]

To establish the invariant for loop (1), note that the invariant for loop (1.1)
holds on loop entry since $z$ is zero-initialized \eqref{ax:det:ind}. And the
loop invariant for loop (1) is established when loop (1.1) exits, when $p = P$.

Next, we tackle loop (2). We take the invariant:
\[
  \bigwedge_{\gamma = 0}^C
  \bigneg_{\beta = 0}^p \apDist{over[\gamma][\beta]}
	\negand
  \bigneg_{\beta = p}^P \apDist{binCt[\gamma][\beta]} .
\]
For the inner loop (2.1), we take the invariant:
\[
  \bigwedge_{\gamma = 0}^C
  \bigneg_{\beta = 0}^{p + [c > \gamma]} \apDist{over[\gamma][\beta]}
  \negand \bigneg_{\beta = p + [c > \gamma]}^P \apDist{binCt[\gamma][\beta]} .
\]
Again, we show the invariant post-conditions for a fixed $\gamma$.  For the
critical iteration $c = \gamma$, we again isolate $binCt[\gamma][p]$, observe
that addition is monotone and the function $[binCt[c][p] > T_{bin}]$ is monotone in
$binCt[\gamma][p]$, and apply the NA frame rule (\textsc{NegFrame}).

Finally at the end of the program, we can show:
\[
  \bigneg_{\beta = 0}^P \apDist{over[\gamma][\beta]}
\]
along with the regular invariant
\[
  \apEq{overCt[\gamma]}{\sum_{\beta = 0}^P over[\gamma][\beta]} .
\]
We can then apply the negative-dependence Chernoff bound \eqref{ax:na-chernoff}:
\[
  \Pr [ overCt[\gamma] > \EE[overCt[\gamma]] + T(\rho_{over}, P) ] \leq \rho_{over} .
\]
Using the expectation bound from the previous step and putting everything
together, we conclude:
\[
  \hoare{\top}
  {\textsc{FDDict}}
  {\bigwedge_{\gamma = 0}^C \Pr [ overCt[\gamma] > P \cdot \rho_{bin} + T(\rho_{over}, P) ] \leq \rho_{over}} ,
\]
thus showing a high-probability upper-bound on the number of overfull pocket
dictionaries within each crate.

\subsection{Repeated Balls-into-Bins}

We will refer to the loops in \Cref{fig:ex:repeat:bib} using the same scheme we
used before: the outer-most loop is loop (1), the three next-outer-most loops
are loops (1.1), (1.2), and (1.3), and the inner-most loop is loop (1.2.1).
Starting from the outside, we take the following invariant for loop~(1):
\[
  \Pr\left[ \bigvee_{\beta = 0}^r (emptyCt[\beta] > T_{empty}) \right] \leq r \cdot \rho_{empty}
  \land \apEq{\sum_{\alpha = 0}^N ct[\alpha]}{N}
\]
Showing the invariant condition requires some work. First, note that:
\[
  \models_\Mem \apEq{\sum_{\alpha = 0}^N ct[\alpha]}{N}
  \to \bigvee_{\sigma : [N] \to [N]} \bigwedge_{\alpha = 0}^N
  \apEq{ct[\alpha]}{|\sigma^{-1}(\alpha)|}
\]
where $\sigma : [N] \to [N]$ ranges over all assignments of $N$ balls to $N$
bins, and where we write $\models_\Mem$ to denote that the formula is valid in
all \emph{memories}, rather than distributions over memories. We write
$\tau(\alpha) = |\sigma^{-1}(\alpha)|$ for the number of balls in bin $\alpha$.
We will show:
\[
  \hoare{\bigwedge_{\alpha = 0}^N \apEq{ct[\alpha]}{\tau(\alpha)} }
  {body}
  {\Pr\left[ emptyCt[r] < T_{empty} \right] \leq \rho_{empty}}
\]
where $body$ is the body of loop (1). For loop (1.1), it is straightforward to
show the invariant using the loop rule \textsc{Loop}:
\[
  \bigwedge_{\alpha = n}^N (\apEq{ct[\alpha]}{\tau(\alpha)})
  \land \bigwedge_{\alpha = 0}^n (\apEq{ct[\alpha]}{\tau(\alpha)- [\tau(\alpha) > 0]})
  \land \apEq{rem}{\sum_{\alpha = 0}^n [\sigma(\alpha) > 0]}
  \land ( \neg(n < N) \to \apEq{n}{N} )
\]
At the exit of loop (1.1), we have:
\[
  \bigneg_{\alpha = 0}^N \apDist{ct[\alpha]}
\]
since counts are all equal to expressions of logical variables, so conditioning
on the logical variables, they are all deterministic; we take this formula to be
the invariant for loop (1.2). Note that the loop guard is not deterministic,
since the value of $rem$ is randomized. However, \emph{under our conditioning},
$rem$ is deterministic under our conditioning since it is fully determined by
the initial counts (i.e., it is the number of buckets that are initially
non-empty). Hence, we may apply the loop rule \textsc{Loop}, treating the loop
guard as deterministic. This is the power of reasoning under conditioning.

Now to establish the invariant for loop (1.2), we reason much as in the previous
examples. The sampling rule \textsc{Samp*} gives:
\[
  \bigneg_{\alpha = 0}^N \apDist{ct[\alpha]} \indand \apDist{bin[j]}
\]
By negative association for one-hot encoding \eqref{ax:oh-pna}:
\[
  \bigneg_{\alpha = 0}^N \apDist{ct[\alpha]} \indand \bigneg_{\alpha = 0}^N \apDist{bin[j][\alpha]}
\]
This implies:
\[
  \bigneg_{\alpha = 0}^N \apDist{ct[\alpha]} \negand \bigneg_{\alpha = 0}^N \apDist{bin[j][\alpha]}
\]
For the inner-most loop (1.2.1), we apply the same technique as for loop (1.2).
Since loop (1.2) has a randomized guard, $k$ is a random variable and loop
(1.2.1) also has a randomized guard. However, under the conditioning, we may
assume that $k$ is deterministic and apply \textsc{Loop} on loop (1.2.1) with
the following invariant:
\[
  \bigneg_{\alpha = 0}^k \apDist{ct[\alpha]}
  \negand \left(\apDist{ct[k]} \negand \bigneg_{\alpha = k}^N \apDist{bin[j][\alpha]}\right)
  \negand \bigneg_{\alpha = k + 1}^N \apDist{ct[\alpha]}
  \land (\neg (k < N) \to \apEq{k}{N})
\]
Like in earlier examples, we can establish this invariant using the negative
dependence frame rule since $ct[n] + bin[j][n]$ is monotone. Thus at exit of
loop (1.2.1), we have:
\[
  \bigneg_{\alpha = 0}^N \apDist{ct[\alpha]}
\]
Next, three applications of the assignment rule \textsc{RAssn} give:
\[
  \bigneg_{\alpha = 0}^N \apDist{ct[\alpha]}
  \land \apEq{n}{0}
  \land \apEq{emptyCt[r]}{0}
  \land \apEq{empty}{isZero(ct)}
\]
The function $isZero(v)$ takes a vector of numbers $v$, and returns a vector
where each index $i$ $1$ if $v[i]$ is zero, else it holds $0$. This is an
antitone function: it is non-increasing in its argument. Thus, the monotone
mapping axiom \eqref{ax:mono-map} gives:
\[
  \bigneg_{\alpha = 0}^N \apDist{empty[\alpha]}
\]
Then, a standard loop invariant for loop (1.3) gives:
\[
  \bigneg_{\alpha = 0}^N \apDist{empty[\alpha]}
\land \apEq{emptyCt[r]}{\sum_{\alpha = 0}^N empty[\alpha]}
\]
Now, we are in position to apply the negative association Chernoff bound
\eqref{ax:na-chernoff}, giving the judgment:
\[
  \hoare{\bigwedge_{\alpha = 0}^N \apEq{ct[\alpha]}{\tau(\alpha)}}
  {body}
  { \Pr [ emptyCt[r] < \EE[emptyCt[r]] - T(\rho_{empty}, N) ] \leq \rho_{empty} }
\]
where $body$ is the body of loop (1). However, we are not yet done. We want
to combine these judgments---one for each map $\sigma : [N] \to [N]$---using the
randomized case analysis rule \textsc{RCase}
%
%
We can take the trivial pre-condition $\phi = \top$, and the case condition:
\[
  \eta = \bigwedge_{\alpha = 0}^N \apEq{ct[\alpha]}{\tau(\alpha)}
\]
Since $\eta$ asserts that the equality holds with probability $1$, it is closed
under conditioning. However, our post-condition has a problem: it mentions the
expected value $\EE[emptyCt[r]]$, which may not be preserved under mixtures, so
the entire assertion is not CM. However, translating an argument by \citet[Lemma
2]{DBLP:journals/dc/BecchettiCNPP19} into our logic gives:
\[
  \hoare{\bigwedge_{\alpha = 0}^N \apEq{ct[\alpha]}{\tau(\alpha)}}
  {body}
  { \EE[emptyCt[r]] \geq N/15 }
\]
assuming that $N \geq 2$. The argument makes use of basic properties of expected
values and the exponential function; we omit the details. Thus, we have:
\[
  \hoare{\bigwedge_{\alpha = 0}^N \apEq{ct[\alpha]}{\tau(\alpha)}}
  {body}
  { \Pr [ emptyCt[r] < N/15 - T(\rho_{empty}, N) ] \leq \rho_{empty} }
\]
and the post-condition is now a CM assertion. Applying \textsc{RCase}, we have:
\[
  \hoare{\sum_{\alpha = 0}^N \apEq{ct[\alpha]}{N}}
  {body}
  { \Pr [ emptyCt[r] < N/15 - T(\rho_{empty}, N) ] \leq \rho_{empty} }
\]
Recalling that we wanted the following invariant for loop (1):
\[
  \Pr\left[ \bigvee_{\beta = 0}^r (emptyCt[\beta] < T_{empty}) \right] \leq r \cdot \rho_{empty}
  \land \sum_{\alpha = 0}^N \apEq{ct[\alpha]}{N}
\]
we can use the rule of constancy \textsc{Const} and the assignment rule
\textsc{RAssn} to preserve the first conjunct to show:
\[
  \Pr\left[ \bigvee_{\beta = 0}^{r-1} (emptyCt[\beta] < T_{empty}) \right] \leq (r - 1) \cdot \rho_{empty}
\]
at the end of the body of loop (1). Combined with the probability bound for
$emptyCt[r]$, an application of the union bound \eqref{ax:union:bd} establishes
the invariant for loop (1). Putting everything together, we have:
\[
  \hoare{N \geq 2 \land \sum_{\alpha = 0}^N \apEq{ct[\alpha]}{N}}
  {\textsc{RepeatBIB}}
  {
    \Pr\left[ \bigvee_{\beta = 0}^R (emptyCt[\beta] < N/15 - T(\rho_{empty}, N)) \right] \leq R \cdot \rho_{empty}
  }
\]
analogous to \citet[Lemma 1 and 2]{DBLP:journals/dc/BecchettiCNPP19}.

\subsection{Axioms for Examples}

For completeness, we present the probability-related axioms that we need for the
examples. For simplicity we present the axioms in binary form, though most
extend directly to big operations.
\begin{itemize}
  \item Linearity of expectation. Let $e, f$ be \emph{bounded} expressions.
    \begin{equation} \label{ax:lin:exp}
      \models \apEq{\EE[\alpha \cdot e + \beta \cdot f]}{\alpha \cdot \EE[e] + \beta \cdot \EE[f]}
      \tag{LinExp}
    \end{equation}
  \item Union bound. Let $ev_1, ev_2 \in \Event$,
    \begin{equation} \label{ax:union:bd}
      \models \Pr [ ev_1 \lor ev_2 ] \leq \Pr [ ev_1 ] + \ Pr [ ev_2 ]
      \tag{UnionBd}
    \end{equation}
  \item Permutation marginal. Let $x$ be an array variable, and let $S$ be a
    finite set.
    \begin{equation} \label{ax:perm:marg}
      \models \apPerm{x}{S} \to \apUnif{x[\alpha]}{S}
      \tag{PermMarg}
    \end{equation}
  \item Expectation Indicator. Let $e$ be a 0/1 valued expression,
    \begin{equation} \label{ax:expect:ind}
      \models \apEq{\EE[ e ]}{\Pr[e = 1]}
      \tag{ExpectInd}
    \end{equation}
	\item Bernoulli variables probabilities. Let $e$ be an expression,
    \begin{equation} \label{ax:bernoulli:p}
      \models \apBern{e}{p} \to \Pr[e = 1] = p
      \tag{BernProb}
    \end{equation}
  \item Probability of uniform. Let $S$ be a finite set.
    \begin{equation} \label{ax:prob:unif}
      \models \apEq{\Pr[ \apUnif{x}{S} = \alpha ]}{1/|S|}
      \tag{ProbUnif}
    \end{equation}
  \item Bijection uniform. Let $S$ be a finite set, and let $f : S
    \to S$ be a bijection.
    \begin{equation} \label{ax:biject:unif}
      \models \apUnif{x}{S} \to \apUnif{f(x)}{S}
      \tag{BijectUnif}
    \end{equation}
  \item One-hot marginal. Let $x$ be an array variable.
    \begin{equation} \label{ax:oh:marg}
      \models \apOnehot{x}{S} \to \apUnif{x[\alpha]}{S}
      \tag{OHMarg}
    \end{equation}
  \item Independent product one-hot.
    \begin{equation} \label{ax:ind:prod:oh}
      \models \apOnehot{x}{[M]} \indand \apOnehot{y}{[N]} \to \apOnehot{x^\top \cdot y}{[M] \times [N]}
      \tag{IndProdOH}
    \end{equation}
  \item Independent map. Let $x$ be an array variable of length $N$.
    \begin{equation} \label{ax:ind:map}
      \models \bigind_{\alpha = 0}^N \apDist{x[\alpha]}
      \to \bigind_{\alpha = 0}^N \apDist{f(x[\alpha])}
      \tag{IndMap}
    \end{equation}
  \item Deterministic independent. Let $x$ be a variable.
    \begin{equation} \label{ax:det:ind}
      \models \apDetm{x} \to \apDist{x} \indand \apDist{e}
      \tag{DetInd}
    \end{equation}
	\item Events happen only if they have probability one. Let $ev \in \Event$,
	\begin{equation} \label{ax:prob:one}
		\models \apEvent{ev} \to \Pr(ev) = 1
		\tag{ProbOne}
	\end{equation}
	\item Uniform sampling from a population.
  We represent a population as a bit-vector, where each entry is an individual
  and 1 indicates they have some feature and 0 indicates not. Then, if we
  uniformly sample from the population,  the probability of getting a one is
  equal to population-level ratio of ones, regardless how they are distributed
  in the population. Let $N \geq J$ be constants or logical variables, $b$ be an
  array variable of length $N$, and $x, hit$ be variables:
	\begin{equation} \label{ax:unif:samp}
		\models \left( \left( \bv{b}{J}{N} \sepand  \apUnif{x}{[N]} \right) \land
		\apEq{hit}{b[x]} \right) \to \apBern{\frac{J}{N}}{hit}\sepand\left(
		\sum_{\beta = 0}^N b[\beta] = J \right).
		\tag{UniformSamp}
	\end{equation}
	\item Independent product probabilities. Let $ev_1, ev_2 \in \Event$ ,
	$J, K$ be two real numbers,
	\begin{equation} \label{ax:ind:prob}
		\models \Pr[ev_1] \leq J  \sepand \Pr[ev_2] \leq K \to \Pr[ev_1 \land ev_2]\leq J \cdot K.
		\tag{IndepProb}
	\end{equation}
	\item Equal probabilities.
	Let $b_1, b_2$ be two boolean expressions. Recall that $b_1, b_2 \in \Event$ too.
	\begin{equation} \label{ax:equal:prob}
    \models \apEq{b_1}{b_2} \to \Pr[b_1]  = \Pr[b_2]
		\tag{EqualProb}
	\end{equation}
\end{itemize}
\fi

\end{document}